\documentclass{article}
\usepackage[utf8]{inputenc}

\usepackage{geometry}
\usepackage{amsmath,amsthm,amsfonts,bm,mathrsfs,amssymb}
\usepackage{booktabs}
\usepackage{url}
\usepackage{hyperref}

\theoremstyle{definition}
\newtheorem{example}{Example}[section]
\newtheorem{definition}{Definition}[section]
\newtheorem{theorem}{Theorem}[section]

\DeclareMathOperator{\EX}{\mathbb{E}}
\DeclareMathOperator{\Var}{\mathrm{Var}}

\newcommand{\bias}[1]{\mathrm{Bias}_R[#1]}
\newcommand{\var}[1]{\mathrm{Var}[#1]}
\newcommand{\cov}[1]{\mathrm{Cov}[#1]}
\newcommand{\corr}[1]{\mathrm{Corr}[#1]}

\newcommand{\maker}{{m}}
\newcommand{\taker}{{t}}
\newcommand{\real}{{r}} 

\newcommand{\makerVal}{{m_i}}
\newcommand{\takerVal}{{t_i}}
\newcommand{\realVal}{{r_i}}

\newcommand{\makerRV}{{M}}
\newcommand{\takerRV}{{T}}
\newcommand{\realRV}{{R}}   

\newcommand{\realFunc}{{\sf r}}
\newcommand{\makerFunc}{{\sf m}}
\newcommand{\takerFunc}{{\sf t}}

\newcommand{\return}{{\rho}}

\newcommand{\un}{{\sf unif}}
\newcommand{\opt}{{\sf sharpe}}

\usepackage{longtable}
\usepackage{pdflscape} 
\usepackage{multirow}

\usepackage{graphicx}
\usepackage[dvipsnames]{xcolor}

\usepackage{tikz}
\usetikzlibrary{calc, positioning, shapes, through, intersections}
\usepackage{tikzscale}
\usepackage{ifthen}

\title{Beating the market with a bad predictive model}

\author{
Ond\v{r}ej Hub\'{a}\v{c}ek 
\and
Gustav \v{S}\'{i}r 
\footnote{
corresponding author: {gustav.sourek@fel.cvut.cz}           
}
}

\date{%
    Department of Computer Science\\
    Faculty of Electrical Engineering\\
    Czech Technical University in Prague\\
}

\begin{document}

\maketitle

\begin{abstract}
It is a common misconception that in order to make consistent profits as a trader,
one needs to posses some extra information leading to an asset value prediction more accurate than that reflected by the current market price.
While the idea makes intuitive sense and is also well substantiated by the widely popular Kelly criterion, we prove that it is generally possible to make systematic profits with a completely inferior price-predicting model. The key idea is to alter the training objective of the predictive models to explicitly decorrelate them from the market, enabling to exploit inconspicuous biases in market maker's pricing, and profit on the inherent advantage of the market taker.
We introduce the problem setting throughout the diverse domains of stock trading and sports betting to provide insights into the common underlying properties of profitable predictive models, their connections to standard portfolio optimization strategies, and the, commonly overlooked, advantage of the market taker. Consequently, we prove desirability of the decorrelation objective across common market distributions, translate the concept into a practical machine learning setting, and demonstrate its viability with real world market data.

\end{abstract}



\section{Introduction}

The attempt to predict the future is at the core of any successful trading strategy. Traders continuously try to predict future prices of assets, such as stocks or commodities, to secure profits on their future positions. Similarly, bettors try to predict discrete future outcomes of random events, such as elections or sports matches, to secure profits on their realizations. Having a high-quality future prediction, one can confidently estimate the expected returns associated with the potential decisions to allocate assets and select securities for successful investments and speculations. The quality of each prediction then naturally reflects the amount of information used for the estimation. In an efficient market, it is assumed that all available information is already reflected in the current market price, and it is thus impossible to make any consistent profits through trading with such predictions~\cite{malkiel1989efficient}.
It intuitively follows that in a partially inefficient market, where the market price reflects some but not all of the information, one needs the extra information, leading to predictions better than those reflected by the current market price, in order to profit.

Indeed, much of the existing literature on profiting in an inefficient market is based on the idea of well-informed investors, whose superior knowledge of the given domain enables to predict the future prices, or estimate outcome probabilities, more closely than the market, commonly quoting some aggregated estimation of the rest of the trading participants. This in turn enables them to secure profits using common investment strategies such as those of Markowitz~\cite{markowitz1952portfolio} and Kelly~\cite{kelly1956new}, the performance of which is tightly connected to such an information advantage.
While possessing more information leading to superior price estimation indeed enables to profit in this straightforward manner, we argue that in order to make positive returns, it is generally not necessary to estimate the true price better than the market. Instead, there is another and generally easier way to exploit market inefficiencies.

The core principle we utilize in this paper is the, commonly overlooked, inherent advantage of the market taker over the market maker, where the two price estimators are eventually penalized for their estimation errors in a very different manner. Consequently, in order to make profits, the trader does not need to estimate the true price of a mispriced asset better than the market and, instead of predictive pricing accuracy, should generally strive to optimize different performance measures. Specifically, we show that trader's profits derived from a price estimator can be systematically increased by decreasing its partial correlation with the market, i.e. by decorrelating its residual estimation error w.r.t. the market.
While this might seem similar in spirit to the common idea of profiting from using ``extra information'', we demonstrate that the trader's estimates can be easily based on inferior data sources and their information-theoretic value w.r.t. true price can be considerably lower than that reflected by the market. Consequently, the respective trader's predictive model can thus be considered as ``bad'' by the means of regular quality measures.



While the proposed concept might seem theoretical in nature, it has important implications on common investment practices based on predictive modelling and portfolio optimization. Particularly, the current systems typically assume the asset price prediction and the subsequent investment decision making as two independent tasks. Namely, at first a statistical model is optimized to predict the future price as closely as possible which, translated into an expected rate of return, forms an input to the subsequent investment strategy designed to maximize some profit-based utility. While this typical workflow has the advantage of decomposing the problem, enabling for a more explicit risk management, we argue that it is considerably suboptimal since the two optimization tasks are not aligned, i.e. the predictive model accuracy optimization is oblivious of the subsequent profit-focused optimization. Consequently, maximizing model price estimation quality will commonly lead to suboptimal final profits.
In contrast, with the proposed decorrelation objective, we target the profits more directly. This enables to generate consistent positive returns even with no information advantage and completely inferior price predicting models.

The paper is structured as follows.
We introduce the necessary background, ranging from market modelling to portfolio optimization, in Section~\ref{sec:back}. In Section~\ref{sec:insights}, we provide insights into essential conditions for model profitability, limitations of popular investment strategies, and the inherent market taker's advantage. Section~\ref{sec:decorrelating} then introduces the proposed concept of increasing profits through decorrelating from the market, its connection to the Kelly strategy, and translation into a machine learning objective. We then demonstrate the concept in a practical setting in Section~\ref{sec:experiments}. Finally, we review the related work in Section~\ref{sec:related} and conclude in Section~\ref{sec:conclusion}.




\section{Background}
\label{sec:back}

\begin{table}

\begin{tabular}{@{}llll@{}}
\toprule
\textbf{symbol}                                   & \textbf{description}                                                          & \textbf{stock trading}       & \textbf{two-way betting}        \\ \midrule
$a_j \in \mathcal{A}$                             & assets being traded on the market                                             & stocks, securities           & sport matches \\
$\{\alpha,\beta\};~\zeta \in \mathrm{Z}$                      & market sides, in two-way markets $|\mathrm{Z}| = 2$                                  & buy/sell (bid/ask)           & home/away (back/lay)          \\
$\omega_{(a_j,\tau)}^{\zeta};~ \omega_i \in {\Omega}$   & opportunities to trade assets at time $\tau$                                        & e.g. buy $a_j$ now            & e.g. bet $a_j$ home            \\

$\mathcal{D}_i \subset \mathcal{D}^*$             & relevant data available for opportunity $\omega_i$                       &                              &                               \\

$\realFunc : \mathcal{D}_i^* \to \mathbb{R}_+$     & fundamental value $r$ of an asset                                                 & true value           & true probability                  \\
$\makerFunc : \mathcal{D}_i^b \to \mathbb{R}_+$     & market maker $\maker$, or simply ``the market''                             & market maker                & bookmaker                    \\
$\takerFunc : \mathcal{D}_i^m \to \mathbb{R}_+$     & market taker $\taker$, or simply trader                                       & trader                &      bettor              \\

$\mathrm{P}_{\Omega}$                                      & distribution of market opportunities                                          &   all possible trades                           &      all possible bets                         \\

$\realRV: \omega_i \mapsto \real_i$                  & r.v. denoting fundamental asset value     & price $r_i \in \mathbb{R}_+$ & probability $r_i \in [0,1]$   \\
$\makerRV: \omega_i \mapsto \maker_i$                 & r.v. denoting market maker's pricing        & price $\maker_i \in \mathbb{R}_+$ & probability $\maker_i \in [0,1]$   \\
$\takerRV: \omega_i \mapsto \taker_i$                  & r.v. denoting market taker's pricing estimate  & price $\taker_i \in \mathbb{R}_+$ & probability $\taker_i \in [0,1]$   \\

$\mathrm{P}(\realRV,\makerRV,\takerRV)$             & distribution of the price estimates (values)                              &                              &                               \\


$\theta \in \Theta$     & model parameters                                       &                  &                      \\
$\sf{\Tilde{e}}: \mathcal{D}_i \to \mathrm{{\Tilde{E}_i}}$    & posterior distribution $\mathrm{{\Tilde{E}}}$ estimator $\sf{\Tilde{e}}$ (e.g. $~\Tilde{\takerFunc}$)                                &                              &                               \\

$\tau$     & time associated with a trade                                      &                  &                      \\

$\epsilon \in \mathbb{R}_+$                                        & market maker's spread (margin)                             & spread                       & margin                        \\
$\return_i \in \mathbb{R}$                                            & trading returns (profit)                                                 & rate of return               & return on investment          \\
$\sf{W} \in \mathbb{R}_+$                                          & wealth of the market taker                                                    & capital                      & bankroll                      \\
$\bm{f} \in \mathbb{R}^n$                                          & vector of the wealth allocations                        & allocation                   & wager                         \\ \bottomrule
\end{tabular}
\caption{Overview of the used notation with the correspondence between the stock and betting market settings.}
\label{tab:notation}
\end{table}




\subsection{Problem Setup}
\label{sec:background}
\label{sec:back_setup}

In this paper we assume a standard market setting where two or more parties participate in an exchange of some forms of assets. For the sake of this paper it is not important to distinguish between securities or other financial instruments, and we further refer to all these jointly as \textit{assets}.
While the analysis in this paper covers a wide variety of generic market settings, ranging from commodities to prediction markets, we will mostly consider two diverse enough examples of (i) stock and (ii) sports betting markets to keep the explanation grounded. Similarly, we do not restrict to any particular asset class, which can vary widely from financial derivatives to fantasy sports, however for clarity of explanation, we will consider a standard ``bid-ask'' setting for the stock, and a corresponding two-way\footnote{While we limit the explanation to two-way markets for clarity, the main concept proposed in this paper is directly applicable to $n$-way markets, too.} sports betting market, such as predicting the winner of a basketball game. Table~\ref{tab:notation} provides an overview of the joint notation used throughout the paper and across the two domains.


\paragraph{Opportunity}
In a two-way market, each asset $a \in \mathcal{A}$ tradeable at time $\tau$ corresponds to an opportunity to buy ($\alpha$) and sell ($\beta$) the asset, respectively. Given a certain time $\tau$, this covers bids and asks for a certain stock, as well as bets on the home ($\alpha$) and away ($\beta$) team win of a certain sport match.
We will further refer to all existing market positions to trade a certain asset $a$ at a certain time $\tau$ jointly as \textit{opportunities} $\omega^{\zeta}_{(a,t)} \in {\Omega}$.
We further treat each opportunity $\omega^{\zeta}_{(a,t)}$ as an independent trading unit denoted $\omega_i$ since, for the analysis in this paper, it is not necessary to distinguish the asset or time that each opportunity is associated with\footnote{i.e. two opportunities $\omega_i$ and $\omega_j$ over the same asset at two different times will be treated the same as two opportunities $\omega_i$ and $\omega_j$ over two different assets.}. Nevertheless, where necessary, we will distinguish the side of the market $\zeta \in \{\alpha,\beta\}$ an opportunity $\omega_i^\zeta$ is to be traded at.


\paragraph{Price}
Each opportunity $\omega_i$ is then associated with a certain {price}. In a stock market, it is simply the current price set up by the market maker to trade the asset, as used in the usual sense.
In the prediction markets, the price reflects the bookmaker's perceived probability of an outcome to occur. It is also commonly referred to as ``odds''\footnote{which we use more specifically later to refer to the app. inverse of the probability value (Section~\ref{sec:betting})} that determine the potential payouts received from a wager.
While the particular calculations naturally differ between monetary prices and odds (Section~\ref{sec:exp_return}), where it is not necessary to distinguish between the two, we will refer to these jointly as \textit{prices}. We note that the market price of a buying opportunity $\omega_i^\alpha$ may be different from the price to sell $\omega_i^\beta$ (Section~\ref{sec:maker_adv}), and similarly for the odds.

\paragraph{Fundamental Value}
We further assume that each opportunity $\omega_i$, i.e. an asset being traded at a given time, has some fundamental \textit{true value} $\realVal \in \mathbb{R}_+$, which can be expressed as a positive real number and by those means {compared} to some price estimate (Section~\ref{sec:estimators}) in the same unit of measurement. For instance in a stock market, a fair market price should reflect the true expected value of the underlying corporation, given the current information context. Similarly in the betting market, fair odds would reflect the true (inverse) probability of the associated outcome to happen. Note that the true value of an asset $a$ at time $\tau$ is always the same for both sides ($\alpha,\beta$) of the market.
While the concept of a fundamental value $\real$ can be seen as somewhat speculative, since it cannot be directly measured, we note that we merely require its theoretical existence for defining market efficiency (Section~\ref{sec:market_efficiency}) and the related concepts (Section~\ref{sec:exp_return})\footnote{We also commonly use the term true fundamental value interchangeably with the mean future price since, assuming the same information context, the latter can be expected to converge to the former in the partially efficient markets assumed (Section~\ref{sec:market_efficiency}).}.

\paragraph{Makers and Takers}
We commonly distinguish two types of the exchange participants as (i) market \textit{makers} ${\maker}$, also referred to as bookmakers, and (ii) market takers ${\taker}$, also referred to as \textit{traders}.
The market makers continuously quote both sides ($\alpha,\beta$) of the market at certain {price} levels $\makerVal$ resulting into trading opportunities $\omega_i \in \Omega$. By continuously \textit{providing} such quotes of asset prices, the makers bring constant liquidity to the market.
The market takers are then \textit{selecting} from the existing opportunities $\omega_i$ to issue specific buy and sell orders. We generally consider the problem analysed in this paper as a two-player game between a market maker ${\maker}$ and taker ${\taker}$, where each player possesses a certain pricing policy $\Omega \to \mathbb{R}_+$ over the market distribution of opportunities $\mathrm{P}_{\Omega}$ given by the game (world) environment. We further assume the two players to never switch their roles.


\paragraph{Beating the market}

We then aim to design a strategy for the role of the market taker, i.e. some generic trader $\taker$, to make positive profits against some particular market maker $\maker$. For example, such strategy would allow a bettor to make consistent long-term profits while wagering against a particular bookmaker. Note that this is a zero-sum setting, where the profits of the market taker are at the direct expense of the market maker, as the fundamental value is objectively the same for all the participants\footnote{We use this to remove trivial non-zero-sum settings, where all the participants can mutually profit from trading, from the scope. While we acknowledge that different traders may subjectively evaluate an opportunity $\omega_i$ differently, e.g. due to distinct preferences for risk tolerance (Section~\ref{sec:exp_return}), this would allow to trivially avoid the hard part of the introduced problem.}.
This is also the typical setting for a closed system of traders operating within the same time-frame, such as in intra-day trading, futures, options, or prediction markets. 
The important aspect here is that the only way to make positive profits in this setting is to ``beat the average trader'', quoting the market price, by some margin (Section~\ref{sec:maker_adv}). In this paper, we then utilize the phrase “beating the market” to refer to profiting in this strictly competitive setting.


\subsection{Market Efficiency}
\label{sec:market_efficiency}
In a (strongly) efficient market, the current market price $\maker \in \mathbb{R_+}$ of an asset $a$ reflects all existing information, making it impossible for any trader to make consistent profit by outsmarting the market\footnote{Note that this does not imply that one needs comparably more information to beat an inefficient market.} \cite{fama1970efficient}. Particularly, the price $\makerVal$ of each opportunity $\omega_i$ being traded needs to be an \textit{unbiased} estimate (Section~\ref{sec:estimators}) of its fundamental true value $\realVal$~\cite{pilbeam2018finance}. Note that this does not imply the price $\makerVal$ to be equal to the true value $\realVal$, but merely that the error of the estimation $\makerVal=\hat{\realVal}$ is fully due to an irreducible {variance} which is \textit{completely random} (Section~\ref{sec:price_estimators})~\cite{samuelson2016proof}. The inherent randomness of the error then ensures that no trader can consistently beat the market to secure systematic profits.

In real world liquid markets, profit-maximizing traders continuously search for under-priced bids or over-priced asks to secure their profits, pushing the market price to quickly converge to the true value in the process. Thanks to this self-regulating mechanism, the market inefficiencies tend to vanish quickly over time~\cite{Stekler2010,fama1998market}. We note that the idea of an efficient market is a rather theoretical one, since in a market that would become completely efficient, the traders would loose the incentive to search for inefficiencies, in turn making the market inefficient again. Any real world market can thus be hardly considered as perfectly efficient, nevertheless some measurable degree of efficiency, such as statistical unbiasedness of the mean prices (Section~\ref{sec:est_dist}), is typically present in liquid markets \cite{franck2010prediction}. In this paper, we will further consider the most realistic setting of a partially (in-)efficient market where the market price $\maker$ is a very good estimate of the true value $\real$, but not a perfect one.

\subsection{Market Maker's Advantage}
\label{sec:maker_adv}


Market makers are essential to trading by providing constant liquidity to both sides $(\alpha,\beta)$ of the market, for which they are typically favored by the exchange operator in terms of fees and commissions. However, the position of a market maker is a difficult one, since she needs to constantly price the assets as accurately as possible. An estimate too low can lead to exploitation on the buy side $\alpha$, and an estimate too high on the sell side $\beta$, respectively. While the true value $\real$ is an unknown real number, it would be in principle impossible to avoid exploitation in full.

To improve her position and secure profit, the market maker incorporates a so called \textit{spread} $\epsilon$ (margin) into her estimates $\makerRV$, causing the offered bid opportunity $\omega^\alpha_i$ price to differ from the ask $\omega^\beta_i$ price by some $2\epsilon$. The spread $\epsilon$ then works as a safety patch on the market maker's estimation error, preventing from one-sided exploitations by the market takers aiming at the discrepancy from the true value. A safe strategy for the market maker is then to set her estimate and spread such that $\forall i : \maker_i - \epsilon < \real_i < \maker_i + \epsilon$, making it impossible for any trader to make any profit, while securing herself an instant profit of $2\epsilon$ for each pair of orders traded at the opposing sides of the book.




However, given that $\maker=\hat{\real}$ is merely an estimate, it is still possible for the true value $\real$ to fall outside the $(\maker - \epsilon,\maker + \epsilon)$ region\footnote{Naturally, this could be mitigated by increasing the $\epsilon$, however a margin set too large will discourage investors from trading, consequently removing the market maker's profit, too.}. To further mitigate possible exploitation by the traders, the market maker can continuously adapt her estimate to the traders' behavior. That is she can responsively move her estimate $\maker$ once the demand of one side of the market starts to prevail, indicating expected value perceived by the traders, stemming from a possibly erroneous price estimate $\maker$. In the ideal case where she is continuously able to maintain a perfectly balanced book, so that half of the orders fall on the buy side and the other half on the sell side respectively, she is again guaranteed a unit profit of $\epsilon$ per trade. Note that the market maker's profit in this case is independent of the true price $\real$.

On the other hand, the market maker can theoretically digress from this purely reactive position to actively speculate against the takers' opinions and aim at a profit even higher than $\epsilon$, at the cost of involving risk stemming from her, possibly erroneous, estimate $\maker$. This effectively allows the market maker to speculate on the true price, which can lead to higher expected profits in settings where she possesses a superior price prediction model\footnote{Since the market maker typically needs an in-depth knowledge of the market to operate, she can often reasonably expect to also possess price estimates superior to the average trader.}. This is very common, for instance, with bookmakers in the sport prediction markets~\cite{ottaviani2008favorite}. Naturally, the market makers can combine the aforementioned methods.

We note we mostly omit the spread $\epsilon$ in calculations and simulations further in this paper for simplicity\footnote{except for the actual experiments with real data (Section~\ref{sec:experiments}) where the spread is naturally present.}. The underlying assumption is that the spread constitutes an independent offset on the market prices (and the resulting profits), and thus does not interfere with the main concepts introduced in this paper.


\subsection{Market Modelling}
\label{sec:market_modelling}




Market modelling generally refers to the approach of fitting a statistical estimator $\sf{e}$ to the available market-related data $\mathcal{D}$ in order to capture its true distribution $\mathrm{P}$. Possession of such a model then enables to answer all sorts of statistical queries, including the essential estimation of the real value $\realRV$ of market opportunities ${\Omega}$.



\paragraph{Data}
The market-related data $\mathcal{D}$ may come from various external sources such as news and economic signals and indicators, as well as from the market itself (e.g. the traders' behavior).
Recall that in an efficient market, all relevant data $\mathcal{D}^*$ must be used by the market maker $\maker$ for pricing of each opportunity $\forall i: \makerFunc(\mathcal{D}_i^*) \mapsto \makerVal$, however, in practice it is more likely that $\mathcal{D}_i^\maker \subset \mathcal{D}_i^*$.
The market takers can theoretically use the same data as the market maker $\mathcal{D}^\taker = \mathcal{D}^\maker$, although their sources are typically more limited (e.g. Section~\ref{sec:data}). However, they commonly strive to gain at least some information advantage by obtaining data which are not available to $\maker$, i.e. $\mathcal{D}^\taker \setminus \mathcal{D}^\maker \neq \varnothing$, and thus not reflected in the market price. Generally if $\mathcal{D}^\taker \subset \mathcal{D}^\maker$, such information advantage is completely missing (e.g. Section~\ref{sec:results}), making it impossible to beat the market through superior price estimations\footnote{However, we note that it is still possible to make profits in such a scenario (Section~\ref{sec:essence}).}, unless using a superior model.

\paragraph{Modelling}

The models used to fit the data are generally some $\theta$-parameterized functions, the properties of which also influence the quality of the estimation. Ultimately, one would strive to model the whole joint distribution over the data $\mathcal{D}$, enabling to answer all possible probabilistic queries about the domain, consequently leading to truly optimal investment decision making. Nevertheless, the common investment strategies (Section~\ref{sec:strategies}) are typically based merely on estimates of returns from the opportunities $\omega_i$ at hand, which restricts the task to modelling of a conditional of the true value $\mathrm{P}_\Omega(\realRV|\mathcal{D}_i)$.

For instance at the stock market, an estimate of the true value $\realVal$ at a given time can be modelled from a retrospective observation of local evolution of the price time series, i.e. one can use the model to repeatedly predict the (mean) future market price within some time interval (information context) used for trading. Similarly, an estimate of the true probability $\realVal$ in the betting market can be derived from repeated observations of the outcomes of the associated stochastic events aggregated over a large enough sample from the market.





\paragraph{Real Value}
Generally, the true value $\realVal$ of an asset $a$ at time $\tau$ can be a distribution itself, such as in the prediction markets. For instance, in an $n$-way betting market where there are $n$ possible outcomes, the true value of each asset is the probability distribution over the $n$ discrete outcomes. However, in the case of the two-way market with binary exclusive outcomes we use for demonstration, it is a simple Bernoulli distribution with a single parameter $\realVal$. Therefore, the true value $\realVal$ of an opportunity $\omega_i$ can still be considered as a single number $\realVal \in [0,1]$, since the value of the complementary opportunity is simply determined by $1-\realVal$. Consequently, there is a direct correspondence with the bid-ask stock market setting which we exploit for comprehensibility in this paper\footnote{However the proposed decorrelation concept (Definition~\ref{def:decorrelation}) can be directly extrapolated into n-way markets, too.}. It follows that we also consider both sides of the stock market to be available, i.e. with a short selling option, for consistency\footnote{although this is also not necessary for the approach to work.}. As a result, we can consider every opportunity $\omega_i$ to be associated with a single real value $\realVal \in \mathbb{R}_+$ in this paper.








\paragraph{Mispricing}

We assume a market that is not fully efficient (Section~\ref{sec:market_efficiency}), i.e. there must be some mispriced opportunities which further need to be identifiable by some systematic means. The goal of the traders is then to identify such opportunities where $\realVal \notin (\makerVal - \epsilon,\makerVal + \epsilon)$ by comparing the market price $\makerVal$ to their own estimate $\takerVal$. Should the market price $\makerVal$ for an opportunity $\omega_i$ actually deviate from the true value $\realVal$ in a non-random manner, such mispricing can be turned into positive returns.
While the concrete calculation of the expected returns differs across the two market settings (Section~\ref{sec:exp_return}), for the sake of this work it does not matter whether we try to predict future price distribution or unknown outcome probability distribution, as in both cases their correct estimation by $\takerFunc$ leads to profits of the trader $\taker$ in the (partially) inefficient market setting assumed. Naturally, this depends on the qualities of the underlying price estimators.


\subsection{Price Estimators}
\label{sec:estimators}
\label{sec:price_estimators}


A price estimator $\sf{e}$ is generally a $\theta$-parameterized function mapping some input data $\mathcal{D}_{i}$ associated with an opportunity $\omega_i$ onto a point price estimate $\hat{\realVal} \in \mathbb{R}_+$

\begin{equation}
    \sf{e} : \mathcal{D}_{i} \mapsto \hat{\realVal}
\end{equation}

However, every prediction is associated with a certain level of uncertainty, which is either inherent\footnote{Note that much of the information is often not just missing but principally unavailable, rendering the prediction problem inherently stochastic.} or stemming from the missing information at the time of making ($\mathcal{D}_i^{\taker} \subset \mathcal{D}_i^*$). One might thus want to quantify the uncertainty by associating each possible estimate $\hat{\realVal}$ for an opportunity $\omega_i$ with a probability, resulting into a posterior distribution estimation

\begin{equation}
    \Tilde{\sf{e}}: \mathcal{D}_{i} \mapsto \mathrm{\Tilde{R}}_i
\end{equation}
where $\mathrm{\Tilde{R}}_i$ is the estimated price distribution $\mathrm{\Tilde{P}}({\hat{\realRV}_i}|\Omega=\omega_i)$ for a single opportunity $\omega_i$. Such distribution $\mathrm{\Tilde{R}}_i$ can be estimated, e.g., from histograms of past prices or event outcomes, marginalized over the same or ``similar'' conditions ($\mathcal{D}_{i}$)~\cite{altman1992introduction}.
When a point price estimate is needed, such as when we need to actually trade an asset at a particular price, it is common to take an expected value from $\mathrm{\Tilde{R}}_i$, calculated by multiplying each point estimate $\hat{\realRV_i}={\hat{\real_j}}$ with its associated probability estimation $\mathrm{e}_j$ (or probability density $\mathrm{e}({\hat{\real}})$):

\begin{equation}
    \EX_{\mathrm{\Tilde{R}}_i}[\sf{\Tilde{e}}(\mathcal{D}_{i})] = \sum_j \mathrm{e}_j \cdot {\hat{\real_j}} \text{~~~~~~~~~or alternatively~~~~~~~~} \EX_{\mathrm{\Tilde{R}}_i}[\sf{\Tilde{e}}(\mathcal{D}_i)] = \int {\hat{\real}} \cdot \mathrm{e}({\hat{\real}})~\mathrm{d}{\hat{\real}}
\end{equation}
We note that most of the machine learning models provide directly the point estimates, realizing some functional mapping $\Omega \to \mathbb{R}$.


\paragraph{Point Estimates}
The aforementioned (middle) market price $\makerVal \in \mathbb{R}_+$ of an opportunity $\omega_i$ can then be thought of as the market maker's point estimate $\makerVal = \makerFunc(\mathcal{D}^\maker_{i})$ of the true price $\realVal$. Similarly, the market taker will also try to predict the actual true value, which we represent with her own point estimate $\takerVal = \EX_{\mathrm{\Tilde{\takerRV}}_i}[\takerFunc(\mathcal{D}^{\taker}_{i})]$, or simply $\takerVal = \takerFunc(\mathcal{D}_i^\taker)$.
By the definition of her role, the market maker $\maker$ continuously evaluates each opportunity $\forall i: \omega_i \to \makerVal$. We generally assume that the trader $\taker$ also has the ability to estimate (predict) the true price $\forall i: \omega_i \to \takerVal$ of each opportunity in the market\footnote{Alternatively, a systematically selected subset can be considered instead.}.
The trader $\taker$ then must posses her own estimate, i.e. be generally different from $\maker$\footnote{since e.g. copying the estimate from the market maker would not lead to any trading incentive.}, and the true value $\real$ must be unknown at the time of trading, otherwise there would be no incentive to trade. Note that the true value is often unknown in principle, i.e. even retrospectively, such as the probability in the betting markets, indicating the need for statistical treatment of the problem.


\begin{definition}
\label{def:market_dist}
We can now generalize the reasoning about individual opportunities and estimates to reasoning over the whole joint \textit{market distribution} $\mathrm{P}_\Omega(\realRV,\makerRV,\takerRV)$ capturing the relationships between the estimates across a whole set of opportunities $\omega_1,\dots,\omega_n \in {\Omega}$ generated by the market environment.
For each such opportunity $\omega_i$, we will thus operate with 3 distinct values we will refer to as (i) the true value $\realRV=\realVal$, (ii) the market maker's estimate $\makerRV=\makerVal$ and (iii) the market taker's estimate $\takerRV=\takerVal$. 
\end{definition}

Given the uncertainty, it follows that both the $\makerVal$ and $\takerVal$ estimates are always going to be to some extent erroneous, where the former guaranties existence of mispriced assets. Note that this is a necessary but insufficient condition for market inefficiencies to exist (Section~\ref{sec:market_efficiency}).





\paragraph{Optimization}

Typically, one optimizes the estimator by tuning its parameters $\theta \in \Theta$ to fit some historical market data $\mathcal{D}^\taker$ relevant for the prediction of the underlying true value $\realRV$.
Estimation of unknown parameters $\theta$ from empirical data $\mathcal{D}$ is then one of the key problems in statistics~\cite{friedman2001elements}. There are several views on the problem. One class of approaches is to maximize probability of the observed data with methods such as maximum likelihood $\mathrm{p}(\mathcal{D}|\theta)$ or maximum a-posteriori $\mathrm{p}(\theta|\mathcal{D})$ estimation. One can also directly search for an estimator with some desired target properties, such as minimum-variance unbiased estimator~\cite{voinov2012unbiased} or best linear unbiased estimator~\cite{henderson1975best}. A very common methodology is Bayesian estimation, both with or without an informative prior, where one tries to minimize (posterior) expectation of some error function~\cite{berger2013statistical}. We take the latter approach, while noting that there are close connections between all the approaches.

\paragraph{Estimation Error}

The quality of an estimator $\sf{e}$ can then be expressed through its empirical error $err(\sf{e})$ over some set of opportunities ${\Omega}$. Since we assume the role of the trader $\taker$, we further present error measurements between the true values $\realRV$ and the trader's estimates $\takerRV$. Some of the most popular error measures then include the mean square error ($MSE$):

\begin{equation}
\label{eq:mse}
    MSE_{{\Omega}}(\realRV,\takerRV) = \EX [(\takerRV - \realRV)^2] = \frac{1}{|{\Omega}|} \sum_{\omega_i \in {\Omega}} (\takerVal - \realVal)^2
\end{equation}
which is commonly used for regression tasks, such as the prediction of the true price $\realVal \in \mathbb{R}_+$, and the mean crossentropy ($XENT$):

\begin{equation}
    XENT_{{\Omega}}(\realRV,\takerRV) = \EX_{\realRV}[ -log(\takerRV)] = \frac{1}{|{\Omega}|} \sum_{\omega_i \in {\Omega}} \sum_{j \in outcomes} \real_j\cdot log(\taker_j)
\end{equation}
which is commonly used for classification tasks, such as predicting one of the discrete outcomes (e.g. $\{home,away\}$) of a sports match. Note that these correspond to the market sides $\zeta$. In the case of the two-sided markets corresponding to binary event outcomes, we can rewrite $XENT$ as

\begin{equation}
    XENT_{{\Omega}}(\realRV,\takerRV) = \frac{1}{|{\Omega}|} \sum_{\omega_i \in {\Omega}} r_i\cdot log(\taker_i) + (1-\real_i) \cdot log(1-\taker_i)
\end{equation}
which can then also be also understood as a regression of the underlying value $\real_i$ of each opportunity $\omega_i$.



These particular error functions are of special interest as minimizing $XENT$ generally corresponds to maximizing the data (log-)likelihood, while $MSE$ corresponds to maximizing the data likelihood with a linear gaussian model~\cite{lehmann2006theory}.
The crossentropy error is then also closely linked to a common measure of ``distance'' between probability distributions known as Kullback-Leibler divergence~\cite{kullback1997information}, which is defined as
\begin{equation}
\label{eq:KL}
    D_{KL}(\realRV||\takerRV) = - \sum_i \realVal \cdot log \frac{\takerVal}{\realVal}
\end{equation}
since
\begin{equation}
    XENT(\realRV,\takerRV) = H(\realRV) + D_{KL}(\takerRV||\realRV)
\end{equation}
where $H(R)$ is the entropy of the true value distribution $\realRV$. Considering that true distribution being estimated does not change, its entropy $H(\realRV)$ can be considered a constant, rendering the cross-entropy error $XENT(\realRV,\takerRV)$ minimization equivalent to minimizing the the KL-divergence $D_{KL}(\realRV||\takerRV)$, sometimes also referred to as the relative entropy~\cite{berger2013statistical} (see Section~\ref{sec:acc_profit} for further connections).


\paragraph{Bias and Variance}

The estimation error can be commonly though of in terms of bias and variance, which are key concepts in the analysis of estimators' quality~\cite{lehmann2006theory}. Bias quantifies the expected difference between the price estimate and the true value:
\begin{equation}
    \bias{\takerRV} = \EX [\takerRV] - \realRV
\end{equation}
while variance is the expected squared distance from the mean estimate:
\begin{equation}
    \var{\takerRV} = \EX[(\takerRV - \EX [\takerRV])^2]
\end{equation}

While we naturally want to minimize both the quantities, this is typically hard to do as they commonly reflect opposing criteria, giving rise to the so called ``bias-variance dilemma'', which is a central problem in statistical learning~\cite{friedman2001elements}. 
A common approach in statistics is to strive for a minimum-variance unbiased estimator, giving good results in most practical settings~\cite{voinov2012unbiased}. The constraint for unbiasedness reflects the necessity for the estimator not to differ systematically from the true price, i.e. the estimation should be correct on average. This constraint is very natural in the proposed market setting, since a systematically biased market price would be easy to directly detect and exploit (Section~\ref{sec:est_dist}).
Note also that any minimum-variance mean-unbiased estimator minimizes the $MSE$, which can be alternatively expressed as:
\begin{equation}
    MSE(\takerRV,\realRV) = \EX [(\takerRV - \realRV)^2] = (\EX [\takerRV] - \realRV )^2 + \EX[(\takerRV - \EX [\takerRV])^2] = \bias{\takerRV}^2 + \var{\takerRV}
\end{equation}




\subsection{Expected Returns}
\label{sec:exp_return}



Being able to predict the future price in a stock market, or estimate the true probability in a prediction market, can be directly turned into positive returns in trading.
Note we do not explicitly distinguish between the mean future price and true value, as in the fairly efficient market setting assumed (Section~\ref{sec:market_efficiency}), the market price $\makerRV$ cannot systematically deviate from the true value $\realRV$ for too long. Similarly in the prediction markets, the relative frequency of the observed outcomes will approach the true probability distribution in the long run. Being able to correctly estimate either of these thus guaranties systematic profits, even though actual profits might deviate from the expectation in short term.

Since the total profit is dependent on the actual amount invested, which is yet to be determined by the investment strategy (Section~\ref{sec:strategies}), the traders commonly assess \textit{relative} profitability of individual opportunities through measures such as rate of return or, without assuming any time period, return on investment (ROI), which we further denote as $\return_i$. In general, $\return_i$ is simply the relative return made from a unit investment w.r.t. the market price $\makerVal$ and the true value $\realVal$.


The uncertain, stochastic nature of the prediction problem renders the value estimates $\hat{r}_i$ for a particular $\omega_i$ as random variables $\hat{r}_i \sim \mathrm{\Tilde{R}}_i$ (Section~\ref{sec:estimators}). Consequently, a return $\return_i$ derived from such an estimate $\hat{r}_i$ will be a random variable, too. Instead of the actual $\return_i$ we can thus again calculate with its expectation $\EX[{\return_i}]$ w.r.t. the used estimator $\sf{e}$, further denoted as $\EX_{\sf{e}}[\return_i]$. For the role of a trader $\taker$, we can then define an expected $\return_i$ of an opportunity $\omega_i$ based on estimator $\takerFunc$ of the true price $\realVal$ of some underlying asset as

\begin{align}
\label{eq:profit_stock}
    \EX_{\takerFunc} [\return^{\alpha}_i] = \EX_{\Tilde{\takerRV}_i} \Bigg[ \frac{\Tilde{\takerRV}_i - \makerRV(\omega_i^\alpha)}{\makerRV(\omega_i^\alpha)} \Bigg] = \frac{\takerVal - \makerVal^\alpha}{\makerVal^\alpha} &~~~~~~~~~~~~~~~& \EX[ \return^{\beta}_i] = \EX_{\Tilde{\takerRV}_i} \Bigg[ \frac{\makerRV(\omega_i^\beta) - \Tilde{\takerRV}_i}{\makerRV(\omega_i^\beta)} \Bigg] = \frac{\makerVal^\beta-\takerVal}{\makerVal^\beta}
\end{align}
where $\maker_i^\zeta$ is the market price at which the trader executes the buy ($\zeta=\alpha$) or (short-)sell ($\zeta=\beta$) orders, respectively. For example when we buy a stock for $\$100$ with the expectation to sell it later for $\takerVal=\$150$ on average, our expected ROI is $\return_i=0.5$ ($50\%$).

In a prediction market, the market price (odds) directly reflect the relative returns, and so the expected $\return_i$ of an opportunity $\omega_i$ based on a probability estimated by $\takerFunc$ can be calculated as
\begin{align}
\label{eq:profit_bet}
    \EX_{\takerFunc} [\return^{\alpha}_i] = \EX_{\Tilde{\takerRV}_i} \Bigg[ \frac{\Tilde{\takerRV}_i}{\makerRV(\omega_i^\alpha)} -1 \Bigg] = \frac{\takerVal}{\makerVal^\alpha} -1 &~~~~~~~~~~~~~~~& \EX[ \return^{\beta}_i] = \EX_{\Tilde{\takerRV}_i} \Bigg[ \frac{\Tilde{\takerRV}_i}{\makerRV(\omega_i^\beta)} -1 \Bigg] = \frac{1-\takerVal}{\makerVal^\beta} - 1
\end{align}
where $\maker_i^\zeta$ is the probability estimated by the bookmaker for the binary outcomes of home ($\zeta=\alpha$) and away ($\zeta=\beta$) team wins, respectively. For example, a bet of $\$100$ on an outcome with estimated probability $\takerVal=\frac{3}{4}$, and associated bookmaker's (decimal) odds of $\frac{1}{\makerRV(\omega_i)}=2.0$, i.e. yielding a net return of $\$100$ if realized, and a loss of $\$-100$ if not, will also result into a ROI of $\return_i=0.5$ ($50\%$).



Although the average returns should converge to the true expected returns in the long run, these can still be very different from the predicted expected returns since generally $\EX_{\takerFunc}(\return_i) \neq \EX_{\realFunc}(\return_i)$. The discrepancy is of course conditioned by the quality of the predictor $\takerFunc$ w.r.t. the true $\realFunc$. The approach to minimize the prediction error (Section~\ref{sec:estimators}) then seems very natural, and there are also some theoretical guaranties like, for instance, in the case where $XENT(\realRV,\takerRV) < XENT(\realRV,\makerRV)$, i.e. the investor possesses a price prediction model superior to the market maker in terms of cross-entropy, we are guaranteed to make long-term profits with optimal investment routines such as the Kelly strategy (Section~\ref{sec:strategies}).

\paragraph{Risk and Utility}
To explicitly account for the discussed uncertainty involved in trading, the concept of \textit{risk} assessment has been proposed. This means that instead of direct maximization of the expected returns, one should strive for a balance between the expected profit and risk, stemming from the uncertainty. The risk can then be quantified by statistical means such as the variance of the expected profit~\cite{markowitz1952portfolio} or probability of a drawdown~\cite{busseti2016risk}.
Note that apart from the quantifiable risk, there is also a structural risk stemming from the fact that, similarly to the expected returns, the assessment of risk is based on merely estimated parameters.

Not all investors then share the same preferences to balance the expected profit and risk in the same manner. To incorporate individual preferences into the decision making, the concept of a \textit{utility function} $\sf{u}$ has been proposed to steer the investment optimization process. A utility function generally maps each alternative onto a real number, defining a total ordering over some set of alternative investments. In our case it is some monotonically growing function $\sf{u}$ transforming the net returns $\return$ into a new real quantity $\sf{u}(\return)$ to be optimized\footnote{Given the stochastic setting, we will again consider expected utility $\EX[\sf{u}(\return)]$ instead.}.
The concept of risk is then closely connected to utility, as maximizing any concave utility function directly reflects a preference for risk aversion~\cite{arrow1965aspects}. 




\subsection{Investment Strategies}
\label{sec:strategies}

An investment strategy can be seen as the final step of the traders's workflow.
Given the expected returns from individual trading opportunities present at a time, the trader needs to decide on how to \textit{allocate} her wealth across the available opportunities $\Omega$ in order to optimize her utility ${u}$, i.e. some requested trade-off between expected returns and risk. Formally, the investment strategy is a function $\sf{s}$ mapping a vector of opportunities $\bm{\omega}$ associated with the estimated returns $\bm{\return}$ onto a wealth allocation vector $\bm{f}$:

\begin{equation}
     \text{\sf{s}} :  \bm{\return} \mapsto \bm{f} ~~~\text{  ~~~~~i.e.~~~~ } \sf{s}: \mathbb{R}^n \to \mathbb{R}^n
\end{equation}
The vector of the wealth allocations $\bm{f}$, corresponding to portions of some current wealth $\sf{W}$, is then often referred to as a \textit{portfolio} over the opportunities $\bm{\omega}$. There are several approaches to this problem and we will briefly review some of the most popular ones~\cite{li2014online}.

\subsubsection{Uniform Portfolio}
\label{sec:strat:unit}
The most trivial investment strategy is a uniform, unit-staking strategy where one independently allocates the same absolute amount $d$ on every opportunity with an assumed positive expected return:

\begin{equation}
    \sf{s} : \return_i \mapsto 
    \begin{cases}
        f_i = d, & \text{if }~ \return_i > 0\\
        f_i = 0, & \text{otherwise}
    \end{cases}
\end{equation}

Despite being very naive, this strategy is also most robust against estimation errors~\cite{pflug20121}, since the allocation simply remains constant no matter the circumstances. Note that the individual investments $d$ are not considered as fractions relative to the current wealth $\sf{W}$ here, and so a small enough unit $d \ll \sf{W}$ has to be chosen so that it is possible to invest into all profitable opportunities. Given that the allocated unit $d$ is small enough, this strategy will typically also have a very conservative risk profile, nevertheless the expected portfolio profits can be way below optimal. The size of $d$ then remains a hyperparameter the choice of which is left to the user.

\subsubsection{Modern Portfolio Theory}
\label{sec:MPT}


A more principled approach is that of the Modern Portfolio Theory (MPT) \cite{markowitz1952portfolio} which strives to balance optimally between the expected return and risk. The general idea behind MPT is that a portfolio $\bm{f^1}$, i.e. a vector of asset capital allocations $\bm{f} = f_1, \dots, f_n$ over some opportunities $\omega_1,\dots,\omega_n$, is superior to $\bm{f^2}$, if its corresponding expected return $\return$ (Section~\ref{sec:exp_return}) is at least as great
\begin{equation}
    \EX_{\takerFunc}[\bm{\return} \cdot \bm{f^1}] \geq \EX_{\takerFunc}[\bm{\return} \cdot \bm{f^2}]
\end{equation}
and a given risk measure $risk : \mathbb{R}^n \to \mathbb{R}$ of the portfolio w.r.t. the returns is no greater
\begin{equation}
    risk_{\EX_{\takerFunc}[{\return}]}(\bm{f^1}) \leq risk_{\EX_{\takerFunc}[{\return}]}(\bm{f^2})
\end{equation}
This creates a partial ordering on the set of all possible portfolios. When combined into a joint utility, we can trade-off the expected profit vs. risk by maximizing the following
\begin{equation}
    \underset{\bm{f} \in \mathbb{R}^n}{\text{maximize}} ~(\EX_{\takerFunc}[\bm{\return} \cdot \bm{f}] - \gamma \cdot risk_{\EX_{\takerFunc}[{\return}]}(\bm{f}))
\end{equation}
where $\gamma$ is a hyperparameter reflecting the user's preference for risk.

In the most common setup, the $risk$ of a portfolio $\bm{f}$ is measured through the expected total variance of its profit $\var{\bm{\return} \cdot \bm{f}} = \bm{f}^T\Sigma \bm{f}$, based on a given covariance matrix $\bm{\Sigma}_n^n$ of returns of the individual opportunities, which can be again estimated from historical data (Section~\ref{sec:market_modelling}). MPT can then be expressed as the following constrained maximization problem:
\begin{equation}
\label{eq:MPT}
\begin{aligned}
& \underset{\bm{f} \in \mathbb{R}^n}{\text{maximize}}~
& & \EX_{\takerFunc}[\bm{\return}\cdot\bm{f}]  - \gamma \cdot \bm{f}^T\Sigma \bm{f}\\
& \text{subject to}
& & \sum_{i=1}^{n} f_i = 1
\end{aligned}
\end{equation}

Note that the capital allocations sum up to one as they simply reflect fractions of the current bankroll $\sf{W}$, and they can possibly be negative if short selling is enabled.

The main weakness of MPT is that the variance of profit is hardly a good measure of risk for profit distributions other than Gaussian~\cite{rom1994post}. Apart from the variance $\var{\bm{w}}$ of the potential net returns $\bm{w} = \bm{\return} \cdot \bm{f}$, different risk measures have been proposed~\cite{markowitz1952portfolio}, such as standard deviation $\sigma(\bm{w}) = \sqrt{\var{\bm{w}}}$ and coefficient of variation $CV(\bm{w}) = \frac{\sigma(\bm{w})}{\EX[\bm{w}]}$. Nevertheless these all share the same weakness. Generally, there is no {agreed-upon} measure of risk, rendering the whole concept a bit dubious. Moreover, the strategy only works with the opportunities $\bm{\omega}$ currently at hand, and thus ignores any knowledge about the actual market distribution $\mathrm{P}_\Omega$.

\paragraph{Sharpe Ratio}
Apart from the choice of the risk measure, the inherent degree of freedom in MPT is how to select a particular portfolio from the efficient frontier (based on the choice of $\gamma$). Perhaps the most popular way to avoid the dilemma is to select a spot in the pareto-front with the highest expected profits w.r.t. the risk. For the risk measure of $\sigma(\bm{w})$, this is known as the ``Sharpe ratio'' \cite{sharpe1994sharpe}, generally defined as
\begin{equation}
\frac{\EX_{\takerFunc}[\bm{w}] - r_f}{\sigma(\bm{w})}
\end{equation}
where $\EX[\bm{w}]$ is the expected return of the portfolio, $\sigma(\bm{w})$ is the standard deviation of the return, and $r_f$ is a ``risk-free rate''. We do not consider any risk free investment in our setting, and so we can reformulate the optimization problem as
\begin{equation}
\begin{aligned}
& \underset{\bm{f} \in \mathbb{R}^n}{\text{maximize}}
& & \frac{\EX_{\takerFunc}[\bm{\return} \cdot \bm{f}]} {\sqrt{\bm{f}^{T}\bm{\Sigma}\bm{f}}} \\
& \text{subject to}
& & \sum_{i=1}^{n} f_i = 1
\end{aligned}
\end{equation}

\subsubsection{Kelly Criterion}
\label{sec:back_Kelly}





The Kelly criterion{~\cite{kelly1956new, thorp2008kelly}} assumes the investment problem in time\footnote{Note this is in contrast to MPT which assumes the problem in an ensemble of traders at the same time, i.e. through expectation.}, i.e. it optimizes multi-period investments in contrast to MPT which is concerned only with single-period portfolio returns. It is based on the idea of expected multiplicative growth $W_{\sf{G}}$ of a continuously reinvested bankroll $\sf{W}_{\tau}$. The goal is to a find a portfolio $\bm{f}$ such that the long-term expected value of the resulting profit $\sf{W}_{\tau\to\infty}$ is maximal, which is equivalent to maximizing the geometric growth rate of wealth defined as

\begin{equation}
    W_{\sf{G}} = \underset{t\to\infty}{lim}~log \Bigg(\frac{W_t}{W_o}\Bigg)^{\frac{1}{t}}
\end{equation}

For its multiplicative nature, it is also known as the geometric mean policy, emphasizing the contrast to the arithmetic mean approaches (e.g. MPT) based directly on the expected value of wealth.
The two can, however, be looked at similarly with the use of a logarithmic utility function, transforming the geometric into the arithmetic mean, and the expected geometric growth rate into the expected value of wealth, respectively. The problem can then be again expressed by the standard means of maximizing the (estimated) expected utility value as

\begin{equation}
\begin{aligned}
& \underset{\bm{f} \in \mathbb{R}^n}{\text{maximize}}
& & \EX_{\takerFunc}[\log(1 + \bm{f}^T \cdot \bm{\return})]\\
& \text{subject to}
& & \sum_{i=1}^{n} f_i = 1
\end{aligned}
\end{equation}
Note that, in contrast to MPT, there is no explicit term for risk here, as the notion of risk is inherently encompassed in the growth-based view of the wealth progression, i.e. the long-term value of a portfolio that is too risky will be smaller than that of a portfolio with the right risk balance (and similarly for portfolios that are too conservative). The risk is thus captured by the logarithmic (concave) utility transformation itself.

The calculated portfolio is then provably optimal, i.e. it accumulates more wealth than any other portfolio chosen by any other strategy in the limit of $\tau\to\infty$. However, this strong result only holds given, considerably unrealistic, assumptions~\cite{kelly1956new, thorp2008kelly, peters2016evaluating}. Similarly to MPT, we assume to know the true returns while calculating merely with estimates and additionally, given the underlying growth perspective, that we are repeatedly presented with the same opportunities from $P_\Omega$ ad infinitum,
making the optimality of the growth-based risk treatment in Kelly likewise a bit dubious. Despite the fact that the given conditions are impossible to meet in practice, the Kelly strategy is very popular, particularly its various modifications to mitigate the aforementioned issues.

\paragraph{Fractional Kelly}
The result of the Kelly optimization problem is, for each opportunity, the ideal fraction $\omega \mapsto f^*$ one is ought to invest to achieve the maximal long-term profits. The fraction $f^*$ thus dictates an upper-bound on the possible profit, meaning that increasing the invested fraction further will actually decrease the long-term profit\footnote{this is due to the assumed multiplicative, growth-based view of Kelly, which is in contrast to the additive MPT, where overbetting would merely increase the risk.}. This is commonly known as ``overbetting''. Since the true expected return $\rho$ is unknown, however, such a situation might occur even while betting with a fraction assumed to be optimal. Intentionally decreasing the calculated fraction $f^*$ by some ratio $\frac{1}{d'}$ then decreases the risk of overbetting stemming from a possibly overvalued estimate. Such an approach is commonly referred to as ``fractional Kelly''~\cite{maclean1992growth}. Ideally, one should estimate the optimal shrinkage ${d'}$ as another hyperparameter~\cite{baker2013optimal,uhrin2019sports} based on backtesting performance, however, it is very common to simply choose a fixed ratio such as $\frac{1}{2}$ of the estimated optimal Kelly fraction $f^*$, commonly referred to as ``half Kelly'' by practitioners. While there are other remedies to mitigate the risk with the Kelly criterion~\cite{busseti2016risk, sun2018distributional}, fractional Kelly is a very effective method which is widely adopted in practice due to its simplicity. In addition to mitigating the overbetting risk, it generally decreases volatility, which also tends to be considerably high with the plain Kelly criterion.

\paragraph{Correspondence to MPT}
\label{sec:kelly-MPT}

While Kelly is clearly based on different principles than MPT, there is an interesting close connection between the two strategies. Following~\cite{busseti2016risk}, let us make an assumption for a Taylor series approximation that our net profits are not too far from zero $\bm{\return}^T\cdot{\bm{f}} \approx \bm{0}$,
allowing us to proceed with the Taylor expansion of the optimized growth as
\begin{equation}
    \log(1 + \bm{\return}^T \cdot \bm{f}) = \bm{\return}^T \cdot \bm{f} - \frac{(\bm{\return}^T \cdot \bm{f})^{2}}{2} + ...
\end{equation}
Now taking only the first two terms from the series we transform the expectation of logarithm into a new problem objective as follows
\begin{equation}
\begin{aligned}
& \underset{\bm{f \in \mathbb{R}^n}}{\text{maximize}}
& & \EX\big[\bm{\return}^T \cdot \bm{f} - \frac{(\bm{\return}^T \cdot \bm{f})^{2}}{2}\big] \\
\end{aligned}
\end{equation}
Note that, interestingly, the problem can now be rewritten to
\begin{equation}
\begin{aligned}
& \underset{\bm{f} \in \mathbb{R}^n}{\text{maximize}}
& & \EX[\bm{\return}^T \cdot \bm{f}] - \frac{1}{2}\EX[\bm{f}^T (\bm{\return} \cdot \bm{\return}^T) \bm{f}] \\
& \text{subject to}
& & \sum_{i=1}^{n} f_i = 1.0, ~f_i \geq 0
\end{aligned}
\end{equation}
corresponding to the original MPT formulation from Equation~\ref{eq:MPT} for the particular user choice of $\gamma=\frac{1}{2}$.
It follows from the fact that the geometric mean is approximately the arithmetic mean minus $\frac{1}{2}$ of variance~\cite{markowitz1952portfolio}, providing further insight into {the} connection of the two popular strategies of Kelly and Markowitz, respectively. While the solution is merely an approximation, it also tends to be more robust to estimation errors than the original Kelly, similarly to the fractional approach.






\section{Problem Insights}
\label{sec:insights}

We have introduced the main parts of a common workflow of a trader $\taker$, who relies on a statistical estimator $\takerFunc$ to predict true value $\realRV$ of market opportunities $\Omega$ based on some available relevant data $\mathcal{D}^\taker$, with the resulting estimates of the expected returns $\EX_{\takerFunc}[\return]$ being fed into some subsequent portfolio optimization strategy $\sf{s}$ to produce final wealth allocations $\bm{f}$. In this Section, we provide some key insights into the problem of profitability from the perspective of the predictive model $\takerFunc$.

Let us briefly recall the problem setup. We generally consider the problem of profiting from the trader's $\taker$ perspective as a stochastic game against the market maker $\maker$.
The market maker $\maker$ uses an estimator $\makerFunc$ to continuously price the incoming opportunities $\Omega$. Following some investment strategy, the trader $\taker$ then takes particular $\alpha$ and $\beta$ actions (allocations) upon these opportunities $\Omega$, based on her own estimates produced by $\takerFunc$.
Given some distribution of market opportunities $\mathrm{P}_{\Omega}$, we can then set up the game in terms of three random variables $\realRV,\makerRV,\takerRV$ corresponding to the true value, market maker's, and market taker's estimates, respectively. 
The goal of the trader (as well as the market maker) is then to maximize her expected (long-term) profits $\sf{W}$ as measured by some utility $\sf{u}$ underlying the chosen strategy $\sf{s}$.

\subsection{From Accuracy to Profit}
\label{sec:acc_profit}

The key issue with the optimal investment strategies based on portfolio optimization is that they are inherently relying on accurate estimates of the asset returns. Their performance is then directly stemming from the quality of these estimates - the better the estimates, the higher the utility of the portfolio can be achieved in general.
While this holds to an extent for most of the common portfolio optimization strategies\footnote{the correspondence between Kelly and MPT is shown in Section~\ref{sec:back_Kelly}}, it is best demonstrated on the optimal investment approach of Kelly.


For simplicity of demonstration, let us consider an idealized case of an $n$-way betting market with no margin (Section~\ref{sec:maker_adv}) on the market maker's odds. Recall that the Kelly strategy is to find wealth fractions $\bm{f}$ so as to
\begin{equation}
    \underset{\bm{f}}{\text{maximize}~~} \EX_{\realRV}[log(\bm{f}^T \cdot (1+\bm{\return}))] = \sum_{i=1}^n \realVal \log(f_i \cdot \frac{1}{\makerVal}), \text{~~~~s.t.~~} \sum_{i=1}^n{f_i} = 1
\end{equation}

Note that in this idealized case, we calculate the expectation of returns w.r.t. the true distribution $P_{\Omega}(\realRV)$. It can then proved~\cite{cover2012elements} that the solution to this constrained optimization problem yields

\begin{equation}
    \bm{f}^* = \bm{\real}
\end{equation}
i.e. the optimal fraction of wealth $f_i$ to invest in each outcome (opportunity) $\omega_i$ is directly equal to the underlying true value (probability) $\realVal$. Interestingly, we can see that in this case, the optimal strategy for the investor is to completely ignore the market pricing $\makerVal$ and focus solely on having the true values predicted correctly, in which case she is guaranteed the maximal possible long term profits. This is commonly known amongst Kelly practitioners as ``betting your beliefs''. Note that this strong result was derived from $\EX_{\realRV}$ and thus it only holds if the true distribution $P_{\Omega}(\realRV)$ is known or, more precisely, if the error in its estimate via $\takerRV=\hat{\realRV}$ as measured through the Kullback-Leibler divergence (Section~\ref{sec:estimators}) is zero~\cite{cover2012elements}:

\begin{equation}
    D_{KL}(\realRV||\takerRV) = - \sum_{i=1}^n \realVal\cdot log \frac{\takerVal}{\realVal} = 0
\end{equation}


Naturally, it is close to impossible to estimate the true distribution $P_{\Omega}(\realRV)$ perfectly in practice\footnote{We exclude scenarios with known artificial distributions such as in casino games.}. Let us thus extrapolate into more practical settings by relaxing the condition into a non-zero $D_{KL}(\realRV||\takerRV)$.
Given that the optimal fractions $\bm{f}$ should be equal to the true outcome probabilities $\bm{\real}$, let us substitute back into the long term growth rate of wealth which Kelly seeks to maximize as

\begin{equation}
   W_{\sf{G}} = \sum_{i=1}^n \realVal \log(\takerVal \cdot \frac{1}{\makerVal})
\end{equation}
Now, following the proof from~\cite{cover2012elements}, this can be rewritten into
\begin{equation}
    W_{\sf{G}} = \sum_{i=1}^n \realVal \log\frac{\takerVal}{\realVal} + \sum_{i=1}^{n} \realVal \log\frac{\realVal}{\makerVal}
\end{equation}
and consequently, using the formula for KL-divergence (Equation~\ref{eq:KL}), back into
\begin{equation}
\label{eq:KL_diff}
    W_{\sf{G}} = D_{KL}(\realRV||\makerRV) - D_{KL}(\realRV||\takerRV) 
\end{equation}
showing the important insight that, for Kelly, the growth of wealth of the trader is directly equal to the difference in quality of her estimates $\takerRV$ over the market prices $\makerRV$ in terms of KL-divergence from the true values $\realRV$. Consequently, positive returns can only be achieved iff the model of the trader achieves a lower cross-entropy error than the market $XENT_{\Omega}(R,M) < XENT_{\Omega}(R,B)$ (Section~\ref{sec:estimators}). Given the information theoretic interpretation of the relative entropy~\cite{kullback1997information}, this is sometimes referred to as the aforementioned ``information advantage'' of the trader over the market maker.

\paragraph{Implications}
Note that this result was derived with the assumption of seeking growth-optimal investments, and its extrapolation beyond that setting may lead to wrong conclusions. Particularly, it is true that one does need a better\footnote{We note we do not distinguish between $XENT$ and other measures of model accuracy here for simplicity.} model to make positive profits if committed to invest optimally with Kelly, however, this does not imply that one needs a better model to make positive profits if one does not care about the growth optimality.

The constraint for better model accuracy in classic portfolio optimization techniques is then inherently connected to the notion of risk (Section~\ref{sec:exp_return}), which is embedded together with expected returns into the same quantity being optimized. For instance, in the Markowitz's model, it is easy to show that even negative return portfolio may be preferred to positive returns should the latter be associated with higher variance. From the Kelly's perspective, overvalued positive return estimates may actually lead to negative growth due to overbetting (Table~\ref{tab:ordering}), and it is also commonly necessary to allocate certain amount of wealth onto opportunities (outcomes) with negative returns to achieve optimal portfolio performance in the long run~\cite{uhrin2018thesis}.

Consequently, wrong assessment of the true prices and probabilities associated with either of such opportunities can lead to inappropriate (over-)investments, resulting into a negative overall profit, even in situations where positive returns could be generally achieved otherwise.

\subsection{The Essence of Profit}
\label{sec:essence}

While the performance of common portfolio optimization strategies is tightly bound to the accuracy of price predictions (Section~\ref{sec:acc_profit}), we argue that accuracy is not essential for profitability in general. This is best demonstrated by taking the, sophisticated but questionable, notions of risk out of the optimization scope, resorting back to simple strategies such as the uniform investments (Section~\ref{sec:strat:unit}). Consequently, one can simply base profitability directly on the ability to correctly detect opportunities with positive expected returns (Section~\ref{sec:exp_return}). Note now that whether the expectation from an opportunity $\omega_i$ is deemed positive depends purely on the comparison between $\takerVal$ and $\makerVal$. This boils down to the renown ``buy low, sell high'' policy to trade mispriced\footnote{Note we again exclude spread (Section~\ref{sec:maker_adv}) from calculations in this Section for simplicity.} assets simply as:

\begin{equation}
\makerVal \neq \takerVal
    \begin{cases}
        \makerVal < \takerVal \implies \alpha = \text{buy \textit{assumed} underpriced asset (bet home)}\\
        \makerVal > \takerVal \implies \beta = \text{sell \textit{assumed} overpriced asset (bet away)}\\
    \end{cases}
\end{equation}

Naturally, to asses the actual return from a trade, the true asset value $\realVal$ needs to be accounted for. In a stock market setting, the actual return from a supposedly profitable opportunity can then be defined as

\begin{equation}
    \EX_{\realRV}[\return_i] =
    \begin{cases}
        \frac{\realVal-\makerVal}{\makerVal}, & \text{if }~ \makerVal < \takerVal~~~(\text{buying asset})\\
       \frac{\makerVal-\realVal}{\makerVal}, & \text{if }~ \makerVal > \takerVal~~~(\text{selling asset})\\
    \end{cases}
\end{equation}
and similarly in a prediction market as
\begin{equation}
    \EX_{\realRV}[\return_i] =
    \begin{cases}
        \frac{\realVal}{\makerVal}-1, & \text{if }~ \makerVal < \takerVal~~~(\text{betting home})\\
        \frac{1-\realVal}{1-\makerVal}-1, & \text{if }~ \makerVal > \takerVal~~~(\text{betting away})\\
    \end{cases}
\end{equation}


Note that the true expected return $\EX_{\realRV}[\return_i]$ can clearly be negative and that its absolute value is not dependent of the model's estimate $\takerVal$. However, the ability to correctly recognize the profitable opportunities through the $\makerVal \lessgtr \takerVal$ comparison is naturally dependent on the ordering of the $\takerVal$ estimates w.r.t. $\makerVal$ and $\realVal$.
Note nevertheless that this comparison-based quality is very different from the accuracy-based reasoning. Consequently, even if $err(\takerFunc) > err(\makerFunc)$, a consistent profit can still be made, as we demonstrate through the following simple examples.


\begin{example}
Assume that a fundamental value of an asset is $\$10$, with the marker maker estimating the price at $\$9$, and the market taker's estimate being $\takerVal=\$15$. Clearly, the market taker's estimate is more erroneous here (e.g. $MSE(\takerVal) > MSE(\makerVal)$). Since the true value is not observable, she compares her $\$15$ evaluation of the asset to the currently offered $\$9$, which she \textit{correctly} evaluates as a profitable opportunity to \textit{buy} the asset (since $9<15$). By buying one unit, her actual expected return from this trade will be positive at $\$1$ (i.e. $10-9$).
\end{example}

\begin{example}
Similarly, assume true probability of an outcome to be $0.6$, with the bookmaker $\maker$ estimating it at $\makerVal=0.5$, with the corresponding fair odds set up to $2.0$, and the bettor $\takerFunc$ estimate being at $\takerVal=0.9$. Clearly, the bettor's estimate is more erroneous here (e.g. $XENT(\takerVal) > XENT(\makerVal)$). Nevertheless she has no choice but to use her estimate to asses the return on investment, which she {correctly} estimates as being {positive} ($\EX_{\takerFunc}[\return]=\frac{0.9}{0.5}-1 > 0$). Despite being very wrong numerically with her expectation of a $80\%$ ROI, by betting a unit of wealth, she can still expect to obtain the actual positive ROI of $20\%$.
\end{example}

Note that the trader's $\taker$ estimates $\takerVal$ in these examples could have been set arbitrarily larger (within the respective domain), making the corresponding model $\takerFunc$ arbitrarily bad by the standard error measures.

\begin{definition}
\label{def:essence}

Followingly, let us define a more relaxed, necessary condition of \textit{essential profitability} of a model $\takerFunc$ simply as the consequent existence of one of the following market opportunities $\omega_i$ in $P_{\Omega}(\realRV,\makerRV,\takerRV)$:

\begin{enumerate}
    \item the market undervalues the true price, and the model estimates a higher value than the market, \\i.e. $\makerVal<\realVal \wedge \takerVal>\makerVal$
    \item the market overvalues the true price, and the model estimates a lower value than the market, \\i.e. $\makerVal>\realVal \wedge \takerVal<\makerVal$
\end{enumerate}
Using sufficiently conservative (small) wealth allocations, investments into either of these cases will lead to systematic profits of the market taker in the long run. On the contrary, no investment strategy can lead to positive profits without such opportunities in the portfolio. 

\end{definition}

Nevertheless the overall profitability of a model $\takerFunc$ will naturally depend on the relative occurrence of such $\omega_i$'s in the actual market distribution $P_{\Omega}$ (Definition~\ref{def:market_dist}).
Let us now generalize the essence of profitability, from reasoning about the necessary relationships between individual $\realVal,\makerVal,\takerVal$ estimates, to the properties of the whole market distribution $\mathrm{P}_\Omega(\realRV,\makerRV,\takerRV)$. Following on the aforementioned ``buy low, sell high'' strategy with uniform investments, the expected profitability from a market distribution  $\mathrm{P}_{\Omega}$ is clearly

\begin{equation}
    \EX_{P_{\Omega}} [\return] = \sum_i \EX_{\realRV}[\return_i] \cdot \mathrm{P}_{\Omega}(\realVal,\makerVal,\takerVal)
\end{equation}
We already know that the fundamental value $\realVal$ is a principally unknown random variable and one can thus never perfectly assess the true return $\return_i$ from any trade in advance, for which we resort to an estimate $\hat{\rho}_i$. However, the market distribution $\mathrm{P}_{\Omega}(\realRV,\makerRV,\takerRV)$ here is also principally unknown, for which one again needs to rely on statistics while estimating it from historical data as $\hat{\mathrm{P}}$. Consequently, one can estimate the essential profitability of a model $\takerFunc$ w.r.t. market pricing $\makerFunc$ as

\begin{equation}
    \EX_{\hat{P}}[\widehat{\return}] =  \sum_i \EX_{\takerRV}[\return_i] \cdot \mathrm{\hat{P}}(\takerVal,\makerVal,\takerVal)
\end{equation}

As with any investment strategy, the calculated expected returns can be very different from the actual return distribution $\EX_{\hat{P}}[\widehat{\return}] \neq \EX_{P_{\Omega}} [\return]$, depending on the properties of the estimates (Section~\ref{sec:exp_return}). However, profitability of the simple unit investment strategy leads to a much more relaxed and robust condition on the model quality, which is what we exploit to yield positive profits even with estimators of inferior predictive performance.

\paragraph{Drawbacks}
We acknowledge that by deflecting from the accuracy-based view and focusing merely on the essential profitability with the unit investments, we downplay the role of explicit optimization of growth and risk in the formal strategies of Kelly and Markowitz, respectively. Nevertheless, as discussed in the respective Sections~\ref{sec:back_Kelly} and~\ref{sec:MPT}, these formal notions are based on rather unrealistic (wrong) assumptions, which is why additional risk management practices, such as the fractioning (Section~\ref{sec:back_Kelly}), need to be commonly employed with the strategies anyway~\cite{maclean1992growth,maclean2011kelly,uhrin2019sports}. Consequently, sacrificing formal optimality w.r.t. unrealistic objectives in order to transition from negative to positive profits does no seem that big of a sacrifice.






\subsection{Market Taker's Advantage}
\label{sec:taker_adv}

The market maker's advantage (Section~\ref{sec:maker_adv}) is a well-worn concept. However, there is also an advantage of the market taker which is rarely discussed explicitly, but is essential to the traders profitability. While the market maker has the obligation to continuously \textit{quote} {price} of both sides of the market (Section~\ref{sec:estimators}), the taker has the crucial liberty to \textit{select} only those of the resulting opportunities deemed profitable. That is she is to decide whether and which side of the market to trade the assets {once} the market maker's prices have been laid out. As the second player, the difficulty of her task is reduced from the correct price estimation to the estimation of the market price error direction. While this might seem as a similarly difficult problem, the latter is a considerably easier task.

\begin{table}

\caption{All possible orderings of the true value ($\realVal$), market maker's ($\makerVal$), and trader's estimates ($\takerVal$), with the implied trading decisions and resulting profitability. Additionally, the relative size of the implied Kelly fraction is indicated.}
\centering
\begin{tabular}{|ccccc|}
\hline
values ordering                  & decision ($\takerVal \lessgtr \makerVal$) & profit    & proportion (in $P_\mathcal{U}$)       &   Kelly         \\ \hline
$\realVal < \takerVal < \makerVal$                     & \textcolor{Green}{sell}                      & \textcolor{Green}{$\makerVal>\realVal$}          &  1/6 & overbet \\
$\realVal < \makerVal <\takerVal  $                  & \textcolor{Red}{buy}                      & \textcolor{red}{$\realVal<\makerVal $ }          & 1/6 & overbet     \\         
$\takerVal < \realVal < \makerVal  $                  & \textcolor{Green}{sell}                      &\textcolor{Green}{ $\makerVal>\realVal$}         &  1/6 & underbet\\
$\takerVal < \makerVal < \realVal  $                  & \textcolor{red}{sell}                      & \textcolor{red}{ $ \makerVal<\realVal$ }        &  1/6 & underbet \\
$\makerVal <\takerVal < \realVal   $               & \textcolor{Green}{buy}                      & \textcolor{Green}{$\realVal>\makerVal$ }          & 1/6   & underbet   \\        
$\makerVal < \realVal <\takerVal   $                 & \textcolor{Green}{buy}                      & \textcolor{Green}{$\realVal>\makerVal$ }         & 1/6  & overbet             \\ \hline
\end{tabular}
\label{tab:ordering}
\end{table}

For demonstration, consider the three values $\realVal,\makerVal,\takerVal$ of the true value, market maker's, and trader's estimates, respectively, to be laid out completely at random, yielding a uniform market distribution $\mathrm{P}_{\mathcal{U}}$ where $\realRV,\makerRV,\takerRV \sim \mathcal{U}^3$. The possible situations that emerge from the ordering of $\realVal,\makerVal,\takerVal$ in such a setting are displayed in Table~\ref{tab:ordering}. 
Since neither of the estimators possesses any information w.r.t. $\realRV$, both $\makerRV$ and $\takerRV$ clearly perform equally by the means of arbitrary statistical estimation measures (Section~\ref{sec:estimators}). While one might thus expect this to be a neutral trading setting for both the sides, interestingly, the market taker $\taker$ would already be able to make a substantial profit with uniform investments by correctly identifying $2/3$ of the profitable opportunities.

Intuitively, this demonstrates a simple fact that it is generally more likely to overestimate an undervalued estimate than to further underestimate it, i.e. 
\begin{equation}
    \makerVal<\realVal \implies \mathrm{P}_{\mathcal{U}}(\makerVal<\takerVal) > \mathrm{P}_{\mathcal{U}}(\takerVal<\makerVal)
\end{equation}
and vice versa for an overvalued estimate:
\begin{equation}
    \makerVal>\realVal \implies \mathrm{P}_{\mathcal{U}}(\makerVal<\takerVal) < \mathrm{P}_{\mathcal{U}}(\takerVal<\makerVal)
\end{equation}
Note, importantly, how this property of the completely uninformative $\mathrm{P}_{\mathcal{U}}$ is conveniently aligned with the essential profitability of the trader's model (Definition~\ref{def:essence}). The concrete proportions of the individual situations will naturally depend on the particular distribution $\mathrm{P}_{\Omega}$, nevertheless the property holds very generally for unskewed distributions with unbiased estimators (Section~\ref{sec:est_dist}). Note the difference from the standard model accuracy measures which would all evaluate both models equally in $\mathrm{P}_{\mathcal{U}}$. Nevertheless from the perspective of profitability, the situation is very different since, as opposed to the market maker, the market taker is not penalized for estimation errors in these two situations that emerge more often than not.
Consequently in $\mathrm{P}_{\Omega} = \mathrm{P}_{\mathcal{U}}$, the trader is in an inherent advantage of $2:1$, which can be directly turned into the corresponding profits.


It is perhaps more instructive to demonstrate the concept on a particular level of true value $\real$. Without loss of generality, let us consider all opportunities with a real value $\real$ being traded in an ideal two-way market (without spread). We can then plot a 2D projection of the $\mathrm{P}(\realRV,\makerRV,\takerRV)$ market distribution by conditioning it as $\mathrm{P}(\makerRV,\takerRV|\realRV=\real)$, and visualize the essential profitability (Definition~\ref{def:essence}) of the corresponding sub-regions of the distribution. The result is displayed in Figure~\ref{fig:projection}. We can observe that the distribution of the profitable regions (green) is clearly in favor of the trader, and that the potential returns progressively increase with the error of the market maker, i.e. the distance of $\makerRV=\maker$ from $\real$.

\begin{figure}[t]

\centering
\resizebox{0.6\textwidth}{!}{
\begin{tikzpicture}

\node [dotted, draw=gray, shape=rectangle, minimum width=10cm, minimum height=10cm, anchor=south west] at (-1,1) {};
\node [dotted, draw=gray, shape=rectangle, minimum width=10cm, minimum height=10cm, anchor=south west] at (-0.5,0.5) {};
\node [draw, shape=rectangle, minimum width=10cm, minimum height=10cm, anchor=south west] at (0,0) {};

\draw[help lines, color=gray!30, dashed] (0,0) grid (10,10);

\foreach \x in {0.5, ..., 10} {
    \foreach \y in {0.5, ..., 9.5} {
        \pgfmathsetmacro{\vy}{0.08*(\y-4)}
            
        \ifthenelse{\lengthtest{\y pt = \x pt}}{}{
        \ifthenelse{\lengthtest{\y pt > \x pt}}{
        
        \ifthenelse{\lengthtest{\y pt < 4 pt}}{
        \node at (\x, \y)[circle, fill=Red, scale=0.15] {};
        \draw[->,color=Red!90] (\x,\y) -- (\x, \y+\vy);
        }{
        \node at (\x, \y)[circle, fill=Green, scale=0.15] {};
        \draw[->,color=Green!90] (\x,\y) -- (\x, \y+\vy);
        }

        }{
        
        \ifthenelse{\lengthtest{\y pt > 4 pt}}{
        \node at (\x, \y)[circle, fill=Red, scale=0.15] {};
        \draw[->,color=Red!90\x] (\x,\y) -- (\x, \y+\vy);
        }{
        \node at (\x, \y)[circle, fill=Green, scale=0.15] {};
        \draw[->,color=Green!90\x] (\x,\y) -- (\x, \y+\vy);
        }
        };
        }
    }
}

\draw [dotted] (0,0) edge node[dotted,near end,sloped,below] {$\makerRV=\takerRV$} (10,10);

\draw [dotted] (4,0) edge node[dotted,sloped,near end,below] {$\takerRV=\real$} (4,10);
\draw [dotted] (0,4) edge node[dotted,sloped,near end,below] {$\makerRV=\real$} (10,4);

\draw [dotted,<->] (-1,1) -- node[dotted,near start,above] {\textbf{$\realRV$}} (1,-1);
\draw [dotted,<->] (1,9) -- node[dotted,near end,above] {$\realRV$} (-1,11);
\draw [dotted,<->] (11,9) -- node[dotted,near end,above] {$\realRV$} (9,11);
\draw [dotted,->] (10,0) -- node[dotted,near end,above] {$\realRV$} (11,-1);

\node [draw=none,anchor=west,color=black] at (4.1,3.8) {$(\real,\real,\real)$};
\node at (4,4)[circle, fill=black, scale=0.35] {};
\node at (0,0)[circle, fill=black, scale=0.35] {};

\node [draw=none] at (2,7) {$\takerRV<\real<\makerRV$};
\node [draw=none] at (6,8) {$\real<\takerRV<\makerRV$};
\node [draw=none] at (8,6) {$\real<\makerRV<\takerRV$};
\node [draw=none] at (7,2) {$\makerRV<\real<\takerRV$};
\node [draw=none] at (2.5,1) {$\makerRV<\takerRV<\real$};
\node [draw=none] at (1.5,3) {$\makerRV<\takerRV<\real$};

\node [draw=none] at (4,-0.3) {$\real$};
\node [draw=none,anchor=east] at (0,4) {$\real$};

\node [draw=none,anchor=east] at (0,5) {\textbf{$\makerRV$}};
\node [draw=none,anchor=north] at (5,0) {\textbf{$\takerRV$}};

\node [draw=none] at (-0.2,-0.5) {$(\real,0,0)$};


\end{tikzpicture} 
}
\caption{A 2D projection of the $\mathrm{P}(\realRV,\makerRV,\takerRV)$ distribution onto $\mathrm{P}(\makerRV,\takerRV|\realRV=\real)$ with visualization of essential profitability (Section~\ref{sec:essence}) of the individual point-estimates as a vector field. Green color denotes positive returns and red negative returns, respectively, while the length of each vector corresponds to the ROI in an idealized two-way stock market (Section~\ref{sec:exp_return}).
}
\label{fig:projection}
\end{figure}
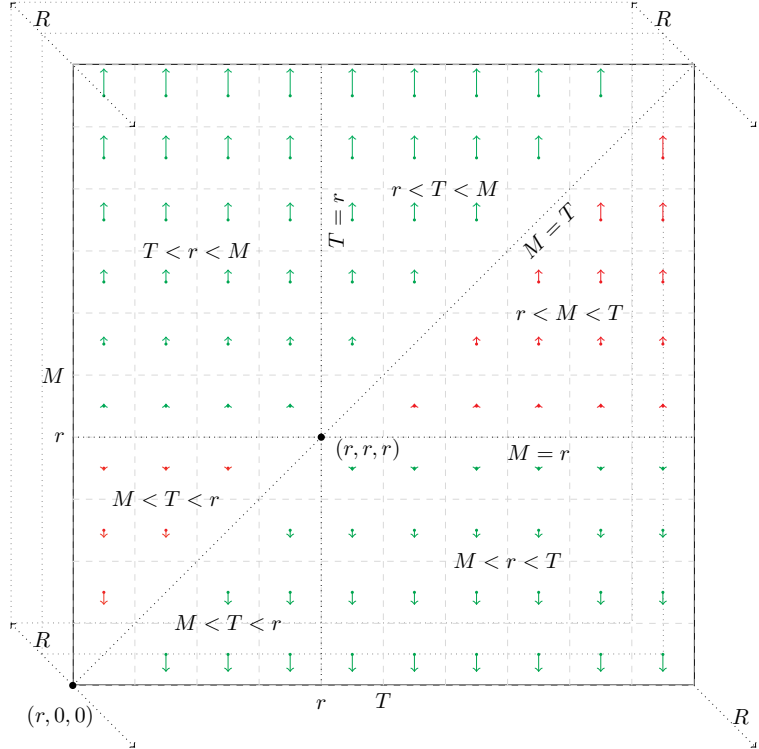


Note that here we did not yet assume any particular (non-uniform) distribution of estimates, and this inherent advantage of the trader is thus completely oblivious of any information advantage as well as any other property of the price estimators. Rather, it stems purely from the unequal roles of the market maker $\maker$ and taker $\taker$. Consequently, should they e.g. switch roles with the same estimation models ($\maker \gets \takerFunc, \taker \gets \makerFunc$), the advantage would stay exactly the same on the side of the market taker $\taker$.


\subsection{Distribution of Estimates}
\label{sec:est_dist}


While demonstrating the concept of the trader's advantage, the uniform $\mathrm{P}_\mathcal{U}(\realRV,\makerRV,\takerRV)$ distribution of estimates with completely uninformed players from Section~\ref{sec:taker_adv} seems rather unlikely in practice. In real world markets, the values $\realRV,\makerRV,\takerRV$ are not going to be independent, but rather correlated with each other, since $\makerRV$ and $\takerRV$ are typically based on similar information sources and both try to model $\realRV$ with similar techniques (Section~\ref{sec:market_modelling}). Let us now review common properties of the more realistic market distributions of these estimates.

\paragraph{Bias}
The market makers are typically very good at being close to the real price, and we can assume their estimates $\makerRV$ to be \textit{unbiased} w.r.t. true value $\realRV$, or, more formally:
\begin{equation}
    \EX_\mathrm{P}(B)-R = 0
\end{equation}
which means that they are not systematically deviating when measured against the true value alone. Should a market maker be biased in this manner\footnote{biased by more than $\epsilon$ in practice (Section~\ref{sec:maker_adv}).}, it would be extremely easy to exploit her merely by trading all opportunities on the corresponding side of the market. Moreover we can reasonably assume her to be point-wise unbiased at each particular price level $\real$, i.e.
\begin{equation}
    \EX_\mathrm{P}(B|R=\real) = r
\end{equation}
If that was not the case, the market maker would again be easily exploitable by correspondingly trading all possible assets within a certain price range $\real \pm \delta$, i.e. by buying all assets in price ranges where the maker $\maker$ systematically undervalues the assets, and vice versa for selling in overvalued price ranges. One can typically check from historical data that the market makers are not biased in this simple manner in any reasonably efficient market\footnote{as we do in experiments in Section~\ref{sec:search_bias}}. Note that the unbiasedness is only one of the conditions for a fully efficient market (Section~\ref{sec:market_efficiency}). There is generally no reason for the market taker $\taker$ to be biased in this trivial way either, unless a systematic error is present in her model $\takerFunc$, or she reversely reflects the market maker's bias to exploit it.

\paragraph{Variance}


Given the assumption that the models behind $\makerRV$ and $\takerRV$ are both unbiased estimators of $\realRV$, we can now focus merely on their (co-)variances. It is a common practice in statistical estimation of functions for one to look for an estimator with the smallest variance among the class of unbiased estimators~\cite{} (Section~\ref{sec:estimators}). Since we are working with function estimators, we are not interested in the total variance of the estimations $\takerRV$, which includes variance due to variations in the true price $\realRV$ itself as
\begin{equation}
    \var{\takerRV} = \EX_\realRV[\var{\takerRV|\realRV}] + \mathrm{Var}_{\realRV}{\EX[\takerRV|\realRV]}
\end{equation}
but merely in the first term capturing the expected variance left w.r.t. predicting $\realRV$. Given the assumption of the point-wise unbiased estimates, the conditional variances w.r.t. $\realRV$ are then equal to the covariances of the models, i.e.
\begin{equation}
    \cov{\makerRV,\realRV} = \var{\makerRV|\realRV} \text{~~~~and similarly~~~~} \cov{\takerRV,\realRV} = \var{\takerRV|\realRV}
\end{equation}
Given the unbiasedness, these co-variances then directly reflect the quality (accuracy) of the underlying models of the market maker ($\cov{\makerRV,\realRV}$) and market taker ($\cov{\takerRV,\realRV}$), respectively.

The last degree of freedom in terms of covariances in $\mathrm{P}$ is the relationship between $\makerRV$ and $\takerRV$, i.e. $\cov{\makerRV,\takerRV}$ ($\cov{\makerRV,\takerRV|\realRV}$). While the other two covariances have a clear common interpretation, the $\cov{\makerRV,\takerRV}$ is more intriguing, but is also essential to the proposed profitability of inferior predictive models (Section~\ref{def:essence}).
As we have seen in the case of the uniform market distribution $\mathrm{P}_\mathcal{U}$, where both the players possess the same amount of information w.r.t. the true price, corresponding to the same accuracies of $\makerFunc$ and $\takerFunc$, the trader $\taker$ is always in advantage. However not all distributions with equally informed players are as such, as demonstrated by the following example.

\begin{example}
Assume a scenario where both the trader $\taker$ and market maker $\maker$ possess the exact same model. Clearly, their information value, accuracy and all statistical measures will be exactly the same, just as in the case of the uniform distribution $\mathrm{P}_\mathcal{U}$. Nevertheless, the trader will not be in an advantageous position anymore. Since all her estimates coincide with the market price $\forall i: \takerVal = \makerVal$, it is not possible to detect any profitable opportunities where $r \not\in (b-\epsilon,b+\epsilon)$, even if they exist in the market distribution $\mathrm{P}_{\Omega}$. Hence, the profitability of the trader in this case is clearly zero.
\end{example}

From the statistical viewpoint, one can note an underlying difference between the two example distributions in the third covariance term $\cov{\makerRV,\takerRV}$. Whereas in the uniform distribution, the two variables were completely independent, i.e. $\cov{\makerRV,\takerRV}=0$, here they are equal and thus display maximal possible covariance ($\cov{\makerRV,\takerRV}=\sigma^2$). While this anecdotal reference to the connection between profitability and $\cov{\makerRV,\takerRV}$ is rather informal, we analyse it in detail in the next Section~\ref{sec:proof}.

%


\section{Increasing Profit through Decorrelation}
\label{sec:proof}
\label{sec:decorrelating}


Recall that we have a two-way market with opportunities $\Omega$ of some fundamental value $\realRV$, being priced by the market maker $\maker$ as $\makerRV$ and taker $\taker$ as $\takerRV$, resulting into some market distribution of estimates $\mathrm{P}_{\Omega}(\realRV,\makerRV,\takerRV)$ (Definition~\ref{def:market_dist}). Let us now consider the context of the common properties (Section~\ref{sec:est_dist}) of such market distributions $\mathrm{P}_{\Omega}$, allowing assessments of model performance in terms of expected returns $\EX_{\mathrm{P}_{\Omega}}(\return)$ w.r.t. the distribution $\mathrm{P}_{\Omega}(\realRV,\makerRV,\takerRV)$. Particularly, we will explore the aforementioned statistical relationship between the market maker $\maker$ and taker $\taker$. To further standardize the relationship study, i.e. to take the individual variances of $\makerRV$ and $\takerRV$ out of scope, we now switch from covariance $\cov{\takerRV,\makerRV|\realRV}$ to correlation $\corr{\takerRV,\makerRV|\realRV}$.




\begin{definition}
\label{def:decorrelation}
We use the term \textit{decorrelation} to refer to the concept of decreasing the partial correlation $\corr{\takerRV,\makerRV|\realRV}$ between a price estimator $\takerFunc$ of the trader and the market maker $\makerFunc$ w.r.t. the real value $\realFunc$ across opportunities $\omega_i \in \Omega$ endowed with some market distribution $\mathrm{P_\Omega}$ (Definition~\ref{def:market_dist}).
\end{definition}

The main goal of this Section~\ref{sec:proof} is to show that enforcing smaller\footnote{Note that utilize the term ``decorrelation'' to refer to decreasing the correlation even below zero.} partial correlation $\corr{\takerRV,\makerRV|\realRV}$ of a price estimator $\takerFunc$ generally increases its essential profitability (Definition~\ref{def:essence}) within common market distributions.


\subsection{Unbiased Estimators}

Let us first consider the common setting of unbiased price estimators, the reasoning behind which was introduced in Section~\ref{sec:est_dist}.

\begin{theorem}

The essential profitability (Definition~\ref{def:essence}) of an unbiased estimator $\takerFunc$ with the lowest partial correlation $\corr{\takerRV,\makerRV|\realRV=\real}=-1$ with the market $\makerFunc$ is maximal. Consequently, no deviation from such a model $\takerFunc$ can thus increase the profitability further.

\end{theorem}




\begin{proof}

Now for $\corr{\takerRV,\makerRV|\realRV}=-1$ and an arbitrary $\real$, the probability distribution $\mathrm{P}(\takerRV,\makerRV|\realRV=\real)$ collapses into a linear function $(\taker;a,b) \mapsto \maker$ of the form

\begin{equation}
    \maker = a \cdot \taker + b~~\text{where}~~a=\frac{\cov{\makerRV,\takerRV|\real}}{\var{\takerRV|\real}}~~\text{and}~~b=(1-a) \cdot \real
\end{equation}

where the $b$ is set so that the mean values $\EX_{\mathrm{P}(\takerRV,\makerRV|\realRV=\real)}[\makerRV]=\EX_{\mathrm{P}(\takerRV,\makerRV|\realRV=\real)}[\takerRV]$ of both the marginals $\mathrm{P}(\makerRV|\real)$ and $\mathrm{P}(\takerRV|\real)$ are at $\real$ (unbiased). The mean of the distribution $\maker=\taker=\real$ thus lies on the line:
\begin{equation}
    \real = a \cdot \real + (1-a)\cdot \real
\end{equation}

Clearly, since $\corr{\takerRV,\makerRV|\real}<0$, we have also $\cov{\makerRV,\takerRV|\real}<0$ implying a negative slope $a<0$ of the function line.
There are consequently only 2 possible price estimate orderings (regions in Figure~\ref{fig:projection}) for all $\Omega$, both of which are profitable, as follows:

\begin{enumerate}
    \item $\takerRV<r<\makerRV$ implying stock returns $\frac{\maker-\real}{\maker} > 0$, or likewise for betting $\frac{1-\real}{1-\maker}-1 > 0$
    \item $\makerRV<r<\takerRV$ implying stock returns $\frac{\real-\maker}{\maker} > 0$, or likewise for betting $\frac{\real}{\maker}-1 > 0$
\end{enumerate}
Note again that the individual values of returns do not depend on the value of $\takerRV$ but merely on the value order. From the assumed role of the trader $\taker$, the only possible changes to the distribution (and profit) can be made via changes in her model $\takerFunc$ estimates $\takerRV$. However, any potential deviation from this distribution (line) with $\corr{\makerRV,\takerRV|\realRV}=-1$ can only result into one of the following ordering (region) transitions:

\begin{enumerate}
    \item $\takerRV<r<\makerRV$ can change to either:
    \begin{enumerate}
        \item $r<\takerRV<\makerRV$ implying \textit{no change} in the returns $\frac{\maker-\real}{\maker} > 0$, or likewise $\frac{1-\real}{1-\maker}-1 > 0$
        \item $r<\makerRV<\takerRV$ implying \textit{decrease} in the returns to $\frac{\real-\maker}{\maker} < 0$, or likewise $\frac{\real}{\maker}-1 < 0$
    \end{enumerate}
    \item $\makerRV<r<\takerRV$ can change to either:
    \begin{enumerate}
        \item $\makerRV<\takerRV<r$ implying \textit{no change} in the returns $\frac{\real-\maker}{\maker} > 0$, or likewise $\frac{\real}{\maker}-1 > 0$
        \item $\takerRV<\makerRV<r$ implying \textit{decrease} in the returns to $\frac{\maker-\real}{\maker} < 0$, or likewise $\frac{1-\real}{1-\maker}-1 < 0$
    \end{enumerate}
\end{enumerate}
ergo no deviation from $\corr{\makerRV,\takerRV|\realRV}=-1$ can increase the profitability any further. 
\end{proof}

Note that this also means that we cannot increase the returns even via transition into a perfect estimator with both $\bias{\takerRV}=\var{\takerRV|\realRV}=0$, i.e. a perfect model $\mathrm{P}(\takerRV=\real|\realRV=\real)=1$ which always returns the correct answer in a deterministic fashion (and for which the partial correlation would be undefined). This means that for the purpose of the essential profit generation (Definition~\ref{def:essence}), the variance of the model no longer acts as an error since it is completely turned into our advantage.
Note this is in direct contrast to a correlated investor who, given a variance higher than the market maker, is doomed to obtain completely negative returns\footnote{We also note that the transition between the two corner cases is somewhat smooth w.r.t. the returns for common market distributions, such as the elliptical distributions used for visualization.}. A visualization of the correlation effect on the returns, with an example elliptical market distribution for the given setting, is depicted in Figure~\ref{fig:unbiased}.


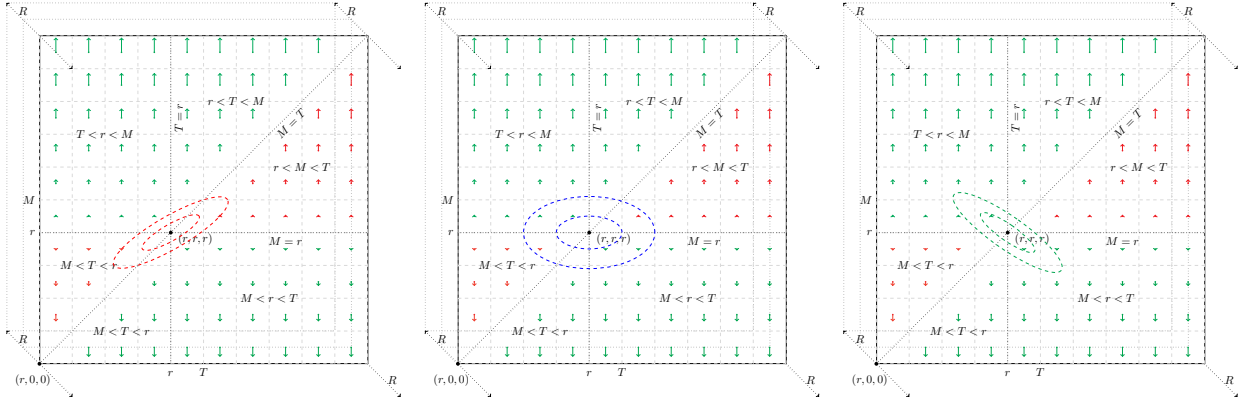
\begin{figure}

\centering
\resizebox{0.32\textwidth}{!}{
\begin{tikzpicture}
\node[scale=1] (grid) at (0,0){\begin{tikzpicture}

\node [dotted, draw=gray, shape=rectangle, minimum width=10cm, minimum height=10cm, anchor=south west] at (-1,1) {};
\node [dotted, draw=gray, shape=rectangle, minimum width=10cm, minimum height=10cm, anchor=south west] at (-0.5,0.5) {};
\node [draw, shape=rectangle, minimum width=10cm, minimum height=10cm, anchor=south west] at (0,0) {};

\draw[help lines, color=gray!30, dashed] (0,0) grid (10,10);

\foreach \x in {0.5, ..., 10} {
    \foreach \y in {0.5, ..., 9.5} {
        \pgfmathsetmacro{\vy}{0.08*(\y-4)}
            
        \ifthenelse{\lengthtest{\y pt = \x pt}}{}{
        \ifthenelse{\lengthtest{\y pt > \x pt}}{
        
        \ifthenelse{\lengthtest{\y pt < 4 pt}}{
        \node at (\x, \y)[circle, fill=Red, scale=0.15] {};
        \draw[->,color=Red!90] (\x,\y) -- (\x, \y+\vy);
        }{
        \node at (\x, \y)[circle, fill=Green, scale=0.15] {};
        \draw[->,color=Green!90] (\x,\y) -- (\x, \y+\vy);
        }

        }{
        
        \ifthenelse{\lengthtest{\y pt > 4 pt}}{
        \node at (\x, \y)[circle, fill=Red, scale=0.15] {};
        \draw[->,color=Red!90\x] (\x,\y) -- (\x, \y+\vy);
        }{
        \node at (\x, \y)[circle, fill=Green, scale=0.15] {};
        \draw[->,color=Green!90\x] (\x,\y) -- (\x, \y+\vy);
        }
        };
        }
    }
}

\draw [dotted] (0,0) edge node[dotted,near end,sloped,below] {$\makerRV=\takerRV$} (10,10);

\draw [dotted] (4,0) edge node[dotted,sloped,near end,below] {$\takerRV=\real$} (4,10);
\draw [dotted] (0,4) edge node[dotted,sloped,near end,below] {$\makerRV=\real$} (10,4);

\draw [dotted,<->] (-1,1) -- node[dotted,near start,above] {\textbf{$\realRV$}} (1,-1);
\draw [dotted,<->] (1,9) -- node[dotted,near end,above] {$\realRV$} (-1,11);
\draw [dotted,<->] (11,9) -- node[dotted,near end,above] {$\realRV$} (9,11);
\draw [dotted,->] (10,0) -- node[dotted,near end,above] {$\realRV$} (11,-1);

\node [draw=none,anchor=west,color=black] at (4.1,3.8) {$(\real,\real,\real)$};
\node at (4,4)[circle, fill=black, scale=0.35] {};
\node at (0,0)[circle, fill=black, scale=0.35] {};

\node [draw=none] at (2,7) {$\takerRV<\real<\makerRV$};
\node [draw=none] at (6,8) {$\real<\takerRV<\makerRV$};
\node [draw=none] at (8,6) {$\real<\makerRV<\takerRV$};
\node [draw=none] at (7,2) {$\makerRV<\real<\takerRV$};
\node [draw=none] at (2.5,1) {$\makerRV<\takerRV<\real$};
\node [draw=none] at (1.5,3) {$\makerRV<\takerRV<\real$};

\node [draw=none] at (4,-0.3) {$\real$};
\node [draw=none,anchor=east] at (0,4) {$\real$};

\node [draw=none,anchor=east] at (0,5) {\textbf{$\makerRV$}};
\node [draw=none,anchor=north] at (5,0) {\textbf{$\takerRV$}};

\node [draw=none] at (-0.2,-0.5) {$(\real,0,0)$};


\end{tikzpicture} };
\draw[rotate around={30:(-1,-1)}, red, dashed] (-1,-1) ellipse (2cm and 0.5cm);
\draw[rotate around={30:(-1,-1)}, red, dashed] (-1,-1) ellipse (1cm and 0.25cm);
\end{tikzpicture}
}
\resizebox{0.32\textwidth}{!}{
\begin{tikzpicture}
\node[scale=1] (grid) at (0,0){\begin{tikzpicture}

\node [dotted, draw=gray, shape=rectangle, minimum width=10cm, minimum height=10cm, anchor=south west] at (-1,1) {};
\node [dotted, draw=gray, shape=rectangle, minimum width=10cm, minimum height=10cm, anchor=south west] at (-0.5,0.5) {};
\node [draw, shape=rectangle, minimum width=10cm, minimum height=10cm, anchor=south west] at (0,0) {};

\draw[help lines, color=gray!30, dashed] (0,0) grid (10,10);

\foreach \x in {0.5, ..., 10} {
    \foreach \y in {0.5, ..., 9.5} {
        \pgfmathsetmacro{\vy}{0.08*(\y-4)}
            
        \ifthenelse{\lengthtest{\y pt = \x pt}}{}{
        \ifthenelse{\lengthtest{\y pt > \x pt}}{
        
        \ifthenelse{\lengthtest{\y pt < 4 pt}}{
        \node at (\x, \y)[circle, fill=Red, scale=0.15] {};
        \draw[->,color=Red!90] (\x,\y) -- (\x, \y+\vy);
        }{
        \node at (\x, \y)[circle, fill=Green, scale=0.15] {};
        \draw[->,color=Green!90] (\x,\y) -- (\x, \y+\vy);
        }

        }{
        
        \ifthenelse{\lengthtest{\y pt > 4 pt}}{
        \node at (\x, \y)[circle, fill=Red, scale=0.15] {};
        \draw[->,color=Red!90\x] (\x,\y) -- (\x, \y+\vy);
        }{
        \node at (\x, \y)[circle, fill=Green, scale=0.15] {};
        \draw[->,color=Green!90\x] (\x,\y) -- (\x, \y+\vy);
        }
        };
        }
    }
}

\draw [dotted] (0,0) edge node[dotted,near end,sloped,below] {$\makerRV=\takerRV$} (10,10);

\draw [dotted] (4,0) edge node[dotted,sloped,near end,below] {$\takerRV=\real$} (4,10);
\draw [dotted] (0,4) edge node[dotted,sloped,near end,below] {$\makerRV=\real$} (10,4);

\draw [dotted,<->] (-1,1) -- node[dotted,near start,above] {\textbf{$\realRV$}} (1,-1);
\draw [dotted,<->] (1,9) -- node[dotted,near end,above] {$\realRV$} (-1,11);
\draw [dotted,<->] (11,9) -- node[dotted,near end,above] {$\realRV$} (9,11);
\draw [dotted,->] (10,0) -- node[dotted,near end,above] {$\realRV$} (11,-1);

\node [draw=none,anchor=west,color=black] at (4.1,3.8) {$(\real,\real,\real)$};
\node at (4,4)[circle, fill=black, scale=0.35] {};
\node at (0,0)[circle, fill=black, scale=0.35] {};

\node [draw=none] at (2,7) {$\takerRV<\real<\makerRV$};
\node [draw=none] at (6,8) {$\real<\takerRV<\makerRV$};
\node [draw=none] at (8,6) {$\real<\makerRV<\takerRV$};
\node [draw=none] at (7,2) {$\makerRV<\real<\takerRV$};
\node [draw=none] at (2.5,1) {$\makerRV<\takerRV<\real$};
\node [draw=none] at (1.5,3) {$\makerRV<\takerRV<\real$};

\node [draw=none] at (4,-0.3) {$\real$};
\node [draw=none,anchor=east] at (0,4) {$\real$};

\node [draw=none,anchor=east] at (0,5) {\textbf{$\makerRV$}};
\node [draw=none,anchor=north] at (5,0) {\textbf{$\takerRV$}};

\node [draw=none] at (-0.2,-0.5) {$(\real,0,0)$};


\end{tikzpicture} };
\draw[rotate=0, blue, dashed] (-1.,-1) ellipse (2cm and 1.1cm);
\draw[rotate=0, blue, dashed] (-1.,-1) ellipse (1cm and .5cm);
\end{tikzpicture}
}
\resizebox{0.32\textwidth}{!}{
\begin{tikzpicture}
\node[scale=1] (grid) at (0,0){\begin{tikzpicture}

\node [dotted, draw=gray, shape=rectangle, minimum width=10cm, minimum height=10cm, anchor=south west] at (-1,1) {};
\node [dotted, draw=gray, shape=rectangle, minimum width=10cm, minimum height=10cm, anchor=south west] at (-0.5,0.5) {};
\node [draw, shape=rectangle, minimum width=10cm, minimum height=10cm, anchor=south west] at (0,0) {};

\draw[help lines, color=gray!30, dashed] (0,0) grid (10,10);

\foreach \x in {0.5, ..., 10} {
    \foreach \y in {0.5, ..., 9.5} {
        \pgfmathsetmacro{\vy}{0.08*(\y-4)}
            
        \ifthenelse{\lengthtest{\y pt = \x pt}}{}{
        \ifthenelse{\lengthtest{\y pt > \x pt}}{
        
        \ifthenelse{\lengthtest{\y pt < 4 pt}}{
        \node at (\x, \y)[circle, fill=Red, scale=0.15] {};
        \draw[->,color=Red!90] (\x,\y) -- (\x, \y+\vy);
        }{
        \node at (\x, \y)[circle, fill=Green, scale=0.15] {};
        \draw[->,color=Green!90] (\x,\y) -- (\x, \y+\vy);
        }

        }{
        
        \ifthenelse{\lengthtest{\y pt > 4 pt}}{
        \node at (\x, \y)[circle, fill=Red, scale=0.15] {};
        \draw[->,color=Red!90\x] (\x,\y) -- (\x, \y+\vy);
        }{
        \node at (\x, \y)[circle, fill=Green, scale=0.15] {};
        \draw[->,color=Green!90\x] (\x,\y) -- (\x, \y+\vy);
        }
        };
        }
    }
}

\draw [dotted] (0,0) edge node[dotted,near end,sloped,below] {$\makerRV=\takerRV$} (10,10);

\draw [dotted] (4,0) edge node[dotted,sloped,near end,below] {$\takerRV=\real$} (4,10);
\draw [dotted] (0,4) edge node[dotted,sloped,near end,below] {$\makerRV=\real$} (10,4);

\draw [dotted,<->] (-1,1) -- node[dotted,near start,above] {\textbf{$\realRV$}} (1,-1);
\draw [dotted,<->] (1,9) -- node[dotted,near end,above] {$\realRV$} (-1,11);
\draw [dotted,<->] (11,9) -- node[dotted,near end,above] {$\realRV$} (9,11);
\draw [dotted,->] (10,0) -- node[dotted,near end,above] {$\realRV$} (11,-1);

\node [draw=none,anchor=west,color=black] at (4.1,3.8) {$(\real,\real,\real)$};
\node at (4,4)[circle, fill=black, scale=0.35] {};
\node at (0,0)[circle, fill=black, scale=0.35] {};

\node [draw=none] at (2,7) {$\takerRV<\real<\makerRV$};
\node [draw=none] at (6,8) {$\real<\takerRV<\makerRV$};
\node [draw=none] at (8,6) {$\real<\makerRV<\takerRV$};
\node [draw=none] at (7,2) {$\makerRV<\real<\takerRV$};
\node [draw=none] at (2.5,1) {$\makerRV<\takerRV<\real$};
\node [draw=none] at (1.5,3) {$\makerRV<\takerRV<\real$};

\node [draw=none] at (4,-0.3) {$\real$};
\node [draw=none,anchor=east] at (0,4) {$\real$};

\node [draw=none,anchor=east] at (0,5) {\textbf{$\makerRV$}};
\node [draw=none,anchor=north] at (5,0) {\textbf{$\takerRV$}};

\node [draw=none] at (-0.2,-0.5) {$(\real,0,0)$};


\end{tikzpicture} };
\draw[rotate around={55:(-1,-1)}, Green, dashed] (-1,-1) ellipse (0.5cm and 2cm);
\draw[rotate around={55:(-1,-1)}, Green, dashed] (-1,-1) ellipse (0.25cm and 1cm);
\end{tikzpicture}
}
\caption{A sketch of the profitability regions from Figure~\ref{fig:projection} with an example distribution of estimates density for the case of unbiased estimations of $\makerRV$ and $\takerRV$ w.r.t $\realRV=\real$, where $\var{\takerRV|\realRV} > \var{\makerRV|\realRV}$, corresponding to a common situation in practice (Section~\ref{sec:est_dist},~\ref{sec:experiments}). The effect of $\corr{\takerRV,\makerRV|\realRV}$ on profitability is demonstrated on a highly correlated model (left,red), an independent model (middle,blue), and a highly negatively correlated model (right,green). By decreasing the correlation, progressively larger parts of the profitable regions of the distributions (green) are being covered, reflecting the corresponding increase in returns.}
\label{fig:unbiased}
\end{figure}

\subsection{Biased Estimators}


While common, the assumption of point-wise unbiased estimators might be seen as too strict in practice. The market maker is very unlikely to be overly biased, due to its constant exposure to the traders, who would likely exploit such easy opportunities (Section~\ref{sec:est_dist}). Nevertheless the market takers are generally free to come up with all sorts of models. Let us briefly review the situation where the trader's model $\takerFunc$, which we seek to optimize, is more biased than $\makerFunc$ w.r.t. $\realRV$, i.e. $\bias{\takerRV} > \bias{\makerRV}$. 

In this setting, one can craft counterexamples showing that $\takerRV$ with $\corr{\takerRV,\makerRV|\realRV}=-1$ is not universally the most profitable model $\takerFunc$ anymore\footnote{For instance, consider a distribution where the model $\takerFunc$ is biased w.r.t. $\realRV$ by some $\delta$ as

\begin{equation}
    \takerRV =
    \begin{cases}
            \makerRV - 1.1 \cdot \delta, & \text{ if~~}  \takerRV > r\\
            -\makerRV+ 1.1 \cdot \delta, & \text{ if~~}  \takerRV < r
    \end{cases}
\end{equation}
This $\delta$-biased model $\takerFunc$ is set to make maximal possible profit, and decreasing its correlation with $\makerFunc$ will actually hurt its performance. However, achieving such a distribution of estimates $\mathrm{P}_{\delta}(\realRV,\makerRV,\takerRV)$ is close to impossible in practice, as it is carefully crafted w.r.t. the unknown value of $r$. Consequently, this scenario is highly unstable w.r.t. $r$ as well as variance of $\takerRV$, a decreasing of which will paradoxically lead to complete loss of all returns.}.
While not the universally best possible model, a decorrelated $\takerRV$ will still perform very well in practice here. Particularly, it will be consistently better than a highly correlated model, even if the latter has a lower variance $\var{\takerRV|\realRV}$. Interestingly, it is also better than a model with zero variance (vertical line), i.e. given some bias, we are able to turn the variance into an advantage by decreasing correlation with $\makerRV$. Lastly, the minimal correlation will be also typically better than no correlation for common elliptical or uniform conditional distributions. 
A visualization of this setting, where $\bias{\takerRV} > \bias{\makerRV}$, is displayed in Figure~\ref{fig:taker_biased} for an example elliptical distribution. Note that the same reasoning is also applicable to cases where the $\bias{\takerRV} = \bias{\makerRV}$.

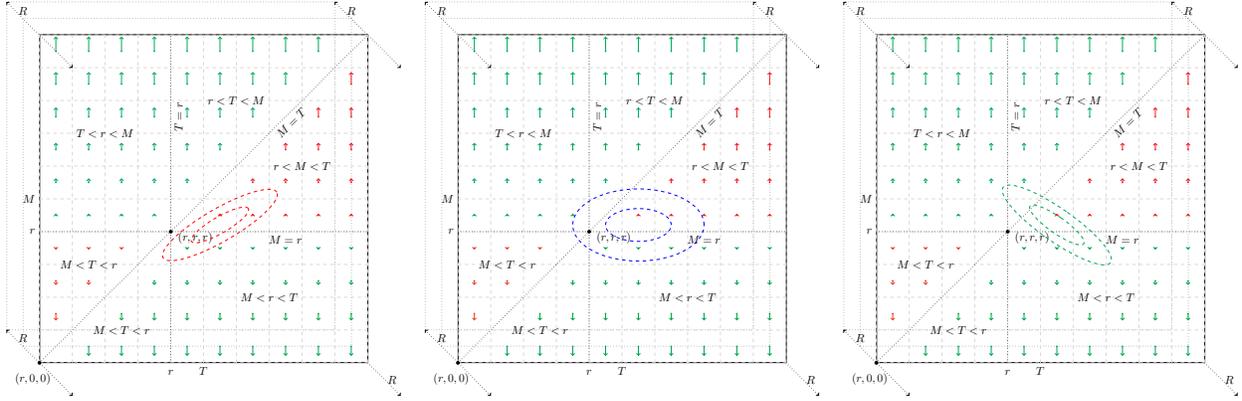
\begin{figure}

\centering
\resizebox{0.32\textwidth}{!}{
\begin{tikzpicture}
\node[scale=1] (grid) at (0,0){\begin{tikzpicture}

\node [dotted, draw=gray, shape=rectangle, minimum width=10cm, minimum height=10cm, anchor=south west] at (-1,1) {};
\node [dotted, draw=gray, shape=rectangle, minimum width=10cm, minimum height=10cm, anchor=south west] at (-0.5,0.5) {};
\node [draw, shape=rectangle, minimum width=10cm, minimum height=10cm, anchor=south west] at (0,0) {};

\draw[help lines, color=gray!30, dashed] (0,0) grid (10,10);

\foreach \x in {0.5, ..., 10} {
    \foreach \y in {0.5, ..., 9.5} {
        \pgfmathsetmacro{\vy}{0.08*(\y-4)}
            
        \ifthenelse{\lengthtest{\y pt = \x pt}}{}{
        \ifthenelse{\lengthtest{\y pt > \x pt}}{
        
        \ifthenelse{\lengthtest{\y pt < 4 pt}}{
        \node at (\x, \y)[circle, fill=Red, scale=0.15] {};
        \draw[->,color=Red!90] (\x,\y) -- (\x, \y+\vy);
        }{
        \node at (\x, \y)[circle, fill=Green, scale=0.15] {};
        \draw[->,color=Green!90] (\x,\y) -- (\x, \y+\vy);
        }

        }{
        
        \ifthenelse{\lengthtest{\y pt > 4 pt}}{
        \node at (\x, \y)[circle, fill=Red, scale=0.15] {};
        \draw[->,color=Red!90\x] (\x,\y) -- (\x, \y+\vy);
        }{
        \node at (\x, \y)[circle, fill=Green, scale=0.15] {};
        \draw[->,color=Green!90\x] (\x,\y) -- (\x, \y+\vy);
        }
        };
        }
    }
}

\draw [dotted] (0,0) edge node[dotted,near end,sloped,below] {$\makerRV=\takerRV$} (10,10);

\draw [dotted] (4,0) edge node[dotted,sloped,near end,below] {$\takerRV=\real$} (4,10);
\draw [dotted] (0,4) edge node[dotted,sloped,near end,below] {$\makerRV=\real$} (10,4);

\draw [dotted,<->] (-1,1) -- node[dotted,near start,above] {\textbf{$\realRV$}} (1,-1);
\draw [dotted,<->] (1,9) -- node[dotted,near end,above] {$\realRV$} (-1,11);
\draw [dotted,<->] (11,9) -- node[dotted,near end,above] {$\realRV$} (9,11);
\draw [dotted,->] (10,0) -- node[dotted,near end,above] {$\realRV$} (11,-1);

\node [draw=none,anchor=west,color=black] at (4.1,3.8) {$(\real,\real,\real)$};
\node at (4,4)[circle, fill=black, scale=0.35] {};
\node at (0,0)[circle, fill=black, scale=0.35] {};

\node [draw=none] at (2,7) {$\takerRV<\real<\makerRV$};
\node [draw=none] at (6,8) {$\real<\takerRV<\makerRV$};
\node [draw=none] at (8,6) {$\real<\makerRV<\takerRV$};
\node [draw=none] at (7,2) {$\makerRV<\real<\takerRV$};
\node [draw=none] at (2.5,1) {$\makerRV<\takerRV<\real$};
\node [draw=none] at (1.5,3) {$\makerRV<\takerRV<\real$};

\node [draw=none] at (4,-0.3) {$\real$};
\node [draw=none,anchor=east] at (0,4) {$\real$};

\node [draw=none,anchor=east] at (0,5) {\textbf{$\makerRV$}};
\node [draw=none,anchor=north] at (5,0) {\textbf{$\takerRV$}};

\node [draw=none] at (-0.2,-0.5) {$(\real,0,0)$};


\end{tikzpicture} };
\draw[rotate around={30:(.5,-.8)}, red, dashed] (.5,-.8) ellipse (2cm and 0.5cm);
\draw[rotate around={30:(.5,-.8)}, red, dashed] (.5,-.8) ellipse (1cm and 0.25cm);
\end{tikzpicture}
}
\resizebox{0.32\textwidth}{!}{
\begin{tikzpicture}
\node[scale=1] (grid) at (0,0){\begin{tikzpicture}

\node [dotted, draw=gray, shape=rectangle, minimum width=10cm, minimum height=10cm, anchor=south west] at (-1,1) {};
\node [dotted, draw=gray, shape=rectangle, minimum width=10cm, minimum height=10cm, anchor=south west] at (-0.5,0.5) {};
\node [draw, shape=rectangle, minimum width=10cm, minimum height=10cm, anchor=south west] at (0,0) {};

\draw[help lines, color=gray!30, dashed] (0,0) grid (10,10);

\foreach \x in {0.5, ..., 10} {
    \foreach \y in {0.5, ..., 9.5} {
        \pgfmathsetmacro{\vy}{0.08*(\y-4)}
            
        \ifthenelse{\lengthtest{\y pt = \x pt}}{}{
        \ifthenelse{\lengthtest{\y pt > \x pt}}{
        
        \ifthenelse{\lengthtest{\y pt < 4 pt}}{
        \node at (\x, \y)[circle, fill=Red, scale=0.15] {};
        \draw[->,color=Red!90] (\x,\y) -- (\x, \y+\vy);
        }{
        \node at (\x, \y)[circle, fill=Green, scale=0.15] {};
        \draw[->,color=Green!90] (\x,\y) -- (\x, \y+\vy);
        }

        }{
        
        \ifthenelse{\lengthtest{\y pt > 4 pt}}{
        \node at (\x, \y)[circle, fill=Red, scale=0.15] {};
        \draw[->,color=Red!90\x] (\x,\y) -- (\x, \y+\vy);
        }{
        \node at (\x, \y)[circle, fill=Green, scale=0.15] {};
        \draw[->,color=Green!90\x] (\x,\y) -- (\x, \y+\vy);
        }
        };
        }
    }
}

\draw [dotted] (0,0) edge node[dotted,near end,sloped,below] {$\makerRV=\takerRV$} (10,10);

\draw [dotted] (4,0) edge node[dotted,sloped,near end,below] {$\takerRV=\real$} (4,10);
\draw [dotted] (0,4) edge node[dotted,sloped,near end,below] {$\makerRV=\real$} (10,4);

\draw [dotted,<->] (-1,1) -- node[dotted,near start,above] {\textbf{$\realRV$}} (1,-1);
\draw [dotted,<->] (1,9) -- node[dotted,near end,above] {$\realRV$} (-1,11);
\draw [dotted,<->] (11,9) -- node[dotted,near end,above] {$\realRV$} (9,11);
\draw [dotted,->] (10,0) -- node[dotted,near end,above] {$\realRV$} (11,-1);

\node [draw=none,anchor=west,color=black] at (4.1,3.8) {$(\real,\real,\real)$};
\node at (4,4)[circle, fill=black, scale=0.35] {};
\node at (0,0)[circle, fill=black, scale=0.35] {};

\node [draw=none] at (2,7) {$\takerRV<\real<\makerRV$};
\node [draw=none] at (6,8) {$\real<\takerRV<\makerRV$};
\node [draw=none] at (8,6) {$\real<\makerRV<\takerRV$};
\node [draw=none] at (7,2) {$\makerRV<\real<\takerRV$};
\node [draw=none] at (2.5,1) {$\makerRV<\takerRV<\real$};
\node [draw=none] at (1.5,3) {$\makerRV<\takerRV<\real$};

\node [draw=none] at (4,-0.3) {$\real$};
\node [draw=none,anchor=east] at (0,4) {$\real$};

\node [draw=none,anchor=east] at (0,5) {\textbf{$\makerRV$}};
\node [draw=none,anchor=north] at (5,0) {\textbf{$\takerRV$}};

\node [draw=none] at (-0.2,-0.5) {$(\real,0,0)$};


\end{tikzpicture} };
\draw[rotate=0, blue, dashed] (.5,-.8) ellipse (2cm and 1.1cm);
\draw[rotate=0, blue, dashed] (.5,-.8) ellipse (1cm and .5cm);
\end{tikzpicture}
}
\resizebox{0.32\textwidth}{!}{
\begin{tikzpicture}
\node[scale=1] (grid) at (0,0){\begin{tikzpicture}

\node [dotted, draw=gray, shape=rectangle, minimum width=10cm, minimum height=10cm, anchor=south west] at (-1,1) {};
\node [dotted, draw=gray, shape=rectangle, minimum width=10cm, minimum height=10cm, anchor=south west] at (-0.5,0.5) {};
\node [draw, shape=rectangle, minimum width=10cm, minimum height=10cm, anchor=south west] at (0,0) {};

\draw[help lines, color=gray!30, dashed] (0,0) grid (10,10);

\foreach \x in {0.5, ..., 10} {
    \foreach \y in {0.5, ..., 9.5} {
        \pgfmathsetmacro{\vy}{0.08*(\y-4)}
            
        \ifthenelse{\lengthtest{\y pt = \x pt}}{}{
        \ifthenelse{\lengthtest{\y pt > \x pt}}{
        
        \ifthenelse{\lengthtest{\y pt < 4 pt}}{
        \node at (\x, \y)[circle, fill=Red, scale=0.15] {};
        \draw[->,color=Red!90] (\x,\y) -- (\x, \y+\vy);
        }{
        \node at (\x, \y)[circle, fill=Green, scale=0.15] {};
        \draw[->,color=Green!90] (\x,\y) -- (\x, \y+\vy);
        }

        }{
        
        \ifthenelse{\lengthtest{\y pt > 4 pt}}{
        \node at (\x, \y)[circle, fill=Red, scale=0.15] {};
        \draw[->,color=Red!90\x] (\x,\y) -- (\x, \y+\vy);
        }{
        \node at (\x, \y)[circle, fill=Green, scale=0.15] {};
        \draw[->,color=Green!90\x] (\x,\y) -- (\x, \y+\vy);
        }
        };
        }
    }
}

\draw [dotted] (0,0) edge node[dotted,near end,sloped,below] {$\makerRV=\takerRV$} (10,10);

\draw [dotted] (4,0) edge node[dotted,sloped,near end,below] {$\takerRV=\real$} (4,10);
\draw [dotted] (0,4) edge node[dotted,sloped,near end,below] {$\makerRV=\real$} (10,4);

\draw [dotted,<->] (-1,1) -- node[dotted,near start,above] {\textbf{$\realRV$}} (1,-1);
\draw [dotted,<->] (1,9) -- node[dotted,near end,above] {$\realRV$} (-1,11);
\draw [dotted,<->] (11,9) -- node[dotted,near end,above] {$\realRV$} (9,11);
\draw [dotted,->] (10,0) -- node[dotted,near end,above] {$\realRV$} (11,-1);

\node [draw=none,anchor=west,color=black] at (4.1,3.8) {$(\real,\real,\real)$};
\node at (4,4)[circle, fill=black, scale=0.35] {};
\node at (0,0)[circle, fill=black, scale=0.35] {};

\node [draw=none] at (2,7) {$\takerRV<\real<\makerRV$};
\node [draw=none] at (6,8) {$\real<\takerRV<\makerRV$};
\node [draw=none] at (8,6) {$\real<\makerRV<\takerRV$};
\node [draw=none] at (7,2) {$\makerRV<\real<\takerRV$};
\node [draw=none] at (2.5,1) {$\makerRV<\takerRV<\real$};
\node [draw=none] at (1.5,3) {$\makerRV<\takerRV<\real$};

\node [draw=none] at (4,-0.3) {$\real$};
\node [draw=none,anchor=east] at (0,4) {$\real$};

\node [draw=none,anchor=east] at (0,5) {\textbf{$\makerRV$}};
\node [draw=none,anchor=north] at (5,0) {\textbf{$\takerRV$}};

\node [draw=none] at (-0.2,-0.5) {$(\real,0,0)$};


\end{tikzpicture} };
\draw[rotate around={55:(.5,-.8)}, Green, dashed] (.5,-.8) ellipse (0.5cm and 2cm);
\draw[rotate around={55:(.5,-.8)}, Green, dashed] (.5,-.8) ellipse (0.25cm and 1cm);
\end{tikzpicture}
}
\caption{A sketch of the distribution of estimates density from Figure~\ref{fig:projection} for the case where the trader $\takerRV$ is biased more than the market maker $\makerRV$ w.r.t $\realRV=\real$. The decorrelation technique still commonly helps as larger parts of the profitable regions of the distribution are being covered by decreasing $\corr{\takerRV,\makerRV|\realRV}$.}
\label{fig:taker_biased}
\end{figure}

\subsection{Having a Superior Model}
\label{sec:superior_model}

The primary motivation behind the concept of lowering the correlation with the market is to make profits with models of inferior quality, which would commonly yield negative profits otherwise. We argued that such a situation is common, since the market (maker) price tends to be a very good estimate of the true value in any fairly efficient market (Section~\ref{sec:est_dist}). Nevertheless, for completeness, let us consider the opposite case of having a superior model, i.e. a situation where $err(\takerFunc) < err(\makerFunc)$. Following the statistical decomposition of estimation errors in terms of bias and variance (Section~\ref{sec:estimators}), let us separately consider two cases of such superiority through a model with (i) lower variance and (ii) lower bias.

\paragraph{Superior Variance}

For a superior model $\takerRV$ with a lower conditional variance $\var{\takerRV|\realRV} < \var{\makerRV|\realRV}$, the concept of decorrelation (Definition~\ref{def:decorrelation}) no longer works as a consistent profit enhancement, even for common, realistic market distributions. Particularly, decreasing the correlation $\corr{\takerRV,\makerRV|\realRV}$ can actually decrease the returns in many scenarios. We again depict the concept on an example elliptical distribution in Figure~\ref{fig:switched_variances}. The variances of the estimators are exactly opposite to those from Figure~\ref{fig:unbiased}. We can see that the returns with an independent model $\corr{\takerRV,\makerRV}=0$ are lower than those of a highly correlated model $\corr{\takerRV,\makerRV|\realRV}=1$. Note however that the decrease of returns in this case is smaller than the increase in the opposite case (i.e. shift from $\corr{\takerRV,\makerRV|\realRV}=1$ to $\corr{\takerRV,\takerRV|\realRV}=0$ in Figure~\ref{fig:unbiased}). Finally, the profits of $\corr{\takerRV,\makerRV|\realRV}=-1$ are still maximal.

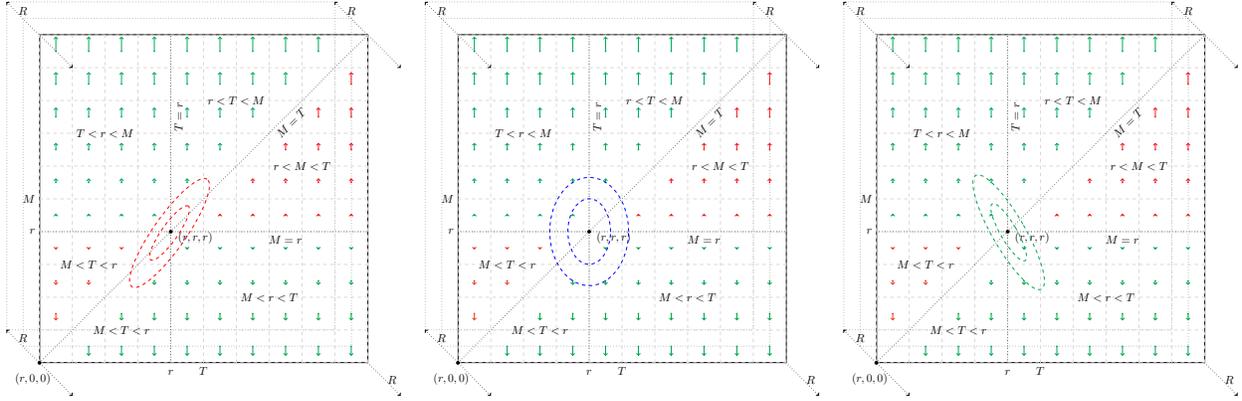
\begin{figure}

\centering
\resizebox{0.32\textwidth}{!}{
\begin{tikzpicture}
\node[scale=1] (grid) at (0,0){\begin{tikzpicture}

\node [dotted, draw=gray, shape=rectangle, minimum width=10cm, minimum height=10cm, anchor=south west] at (-1,1) {};
\node [dotted, draw=gray, shape=rectangle, minimum width=10cm, minimum height=10cm, anchor=south west] at (-0.5,0.5) {};
\node [draw, shape=rectangle, minimum width=10cm, minimum height=10cm, anchor=south west] at (0,0) {};

\draw[help lines, color=gray!30, dashed] (0,0) grid (10,10);

\foreach \x in {0.5, ..., 10} {
    \foreach \y in {0.5, ..., 9.5} {
        \pgfmathsetmacro{\vy}{0.08*(\y-4)}
            
        \ifthenelse{\lengthtest{\y pt = \x pt}}{}{
        \ifthenelse{\lengthtest{\y pt > \x pt}}{
        
        \ifthenelse{\lengthtest{\y pt < 4 pt}}{
        \node at (\x, \y)[circle, fill=Red, scale=0.15] {};
        \draw[->,color=Red!90] (\x,\y) -- (\x, \y+\vy);
        }{
        \node at (\x, \y)[circle, fill=Green, scale=0.15] {};
        \draw[->,color=Green!90] (\x,\y) -- (\x, \y+\vy);
        }

        }{
        
        \ifthenelse{\lengthtest{\y pt > 4 pt}}{
        \node at (\x, \y)[circle, fill=Red, scale=0.15] {};
        \draw[->,color=Red!90\x] (\x,\y) -- (\x, \y+\vy);
        }{
        \node at (\x, \y)[circle, fill=Green, scale=0.15] {};
        \draw[->,color=Green!90\x] (\x,\y) -- (\x, \y+\vy);
        }
        };
        }
    }
}

\draw [dotted] (0,0) edge node[dotted,near end,sloped,below] {$\makerRV=\takerRV$} (10,10);

\draw [dotted] (4,0) edge node[dotted,sloped,near end,below] {$\takerRV=\real$} (4,10);
\draw [dotted] (0,4) edge node[dotted,sloped,near end,below] {$\makerRV=\real$} (10,4);

\draw [dotted,<->] (-1,1) -- node[dotted,near start,above] {\textbf{$\realRV$}} (1,-1);
\draw [dotted,<->] (1,9) -- node[dotted,near end,above] {$\realRV$} (-1,11);
\draw [dotted,<->] (11,9) -- node[dotted,near end,above] {$\realRV$} (9,11);
\draw [dotted,->] (10,0) -- node[dotted,near end,above] {$\realRV$} (11,-1);

\node [draw=none,anchor=west,color=black] at (4.1,3.8) {$(\real,\real,\real)$};
\node at (4,4)[circle, fill=black, scale=0.35] {};
\node at (0,0)[circle, fill=black, scale=0.35] {};

\node [draw=none] at (2,7) {$\takerRV<\real<\makerRV$};
\node [draw=none] at (6,8) {$\real<\takerRV<\makerRV$};
\node [draw=none] at (8,6) {$\real<\makerRV<\takerRV$};
\node [draw=none] at (7,2) {$\makerRV<\real<\takerRV$};
\node [draw=none] at (2.5,1) {$\makerRV<\takerRV<\real$};
\node [draw=none] at (1.5,3) {$\makerRV<\takerRV<\real$};

\node [draw=none] at (4,-0.3) {$\real$};
\node [draw=none,anchor=east] at (0,4) {$\real$};

\node [draw=none,anchor=east] at (0,5) {\textbf{$\makerRV$}};
\node [draw=none,anchor=north] at (5,0) {\textbf{$\takerRV$}};

\node [draw=none] at (-0.2,-0.5) {$(\real,0,0)$};


\end{tikzpicture} };
\draw[rotate=55, red, dashed] (-1.45,0.25) ellipse (2cm and 0.5cm);
\draw[rotate=55, red, dashed] (-1.45,0.25) ellipse (1cm and 0.25cm);
\end{tikzpicture}
}
\resizebox{0.32\textwidth}{!}{
\begin{tikzpicture}
\node[scale=1] (grid) at (0,0){\begin{tikzpicture}

\node [dotted, draw=gray, shape=rectangle, minimum width=10cm, minimum height=10cm, anchor=south west] at (-1,1) {};
\node [dotted, draw=gray, shape=rectangle, minimum width=10cm, minimum height=10cm, anchor=south west] at (-0.5,0.5) {};
\node [draw, shape=rectangle, minimum width=10cm, minimum height=10cm, anchor=south west] at (0,0) {};

\draw[help lines, color=gray!30, dashed] (0,0) grid (10,10);

\foreach \x in {0.5, ..., 10} {
    \foreach \y in {0.5, ..., 9.5} {
        \pgfmathsetmacro{\vy}{0.08*(\y-4)}
            
        \ifthenelse{\lengthtest{\y pt = \x pt}}{}{
        \ifthenelse{\lengthtest{\y pt > \x pt}}{
        
        \ifthenelse{\lengthtest{\y pt < 4 pt}}{
        \node at (\x, \y)[circle, fill=Red, scale=0.15] {};
        \draw[->,color=Red!90] (\x,\y) -- (\x, \y+\vy);
        }{
        \node at (\x, \y)[circle, fill=Green, scale=0.15] {};
        \draw[->,color=Green!90] (\x,\y) -- (\x, \y+\vy);
        }

        }{
        
        \ifthenelse{\lengthtest{\y pt > 4 pt}}{
        \node at (\x, \y)[circle, fill=Red, scale=0.15] {};
        \draw[->,color=Red!90\x] (\x,\y) -- (\x, \y+\vy);
        }{
        \node at (\x, \y)[circle, fill=Green, scale=0.15] {};
        \draw[->,color=Green!90\x] (\x,\y) -- (\x, \y+\vy);
        }
        };
        }
    }
}

\draw [dotted] (0,0) edge node[dotted,near end,sloped,below] {$\makerRV=\takerRV$} (10,10);

\draw [dotted] (4,0) edge node[dotted,sloped,near end,below] {$\takerRV=\real$} (4,10);
\draw [dotted] (0,4) edge node[dotted,sloped,near end,below] {$\makerRV=\real$} (10,4);

\draw [dotted,<->] (-1,1) -- node[dotted,near start,above] {\textbf{$\realRV$}} (1,-1);
\draw [dotted,<->] (1,9) -- node[dotted,near end,above] {$\realRV$} (-1,11);
\draw [dotted,<->] (11,9) -- node[dotted,near end,above] {$\realRV$} (9,11);
\draw [dotted,->] (10,0) -- node[dotted,near end,above] {$\realRV$} (11,-1);

\node [draw=none,anchor=west,color=black] at (4.1,3.8) {$(\real,\real,\real)$};
\node at (4,4)[circle, fill=black, scale=0.35] {};
\node at (0,0)[circle, fill=black, scale=0.35] {};

\node [draw=none] at (2,7) {$\takerRV<\real<\makerRV$};
\node [draw=none] at (6,8) {$\real<\takerRV<\makerRV$};
\node [draw=none] at (8,6) {$\real<\makerRV<\takerRV$};
\node [draw=none] at (7,2) {$\makerRV<\real<\takerRV$};
\node [draw=none] at (2.5,1) {$\makerRV<\takerRV<\real$};
\node [draw=none] at (1.5,3) {$\makerRV<\takerRV<\real$};

\node [draw=none] at (4,-0.3) {$\real$};
\node [draw=none,anchor=east] at (0,4) {$\real$};

\node [draw=none,anchor=east] at (0,5) {\textbf{$\makerRV$}};
\node [draw=none,anchor=north] at (5,0) {\textbf{$\takerRV$}};

\node [draw=none] at (-0.2,-0.5) {$(\real,0,0)$};


\end{tikzpicture} };
\draw[rotate=0, blue, dashed] (-1.,-1) ellipse (1.2cm and 1.65cm);
\draw[rotate=0, blue, dashed] (-1.,-1) ellipse (0.65cm and 1cm);
\end{tikzpicture}
}
\resizebox{0.32\textwidth}{!}{
\begin{tikzpicture}
\node[scale=1] (grid) at (0,0){\begin{tikzpicture}

\node [dotted, draw=gray, shape=rectangle, minimum width=10cm, minimum height=10cm, anchor=south west] at (-1,1) {};
\node [dotted, draw=gray, shape=rectangle, minimum width=10cm, minimum height=10cm, anchor=south west] at (-0.5,0.5) {};
\node [draw, shape=rectangle, minimum width=10cm, minimum height=10cm, anchor=south west] at (0,0) {};

\draw[help lines, color=gray!30, dashed] (0,0) grid (10,10);

\foreach \x in {0.5, ..., 10} {
    \foreach \y in {0.5, ..., 9.5} {
        \pgfmathsetmacro{\vy}{0.08*(\y-4)}
            
        \ifthenelse{\lengthtest{\y pt = \x pt}}{}{
        \ifthenelse{\lengthtest{\y pt > \x pt}}{
        
        \ifthenelse{\lengthtest{\y pt < 4 pt}}{
        \node at (\x, \y)[circle, fill=Red, scale=0.15] {};
        \draw[->,color=Red!90] (\x,\y) -- (\x, \y+\vy);
        }{
        \node at (\x, \y)[circle, fill=Green, scale=0.15] {};
        \draw[->,color=Green!90] (\x,\y) -- (\x, \y+\vy);
        }

        }{
        
        \ifthenelse{\lengthtest{\y pt > 4 pt}}{
        \node at (\x, \y)[circle, fill=Red, scale=0.15] {};
        \draw[->,color=Red!90\x] (\x,\y) -- (\x, \y+\vy);
        }{
        \node at (\x, \y)[circle, fill=Green, scale=0.15] {};
        \draw[->,color=Green!90\x] (\x,\y) -- (\x, \y+\vy);
        }
        };
        }
    }
}

\draw [dotted] (0,0) edge node[dotted,near end,sloped,below] {$\makerRV=\takerRV$} (10,10);

\draw [dotted] (4,0) edge node[dotted,sloped,near end,below] {$\takerRV=\real$} (4,10);
\draw [dotted] (0,4) edge node[dotted,sloped,near end,below] {$\makerRV=\real$} (10,4);

\draw [dotted,<->] (-1,1) -- node[dotted,near start,above] {\textbf{$\realRV$}} (1,-1);
\draw [dotted,<->] (1,9) -- node[dotted,near end,above] {$\realRV$} (-1,11);
\draw [dotted,<->] (11,9) -- node[dotted,near end,above] {$\realRV$} (9,11);
\draw [dotted,->] (10,0) -- node[dotted,near end,above] {$\realRV$} (11,-1);

\node [draw=none,anchor=west,color=black] at (4.1,3.8) {$(\real,\real,\real)$};
\node at (4,4)[circle, fill=black, scale=0.35] {};
\node at (0,0)[circle, fill=black, scale=0.35] {};

\node [draw=none] at (2,7) {$\takerRV<\real<\makerRV$};
\node [draw=none] at (6,8) {$\real<\takerRV<\makerRV$};
\node [draw=none] at (8,6) {$\real<\makerRV<\takerRV$};
\node [draw=none] at (7,2) {$\makerRV<\real<\takerRV$};
\node [draw=none] at (2.5,1) {$\makerRV<\takerRV<\real$};
\node [draw=none] at (1.5,3) {$\makerRV<\takerRV<\real$};

\node [draw=none] at (4,-0.3) {$\real$};
\node [draw=none,anchor=east] at (0,4) {$\real$};

\node [draw=none,anchor=east] at (0,5) {\textbf{$\makerRV$}};
\node [draw=none,anchor=north] at (5,0) {\textbf{$\takerRV$}};

\node [draw=none] at (-0.2,-0.5) {$(\real,0,0)$};


\end{tikzpicture} };
\draw[rotate=30, Green, dashed] (-1.35,-0.4) ellipse (0.5cm and 2cm) ;
\draw[rotate=30, Green, dashed] (-1.35,-0.4) ellipse (0.25cm and 1cm) ;
\end{tikzpicture}
}
\caption{A sketch of the distribution of estimates density from Figure~\ref{fig:projection} for the case of unbiased estimators $\makerRV$ and $\takerRV$ w.r.t $\realRV=\real$, where $\var{\takerRV|\realRV} < \var{\makerRV|\realRV}$, corresponding to a superior estimation model $\takerFunc$ of the trader. For a superior model in terms of the conditional variance, the correlation with the market does not pose such a problem as the model is already profitable, and decreasing it might actually hurt the performance, nevertheless, a model with the lowest covariance still generally provides the highest profitability.}
\label{fig:switched_variances}
\end{figure}

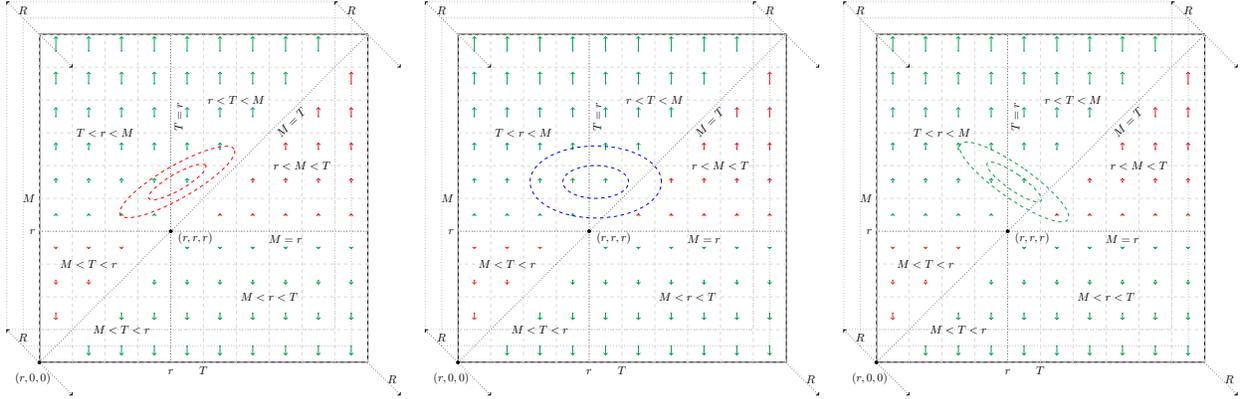
\begin{figure}

\centering
\resizebox{0.32\textwidth}{!}{
\begin{tikzpicture}
\node[scale=1] (grid) at (0,0){\begin{tikzpicture}

\node [dotted, draw=gray, shape=rectangle, minimum width=10cm, minimum height=10cm, anchor=south west] at (-1,1) {};
\node [dotted, draw=gray, shape=rectangle, minimum width=10cm, minimum height=10cm, anchor=south west] at (-0.5,0.5) {};
\node [draw, shape=rectangle, minimum width=10cm, minimum height=10cm, anchor=south west] at (0,0) {};

\draw[help lines, color=gray!30, dashed] (0,0) grid (10,10);

\foreach \x in {0.5, ..., 10} {
    \foreach \y in {0.5, ..., 9.5} {
        \pgfmathsetmacro{\vy}{0.08*(\y-4)}
            
        \ifthenelse{\lengthtest{\y pt = \x pt}}{}{
        \ifthenelse{\lengthtest{\y pt > \x pt}}{
        
        \ifthenelse{\lengthtest{\y pt < 4 pt}}{
        \node at (\x, \y)[circle, fill=Red, scale=0.15] {};
        \draw[->,color=Red!90] (\x,\y) -- (\x, \y+\vy);
        }{
        \node at (\x, \y)[circle, fill=Green, scale=0.15] {};
        \draw[->,color=Green!90] (\x,\y) -- (\x, \y+\vy);
        }

        }{
        
        \ifthenelse{\lengthtest{\y pt > 4 pt}}{
        \node at (\x, \y)[circle, fill=Red, scale=0.15] {};
        \draw[->,color=Red!90\x] (\x,\y) -- (\x, \y+\vy);
        }{
        \node at (\x, \y)[circle, fill=Green, scale=0.15] {};
        \draw[->,color=Green!90\x] (\x,\y) -- (\x, \y+\vy);
        }
        };
        }
    }
}

\draw [dotted] (0,0) edge node[dotted,near end,sloped,below] {$\makerRV=\takerRV$} (10,10);

\draw [dotted] (4,0) edge node[dotted,sloped,near end,below] {$\takerRV=\real$} (4,10);
\draw [dotted] (0,4) edge node[dotted,sloped,near end,below] {$\makerRV=\real$} (10,4);

\draw [dotted,<->] (-1,1) -- node[dotted,near start,above] {\textbf{$\realRV$}} (1,-1);
\draw [dotted,<->] (1,9) -- node[dotted,near end,above] {$\realRV$} (-1,11);
\draw [dotted,<->] (11,9) -- node[dotted,near end,above] {$\realRV$} (9,11);
\draw [dotted,->] (10,0) -- node[dotted,near end,above] {$\realRV$} (11,-1);

\node [draw=none,anchor=west,color=black] at (4.1,3.8) {$(\real,\real,\real)$};
\node at (4,4)[circle, fill=black, scale=0.35] {};
\node at (0,0)[circle, fill=black, scale=0.35] {};

\node [draw=none] at (2,7) {$\takerRV<\real<\makerRV$};
\node [draw=none] at (6,8) {$\real<\takerRV<\makerRV$};
\node [draw=none] at (8,6) {$\real<\makerRV<\takerRV$};
\node [draw=none] at (7,2) {$\makerRV<\real<\takerRV$};
\node [draw=none] at (2.5,1) {$\makerRV<\takerRV<\real$};
\node [draw=none] at (1.5,3) {$\makerRV<\takerRV<\real$};

\node [draw=none] at (4,-0.3) {$\real$};
\node [draw=none,anchor=east] at (0,4) {$\real$};

\node [draw=none,anchor=east] at (0,5) {\textbf{$\makerRV$}};
\node [draw=none,anchor=north] at (5,0) {\textbf{$\takerRV$}};

\node [draw=none] at (-0.2,-0.5) {$(\real,0,0)$};


\end{tikzpicture} };
\draw[rotate around={30:(-.8,0.5)}, red, dashed] (-.8,0.5) ellipse (2cm and 0.5cm);
\draw[rotate around={30:(-.8,0.5)}, red, dashed] (-.8,0.5) ellipse (1cm and 0.25cm);
\end{tikzpicture}
}
\resizebox{0.32\textwidth}{!}{
\begin{tikzpicture}
\node[scale=1] (grid) at (0,0){\begin{tikzpicture}

\node [dotted, draw=gray, shape=rectangle, minimum width=10cm, minimum height=10cm, anchor=south west] at (-1,1) {};
\node [dotted, draw=gray, shape=rectangle, minimum width=10cm, minimum height=10cm, anchor=south west] at (-0.5,0.5) {};
\node [draw, shape=rectangle, minimum width=10cm, minimum height=10cm, anchor=south west] at (0,0) {};

\draw[help lines, color=gray!30, dashed] (0,0) grid (10,10);

\foreach \x in {0.5, ..., 10} {
    \foreach \y in {0.5, ..., 9.5} {
        \pgfmathsetmacro{\vy}{0.08*(\y-4)}
            
        \ifthenelse{\lengthtest{\y pt = \x pt}}{}{
        \ifthenelse{\lengthtest{\y pt > \x pt}}{
        
        \ifthenelse{\lengthtest{\y pt < 4 pt}}{
        \node at (\x, \y)[circle, fill=Red, scale=0.15] {};
        \draw[->,color=Red!90] (\x,\y) -- (\x, \y+\vy);
        }{
        \node at (\x, \y)[circle, fill=Green, scale=0.15] {};
        \draw[->,color=Green!90] (\x,\y) -- (\x, \y+\vy);
        }

        }{
        
        \ifthenelse{\lengthtest{\y pt > 4 pt}}{
        \node at (\x, \y)[circle, fill=Red, scale=0.15] {};
        \draw[->,color=Red!90\x] (\x,\y) -- (\x, \y+\vy);
        }{
        \node at (\x, \y)[circle, fill=Green, scale=0.15] {};
        \draw[->,color=Green!90\x] (\x,\y) -- (\x, \y+\vy);
        }
        };
        }
    }
}

\draw [dotted] (0,0) edge node[dotted,near end,sloped,below] {$\makerRV=\takerRV$} (10,10);

\draw [dotted] (4,0) edge node[dotted,sloped,near end,below] {$\takerRV=\real$} (4,10);
\draw [dotted] (0,4) edge node[dotted,sloped,near end,below] {$\makerRV=\real$} (10,4);

\draw [dotted,<->] (-1,1) -- node[dotted,near start,above] {\textbf{$\realRV$}} (1,-1);
\draw [dotted,<->] (1,9) -- node[dotted,near end,above] {$\realRV$} (-1,11);
\draw [dotted,<->] (11,9) -- node[dotted,near end,above] {$\realRV$} (9,11);
\draw [dotted,->] (10,0) -- node[dotted,near end,above] {$\realRV$} (11,-1);

\node [draw=none,anchor=west,color=black] at (4.1,3.8) {$(\real,\real,\real)$};
\node at (4,4)[circle, fill=black, scale=0.35] {};
\node at (0,0)[circle, fill=black, scale=0.35] {};

\node [draw=none] at (2,7) {$\takerRV<\real<\makerRV$};
\node [draw=none] at (6,8) {$\real<\takerRV<\makerRV$};
\node [draw=none] at (8,6) {$\real<\makerRV<\takerRV$};
\node [draw=none] at (7,2) {$\makerRV<\real<\takerRV$};
\node [draw=none] at (2.5,1) {$\makerRV<\takerRV<\real$};
\node [draw=none] at (1.5,3) {$\makerRV<\takerRV<\real$};

\node [draw=none] at (4,-0.3) {$\real$};
\node [draw=none,anchor=east] at (0,4) {$\real$};

\node [draw=none,anchor=east] at (0,5) {\textbf{$\makerRV$}};
\node [draw=none,anchor=north] at (5,0) {\textbf{$\takerRV$}};

\node [draw=none] at (-0.2,-0.5) {$(\real,0,0)$};


\end{tikzpicture} };
\draw[rotate=0, blue, dashed] (-.8,0.5) ellipse (2cm and 1.1cm);
\draw[rotate=0, blue, dashed] (-.8,0.5) ellipse (1cm and .5cm);
\end{tikzpicture}
}
\resizebox{0.32\textwidth}{!}{
\begin{tikzpicture}
\node[scale=1] (grid) at (0,0){\begin{tikzpicture}

\node [dotted, draw=gray, shape=rectangle, minimum width=10cm, minimum height=10cm, anchor=south west] at (-1,1) {};
\node [dotted, draw=gray, shape=rectangle, minimum width=10cm, minimum height=10cm, anchor=south west] at (-0.5,0.5) {};
\node [draw, shape=rectangle, minimum width=10cm, minimum height=10cm, anchor=south west] at (0,0) {};

\draw[help lines, color=gray!30, dashed] (0,0) grid (10,10);

\foreach \x in {0.5, ..., 10} {
    \foreach \y in {0.5, ..., 9.5} {
        \pgfmathsetmacro{\vy}{0.08*(\y-4)}
            
        \ifthenelse{\lengthtest{\y pt = \x pt}}{}{
        \ifthenelse{\lengthtest{\y pt > \x pt}}{
        
        \ifthenelse{\lengthtest{\y pt < 4 pt}}{
        \node at (\x, \y)[circle, fill=Red, scale=0.15] {};
        \draw[->,color=Red!90] (\x,\y) -- (\x, \y+\vy);
        }{
        \node at (\x, \y)[circle, fill=Green, scale=0.15] {};
        \draw[->,color=Green!90] (\x,\y) -- (\x, \y+\vy);
        }

        }{
        
        \ifthenelse{\lengthtest{\y pt > 4 pt}}{
        \node at (\x, \y)[circle, fill=Red, scale=0.15] {};
        \draw[->,color=Red!90\x] (\x,\y) -- (\x, \y+\vy);
        }{
        \node at (\x, \y)[circle, fill=Green, scale=0.15] {};
        \draw[->,color=Green!90\x] (\x,\y) -- (\x, \y+\vy);
        }
        };
        }
    }
}

\draw [dotted] (0,0) edge node[dotted,near end,sloped,below] {$\makerRV=\takerRV$} (10,10);

\draw [dotted] (4,0) edge node[dotted,sloped,near end,below] {$\takerRV=\real$} (4,10);
\draw [dotted] (0,4) edge node[dotted,sloped,near end,below] {$\makerRV=\real$} (10,4);

\draw [dotted,<->] (-1,1) -- node[dotted,near start,above] {\textbf{$\realRV$}} (1,-1);
\draw [dotted,<->] (1,9) -- node[dotted,near end,above] {$\realRV$} (-1,11);
\draw [dotted,<->] (11,9) -- node[dotted,near end,above] {$\realRV$} (9,11);
\draw [dotted,->] (10,0) -- node[dotted,near end,above] {$\realRV$} (11,-1);

\node [draw=none,anchor=west,color=black] at (4.1,3.8) {$(\real,\real,\real)$};
\node at (4,4)[circle, fill=black, scale=0.35] {};
\node at (0,0)[circle, fill=black, scale=0.35] {};

\node [draw=none] at (2,7) {$\takerRV<\real<\makerRV$};
\node [draw=none] at (6,8) {$\real<\takerRV<\makerRV$};
\node [draw=none] at (8,6) {$\real<\makerRV<\takerRV$};
\node [draw=none] at (7,2) {$\makerRV<\real<\takerRV$};
\node [draw=none] at (2.5,1) {$\makerRV<\takerRV<\real$};
\node [draw=none] at (1.5,3) {$\makerRV<\takerRV<\real$};

\node [draw=none] at (4,-0.3) {$\real$};
\node [draw=none,anchor=east] at (0,4) {$\real$};

\node [draw=none,anchor=east] at (0,5) {\textbf{$\makerRV$}};
\node [draw=none,anchor=north] at (5,0) {\textbf{$\takerRV$}};

\node [draw=none] at (-0.2,-0.5) {$(\real,0,0)$};


\end{tikzpicture} };
\draw[rotate around={55:(-.8,0.5)}, Green, dashed] (-.8,0.5) ellipse (0.5cm and 2cm);
\draw[rotate around={55:(-.8,0.5)}, Green, dashed] (-.8,0.5) ellipse (0.25cm and 1cm);
\end{tikzpicture}
}
\caption{A sketch of the distribution of estimates density from Figure~\ref{fig:projection} for the case where the market maker $\makerRV$ is biased more than the trader $\takerRV$  w.r.t $\realRV=\real$. While having such a superior model, being correlated with the market poses no problem to profit, and decreasing the correlation will commonly hurt the performance.}
\label{fig:superior_model}
\end{figure}

\paragraph{Superior Bias}

We have argued for the practical necessity of a market maker not to be systematically biased in Section~\ref{sec:est_dist}, which leaves a little space for the trader to beat the market in terms of $\bias{\takerRV}<\bias{\makerRV}$. Nevertheless it should be acknowledged that in the unlikely case that the market maker indeed is more biased than the trader, the concept of decorrelating the estimates for increased profits breaks down severely. The situation is depicted in Figure~\ref{fig:superior_model}. We can see that in this setting, decreasing $\corr{\takerRV,\makerRV|\realRV}$ can easily work in a directly counterproductive fashion by consistently decreasing the profits.

Importantly, however, it should be noted that with a model that is superior by either means, i.e. where $err(\takerFunc) < err(\makerFunc)$, there is no need in trying to decrease $\corr{\takerRV,\makerRV|\realRV}$ to make profits, since such a model $\takerFunc$ needs no help with that to begin with. This can be conveniently detected in advance by measuring $err(\takerFunc)$ and trading the model with standard investment strategies (Section~\ref{sec:strategies}).




\subsection{The Problem with Kelly}
\label{sec:kelly_problem}


The scenarios we have demonstrated so far operated with the simple uniform investment strategy (Section~\ref{sec:strategies}), which allowed us to generate profits through decorrelation (Definition~\ref{def:decorrelation}), even with models of inferior accuracy w.r.t. market. As indicated in Section~\ref{sec:acc_profit}, let us now explain why this cannot be done with the, more sophisticated, Kelly investment strategy. 

We have shown that the growth of wealth $W_{\sf{G}}$ with the Kelly strategy is directly equal to difference between the KL-divergence of the market maker from the true distribution $D_{KL}(\makerRV||\realRV)$ and the trader from the true distribution $D_{KL}(\takerRV||\realRV)$, respectively, in Equation~\ref{eq:KL_diff}. It follows that a Kelly trader does not care at all about the particular relationships between $\realRV,\makerRV,\takerRV$ which are essential to profitability (Section~\ref{sec:essence}), since the only thing that matters is how close are our estimates to the true distribution as compared to the bookmaker. The two principally different views of the market distribution $\mathrm{P}_\Omega$ properties are depicted in Figure~\ref{fig:triangle}.

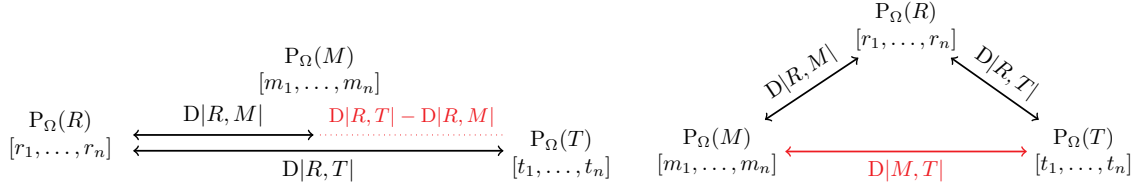
\begin{figure}

\centering
\resizebox{0.5\textwidth}{!}{
\begin{tikzpicture}
\begin{scope}[yshift=4cm]

\node[color=black, align=center] (Pr) [] {$\mathrm{P}_{\Omega}(\realRV)$ \\ $[{\real_1}, \dots, {\real_n}]$};
\node[color=black, align=center] (Pb) [above right= 0.51cm and 3cm] {$\mathrm{P}_{\Omega}(\makerRV)$ \\ $[{\maker_1}, \dots, {\maker_n}]$};
\node[color=black, align=center] (Pp) [below right= -0.21cm and 7cm] {$\mathrm{P}_\Omega(\takerRV)$ \\ $[{\taker_1}, \dots, {\taker_n}]$};

\draw [thick,<->] (1.15,0) -- node[above,sloped] { $\mathrm{D}|\realRV,\makerRV|$} (4,0);
\draw [dotted,Red] (4.1,0) -- node[above,sloped] { \small{$\mathrm{D}|\realRV,\takerRV| - \mathrm{D}|\realRV,\makerRV|$}} (7,0);

\draw [thick,<->] (1.15,-0.25) -- node[below,sloped] { $\mathrm{D}|\realRV,\takerRV|$} (7,-0.25);

\end{scope}
\end{tikzpicture}
}
\resizebox{0.4\textwidth}{!}{
\begin{tikzpicture}[scale=0.7]
\node[color=black, align=center] (Pr) [] {$\mathrm{P}_{\Omega}(\realRV)$ \\ $[{\real_1}, \dots, {\real_n}]$};
\node[color=black, align=center] (Pb) [below left= 1cm and 1cm of Pr] {$\mathrm{P}_{\Omega}(\makerRV)$ \\ $[{\maker_1}, \dots, {\maker_n}]$};
\node[color=black, align=center] (Pp) [below right= 1cm and 1cm of Pr] {$\mathrm{P}_{\Omega}(\takerRV)$ \\ $[{\taker_1}, \dots, {\taker_n}]$};

\draw [thick,<->] (Pr) edge node[above,sloped] { $\mathrm{D}|\realRV,\makerRV|$} (Pb);
\draw [thick,<->] (Pr) edge node[above,sloped] { $\mathrm{D}|\realRV,\takerRV|$} (Pp);
\draw [thick,<->, Red] (Pb) edge node[below,sloped] { $\mathrm{D}|\makerRV,\takerRV|$} (Pp);
\end{tikzpicture}
}
\caption{The difference between a typical statistical treatment of quality of the estimators, where the relationship between $\makerRV$ and $\takerRV$ is considered merely in the terms of their distances to $\realRV$ (left), and the proposed scenario, where we explicitly consider their mutual statistical relationship (right).}
\label{fig:triangle}
\end{figure}

Consequently, it thus does not matter whether the opportunity seems to have a positive or negative expected profit, since an optimal Kelly trader will bet an amount derived merely from the expected growth corresponding to the information advantage w.r.t. $\realRV$. Given the knowledge of the true value distribution $\mathrm{P}_{\Omega}(R)$, this provably leads to more wealth than any other strategy. It is tempting to believe that it generally provides an upper bound to the amount of profit that can be made in any scenario. Nevertheless, such property cannot be extrapolated to the settings where we lack knowledge of the true probabilities, such as in the real world trading, which we demonstrate on the following example.

\begin{example}
\label{ex:frac_kelly}
Consider a simple scenario of Kelly betting on a match (asset) $a_i$ with two equally probable exclusive outcomes $\{home, away\}$, and the two corresponding opportunities $\omega_i^\alpha,\omega_i^\beta$ priced by the bookmaker $\maker$ and the trader $\taker$ equally as follows


\begin{equation}
\label{eq:neutral_opp}
    \realFunc(\omega_i) = \begin{cases}
            0.5  \text{~~on $\alpha$=home}\\
            0.5 \text{~~on $\beta$=away}
    \end{cases}
    ~~~~~~~~~\text{}
    \makerFunc(\omega_i) = \begin{cases}
            0.3  \text{~~on $\alpha$=home}\\
            0.7 \text{~~on $\beta$=away}
    \end{cases}
    ~~~~~~~~~\text{}
    \takerFunc(\omega_i) = \begin{cases}
            0.3  \text{~~on $\alpha$=home}\\
            0.7 \text{~~on $\beta$=away}
    \end{cases}
\end{equation}
In this example, we again omit the bookmaker's margin for simplicity and so, following the derivation from Section~\ref{sec:acc_profit}, the optimal vector of fractions $\bm{f}$ to bet on the outcomes would be

\begin{equation}
    f= \begin{cases}
            0.3  \text{~~on $\alpha$=home}\\
            0.7  \text{~~on $\beta$=away}
    \end{cases}
\end{equation}

Since the two estimates of $\maker$ and $\taker$ coincide, there is clearly no information advantage of the trader and consequently zero profit to be made with Kelly. And the situation is the same for any other strategy, too, since the expected profitability of the opportunities, following the expected profit definition from Equation~\ref{eq:profit_bet}, from the perspective of the trader is simply zero

\begin{align}
    \EX_\takerFunc [\return_i^\alpha] = \frac{0.3}{0.3}-1 = 0 &~~~~~~~& \EX_\takerFunc [\return_i^\beta] = \frac{0.7}{0.7}-1 = 0
\end{align}

Now consider altering the scenario by decorrelating the estimates of the trader $\taker$ as follows
\begin{equation}
\label{eq:profitable_opp}
    \realFunc(\omega_i) = \begin{cases}
            0.5  \text{~~on $\alpha$=home}\\
            0.5 \text{~~on $\beta$=away}
    \end{cases}
    ~~~~~~~~~\text{}
    \makerFunc(\omega_i) = \begin{cases}
            0.3  \text{~~on $\alpha$=home}\\
            0.7 \text{~~on $\beta$=away}
    \end{cases}
    ~~~~~~~~~\text{}
    \takerFunc(\omega_i) = \begin{cases}
            0.7  \text{~~on $\alpha$=home}\\
            0.3 \text{~~on $\beta$=away}
    \end{cases}
\end{equation}
Despite switching the estimates, we can clearly see that both the bookmaker $\maker$ and the trader $\taker$ are still equally distanced from the true value $\real$ (by the means of $D_{KL}$ as well as any other possible metric), leading again to no information advantage and, as expected, to a zero expected growth of wealth (Section~\ref{sec:back_Kelly}):

\begin{equation}
    W_{\sf{G}} = \frac{1}{t}\cdot log(\frac{W_t}{W_0}) = \sum_i \realVal \cdot log (\frac{f_i}{\makerVal}) = 0.5 \cdot log (\frac{0.7}{0.3}) + 0.5 \cdot log (0.5 + \frac{3}{0.7}) = 0
\end{equation}
Nevertheless, the essential profitability of the opportunities from the perspective of the trader is now

\begin{align}
    \EX_\takerFunc [\return_i^\alpha] = \frac{0.7}{0.3}-1 = 1.33 &~~~~~~~& \EX_\takerFunc [\return_i^\beta] = \frac{1-\EX[{m}(o_2)]}{1-b(o_2)} - 1 = \frac{0.3}{0.7}-1 = -0.57
\end{align}
and betting uniformly (Section~\ref{sec:strat:unit}) some unit on the first opportunity $\alpha=home$, the trader would make a consistent profit of

\begin{align}
    \EX_\realFunc [\return_i^\alpha] = \frac{\realVal}{\makerFunc(\omega_i^\alpha)} - 1 = \frac{0.5}{0.3}-1 = 0.66
\end{align}
despite it being different from her estimated $\return_i^\alpha=1.33$.
\end{example}

The inability of Kelly to make profit in such profitable scenarios follows from its growth-based view of optimal investments (Section~\ref{sec:back_Kelly}). While maximizing the growth $W_{\sf{G}}$, less than optimal investments are just as harmful as over-investment, despite the latter leading to definite ruin while the former only leads to sub-optimal growth. To distinguish between the two arguably different types of risk, various modifications of the Kelly criterion have been proposed to reflect the natural preference for avoiding the ruin at the cost of sub-optimal growth.

\subsubsection{Fractional Kelly}
\label{sec:kelly_problem_frac}

Perhaps the most common remedy to mitigate the risk stemming from the erroneous estimates is fractional Kelly (Section~\ref{sec:back_Kelly}). Let us demonstrate the effect of this risk management practice on the introduced setting from Example~\ref{ex:frac_kelly} as follows.

\begin{example}
Being aware of the uncertainty in her estimates, the Kelly trader now decreases the optimal fraction by one-half, popularly referred to as ``half-Kelly'' betting. Considering the first scenario of coincidental estimates from Equation~\ref{eq:neutral_opp}, the invested fractions are now thus decreased by half as

\begin{equation}
    f= \begin{cases}
            0.15  \text{~~on $\alpha$=home}\\
            0.35  \text{~~on $\beta$=away}
    \end{cases}
\end{equation}
However, the growth of wealth, accounting for the half of it being held separately, stays inert at zero since

\begin{equation}
    W_{\sf{G}} = \frac{1}{t}\cdot log(\frac{W_t}{W_0}) = \sum_i \realVal \cdot log (0.5 + \frac{f_i}{\makerVal}) = 0.5 \cdot log (0.5 + \frac{0.15}{0.3}) + 0.5 \cdot log (0.5 + \frac{0.35}{0.7}) = 0
\end{equation}
In words, the trader is still under-betting the the first opportunity $\alpha=home$, which is profitable, while putting a larger amount on the second opportunity $\beta=away$, which is loss-making.
Decreasing the fractions then cannot change the simple fact that the opportunity is not recognized correctly by the estimator $\takerFunc$ w.r.t. essential profitability (Definition~\ref{def:essence}).

Consider now the latter scenario of the decorrelated estimates laid out in Equation~\ref{eq:profitable_opp}. The invested fractions are again cut by half as

\begin{equation}
    f= \begin{cases}
            0.35  \text{~~on $\alpha$=home}\\
            0.15  \text{~~on $\beta$=away}
    \end{cases}
\end{equation}
While the information advantage and all the properties of the estimators stay the same as in Equation~\ref{eq:profitable_opp}, the growth of wealth, accounting for the half of it being held separately, is now positive, particularly
\begin{equation}
    W_{\sf{G}} = \frac{1}{t}\cdot log(\frac{W_t}{W_0}) = \sum_i \realVal \cdot log (0.5 + \frac{f_i}{\makerVal}) = 0.5 \cdot log (0.5 + \frac{0.35}{0.3}) + 0.5 \cdot log (0.5 + \frac{0.15}{0.7}) = 0.038
\end{equation}
In words, decreasing the amount invested into the first opportunity $\alpha=home$ from $0.7$ down to $0.3$, which is now below the optimal $f_\alpha^*=\real_i^\alpha=0.5$ for the full Kelly fraction, removes the overbetting problem, while keeping the second fraction on the loss-making opportunity $\omega_i^\beta$ comparably low.

\end{example}

While the full Kelly forms an interesting corner case, oblivious to the correlation of the estimator w.r.t. the market, we have demonstrated that the decorrelation concept has a positive impact on the commonly used fractional Kelly modification, where decreasing the correlation $\corr{\makerRV,\takerRV|\realRV}$ consistently increases profits given that other properties of the distribution stay the same.
Similarly, it also improves profitability of other common portfolio optimization strategies, such as the MPT, which we demonstrate practically in experiments~(Section~\ref{sec:experiments}).



\subsection{Machine Learning}
\label{sec:machine_learning}

So far, we have analysed profitability of the price estimators $\takerFunc$ w.r.t. various market distributions $\mathrm{P}_{\Omega}$ and investment strategies $\sf{s}$ to derive the desired decorrelation property. Nevertheless, we have not yet discussed how to actually create such profitable estimators $\takerFunc$. As outlined in Section~\ref{sec:estimators}, the common way to create a fundamental price estimator $\sf{e}$ is by optimizing its fit to historical data $\mathcal{D}$ via minimization of some error measure $err(\sf{e})$. The standard desideratum is then to predict the real values $R$ within $\mathrm{P}_{\Omega}$ as closely as possible\footnote{while keeping the model complexity reasonably low for good generalization.}.

However, we now know that being accurate is not the only possible desideratum of the estimator, and that for the subsequent trading it might be even more important to minimize the conditional correlation with the market $\corr{\makerRV,\takerRV|\realRV}$. Note that an optimal trade-off between these properties of an estimator will naturally depend on the subsequent investment strategy being used (Section~\ref{sec:strategies}). For instance, we have shown that in the corner case of the full Kelly investor, minimizing the correlation has no impact on the profits at all (Section~\ref{sec:kelly_problem}) and to maximize the Kelly growth (Section~\ref{sec:back_Kelly}) it is sufficient to resort to the plain cross-entropy error (Section~\ref{sec:estimators}), minimizing the KL-distance from the true value distribution (Section~\ref{sec:acc_profit}). Nevertheless for the uniform investment strategy (Section~\ref{sec:proof}) as well as other practical strategies, minimizing the covariance will generally improve profitability (Section~\ref{sec:kelly_problem_frac}).

Without discussing the optimal trade-off, we can roughly conclude that it is generally advisable to strive for an unbiased estimator with low variance and low covariance with the market. One can then alter the standard machine learning objectives to reflect these new desiderata, and validate the optimal setting experimentally. The scope of the corresponding error measure would thus include not only the standard estimates $\takerRV$ and ground truth values $\realRV$, but also the market estimates $\makerRV$, and could then look like

\begin{equation}
    err^*(\realRV,\makerRV,\takerRV) = Bias[\takerRV|\realRV] + \var{\takerRV|\realRV} + \cov{\takerRV,\makerRV|\realRV}
\end{equation}

Recall the purpose of the $MSE$ measure (Equation~\ref{eq:mse}), which is to minimize the squared distance from the real value as

\begin{equation*}
    MSE_{{\Omega}}(R,\takerRV) = \EX [(\takerRV - \realRV)^2] = \frac{1}{|{\Omega}|} \sum_{\omega_i \in {\Omega}} (\takerVal - \realVal)^2
\end{equation*}
which can be thought of as trying to jointly minimize the bias and variance of the estimator (Section~\ref{sec:estimators}). A straightforward approach to reflect the outlined desiderata is to modify the existing $MSE$ measure with an extra term to also penalize the covariance between $\makerRV$ and $\takerRV$ as

\begin{equation}
\label{eq:loss_fcn}
    {MSE}^*_{{\Omega}}(\realRV,\makerRV,\takerRV) = \frac{1}{|{\Omega}|} \sum_{\omega_i \in {\Omega}} (\takerVal - \realVal)^2 + \gamma \cdot (\takerVal - \realVal) (\makerVal - \realVal)
\end{equation}
where the extra ``decorrelation term'' is weighted by a tunable hyperparameter $\gamma > 0$. The purpose of $\gamma$ is then to adjust the trade-off between the decorrelation term and the standard $MSE$, since the two normally represent opposing objectives, as the market maker tends to be very close to the real value. 

We note that altering the $MSE$ might seem counter-intuitive from the perspective of a perfect estimator, corresponding to the minimal $MSE$, where the additional term will only hurt its performance by pushing it away, not only from the market price, but from the true value, too, since $argmin[MSE] \neq argmin[MSE^*]$. However, recall that the decorrelation only makes sense when one is not able to train such a superior model (Section~\ref{sec:superior_model}). In those common cases, the model generally occupies some wider area of the error landscape around the $MSE$ minima, where the additional term is meant to navigate away from the correlated regions (Figure~\ref{fig:projection}), which may be equally distanced from the $MSE$ minimum, but will yield inferior profits. 

While we do not attempt to argue about optimality of the proposed $MSE^*$ metric, and we acknowledge that there are likely more appropriate measures to maximize the model profitability, we will demonstrate that this simple $MSE$-extension already works well enough in practice to proof viability of the decorrelation concept.

%





\section{Application in Sports Betting}
\label{sec:experiments}


To demonstrate viability of the concept in a concrete, practical setting, we chose the domain of sports betting. Sports betting closely follows the outlined setup of the market taker $\taker$ vs. market maker $\maker$ stochastic pricing game over opportunities $\Omega$ based in some tradeable assets $\mathcal{A}$. The commonalities\footnote{We note that for the correspondence with the stock market (Section~\ref{sec:background}), we restricted the setting to a two-way betting market.} were introduced in Table~\ref{tab:notation} and throughout the paper.



Nevertheless, there are some interesting technical differences to note. Firstly, we do not need to model the future price distribution, as the true value is directly determined by the probability of the traded (future) outcome to occur. On the other hand, the true probability is not directly observable, and has to be recovered from realizations of the corresponding stochastic outcomes, similarly to the observations of the future prices.
Secondly, the returns are determined by the odds, reflecting the inverse of the probability estimate, in a multiplicative fashion, as opposed to the stock market setting where one profits additively on the difference from the true price (Section~\ref{sec:exp_return}). Consequently, the errors in the estimates also reflect themselves differently into the estimated expected returns. Last but not least, there are lots of publicly available data from the sports betting market, which make it generally interesting as a benchmark for testing the market efficiency and the associated hypotheses on profitability of models and trading strategies. We note that much of the material discussed in this Section is adapted from our previous work~\cite{hubavcek2019exploiting}.




\subsection{Principles of Sports Betting}
\label{sec:betting}



Following the established trading setting (Section~\ref{sec:background}), sports betting can be understood as a simple 2-player stochastic game where the bettor $\taker$ repeatedly allocates wagers $\bm{f}$ from her current bankroll $\sf{W}_\tau \in \mathbb{R}_+$ over available opportunities $\omega_i^\zeta \in \Omega$, corresponding to the stochastic outcomes $\zeta \in \mathrm{Z}$ of sport matches $a \in \mathcal{A}$. In case of the established two-way betting markets (Section~\ref{sec:background}) there are two possible match outcomes $\mathrm{Z} = \{\alpha,\beta\}$. The outcomes $\zeta$ can then be understood as coming from a Bernoulli distribution parameterized by the corresponding true probability value $\realVal$. Each of the possible opportunities $\omega_i^\zeta$ is then also associated with some probability $\makerVal \in [0,1]$ by the bookmaker $\makerRV: \omega_i \mapsto \makerVal$, which is presented indirectly through the offered market \textit{odds} $o_i\approx\frac{1}{\makerVal\pm\epsilon}$, which include some margin $\epsilon$ (Section~\ref{sec:maker_adv}). Should the outcome $\zeta$ associated with a particular opportunity $\omega_i^\zeta$ be realized, a payoff ${f_i^\zeta}\cdot{o_i^\zeta}$ from the associated odds and the waged fraction is to be received by the bettor $\taker$. In the opposite case, the bettor loses the allocated part $f_i^\zeta$ of her bankroll $\sf{W}_\tau$ to the bookmaker.

We note that in the proposed application setting, we further assume betting only on opportunities with positive expected returns (Section~\ref{sec:essence}). Consequently, only one of the exclusive binary outcomes of each game can be waged at.
A payoff from each of the particular {betting} opportunities $\omega_i^\zeta$ is {thus} binary
in nature, and the potential net profit $w_i$ from allocation $f_i$ on the opportunity $\omega_i^\zeta$ is thus
\begin{equation}
\label{eq:bet_exp_profit}
    w_i =
\left\{
	\begin{array}{lll}
		{f_i^\zeta}\cdot{o_i^\zeta} - 1 ~~& \mbox{with prob. $\realVal$}  &\mbox{(if ${Z}=\zeta$ is realized)} \\
        -{f_i^\zeta}   ~~& \mbox{with prob. $1-\realVal$}  &\mbox{(if ${Z}\neq\zeta$)}
	\end{array}
\right.
\end{equation}

\subsection{Data}
\label{sec:data}

For the experiments we chose the game of basketball as one of the world's most popular sports. Conveniently, there is no draw in basketball matches, making the corresponding match outcomes ($|\mathrm{Z}|=2)$ betting market directly aligned with the two-way (home-away) market setting introduced so far. Particularly, we chose the National Basketball Association (NBA) league, which has one of the largest and most fluid betting markets, available through virtually any bookmaker. Consequently, the mainstream NBA market can be considered highly efficient and hard to beat~\cite{paul2004efficient,hubavcek2019exploiting}.


There are many data sources with a much detailed granularity, such as play-by-play transcripts or video records. Theoretically, these could enable to potentially extract information superior to that of the bookmaker (Section~\ref{sec:market_modelling}). For a proper demonstration of the decorrelation concept (Definition~\ref{def:decorrelation}), we resort to clearly inferior data sources. For that, we suffice with basic player statistics, such as numbers of shots, passes, steals, etc. Conveniently, there are large historical datasets with comprehensive sets of such statistics publicly available for each NBA match.
Particularly, we retrieved the official box score NBA data from the seasons 2000 to 2014. 

\paragraph{Features}
The information we use for predicting the outcome of a match combines historical data relating to the players of the home and away teams, respectively. For each of the two, we aggregate the aforementioned basic quantitative measures of the players' performances from all of the preceding matches since the beginning of the season in which each respective prediction takes place. Additionally, we also include seasonal aggregates from the preceding seasons.


Besides the historical track data considered above, the bookmaker's odds $o_i$ assigned to a match represent another piece of information, obviously highly relevant to the prediction of its outcome $\zeta$. While the odds $o_i$ clearly present a very informative feature, their incorporation in a model naturally increases the undesirable correlation with the bookmaker (Section~\ref{sec:proof}). Whether to omit or include the odds as a feature thus remains an interesting trade-off question, and so we further consider both the options for experimental evaluation.

\subsection{Odds}

\begin{figure}[t]
\centering
\includegraphics[width=.9\textwidth]{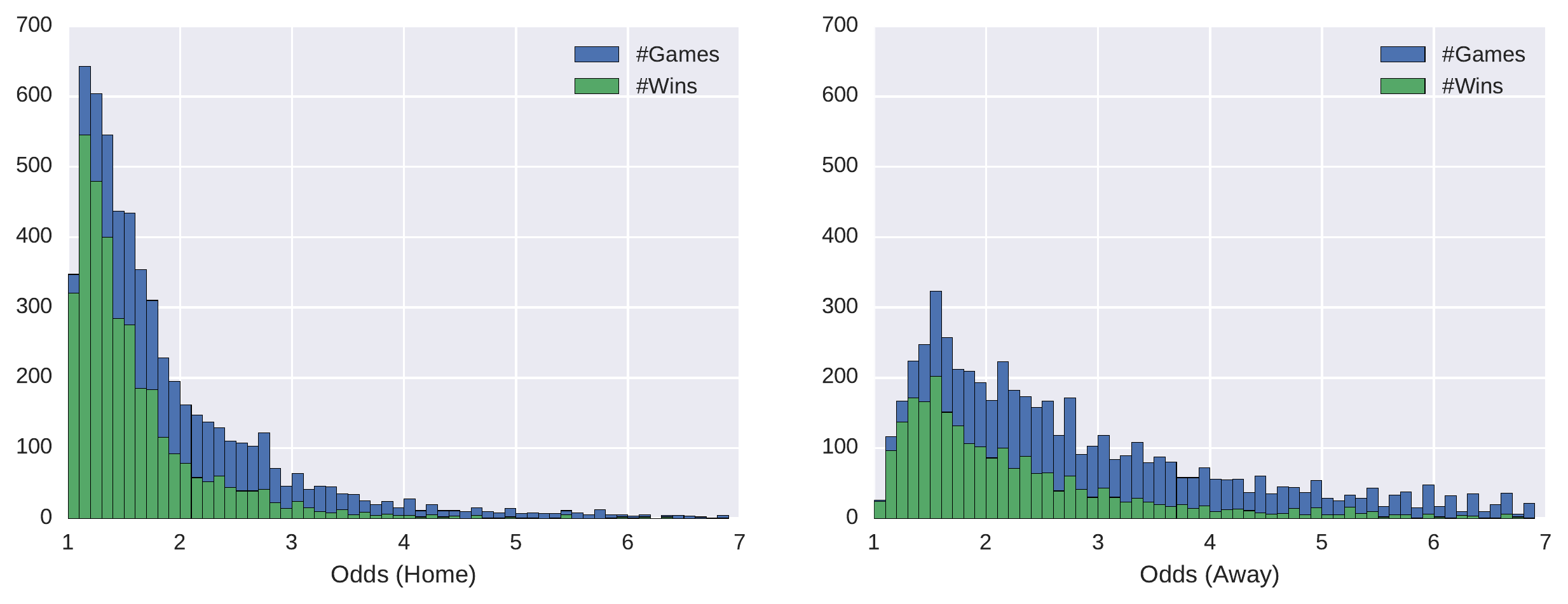}
\caption{Raw odds histograms, displaying the progressive decrease in sample size with the increasing odds.}
\label{fig:odds}
\end{figure}

For the betting odds, we used the \textit{Pinnacle}\footnote{\url{https://www.pinnacle.com/}} closing odds for seasons $2010-2014$. No that, as opposed to opening odds, the used closing odds are already adjusted by the market, reflecting also new information that came from the traders, making it more difficult to beat~\cite{}, as it tends to be more robust against systematic errors. For the earlier seasons ($2000-2010$), we unfortunately had to collect odds data from multiple different bookmakers. Note that this also makes decorrelating against the market more difficult, as we do not aim to exploit errors of a particular market maker but an aggregate of many. Moreover, these bookmakers have generally higher margins of app. $\epsilon=4.5\%$ as compared to the Pinnacle's $\epsilon=2.5\%$, making it more difficult to achieve returns that are positive in total.

Figure~\ref{fig:odds} shows histograms of the odds distribution for the home and away teams and their winnings, respectively. Since the odds reflect the inverse of the estimated probabilities, both histograms exhibit the corresponding long-tail distributions. We can casually see that the proportion of the winning games progressively decreases with the increasing odds, and that the base winning probability of the home team is generally higher\footnote{commonly referred to as the home team advantage.}.




\subsubsection{Searching for the Bias}
\label{sec:search_bias}


For the decorrelation concept (Definition~\ref{def:decorrelation}) and betting in general, we were particularly interested in a potential statistical bias of the bookmaker. As argued in Section~\ref{sec:est_dist}, should there be a statistical bias in the bookmaker's odds (Section~\ref{sec:estimators}), it would allow to easily exploit the market in certain ranges where the odds $o$ are systematically overestimated. For that we can simply take the histograms from Figure~\ref{fig:odds}, and align the probability distribution induced by the inverted odds with the relative frequencies of the actual outcomes, the result of which is displayed in Figure~~\ref{fig:odss2prob}. We can observe that for the lower odds near $1.0$ the bookmaker seems to be very accurate, while the accuracy starts to fluctuate as the odds increase. While this might incite bettors to interpret as bias in the odds, as speculated in some previous works, we argue that this is simply an artefact due to the decreasing sample size with the increasing odds (Figure~\ref{fig:odds}). Consequently, the odds would still be centered properly should more data be available.

To investigate further, we smoothed out the empirical data by fitting a corresponding Beta distribution, which removes the artefacts and reveals that the bookmakers probability distribution follows the true distribution very tightly\footnote{While there seems to be slightly more value on the home side of the market, which can be possibly attributed to a phenomenon known as the long-shot bias, there is no bias in the odds significant enough to be traded for profit.}.
An alternative way to reveal a systematic bias in the bookmaker's odds is to train a statistical model to predict the outcomes based on the odds. Should there be a systematic error in the odds, such a model should be able to exploit it directly by providing more accurate predictions than the bookmaker, given enough data. However, our previous experiments~\cite{hubacek2017thesis} verified that, given a large enough out-of-sample measurement, obtaining such a model does not seem in the realm of possibility. From all our measurements, we can thus conclude that the bookmaker is \textit{on average} precise, i.e. $\forall r \in [0,1] : \EX_{\mathrm{P}(\makerRV|r)}(\makerRV)=r$,
confirming the assumption of unbiasedness of the market maker (Section~\ref{sec:est_dist}), and supporting the use of the decorrelation technique (Section~\ref{sec:proof}), aiming to exploit the remaining (hopefully non-random) variance in the odds.

\begin{figure}

\centering
\includegraphics[width=.9\textwidth]{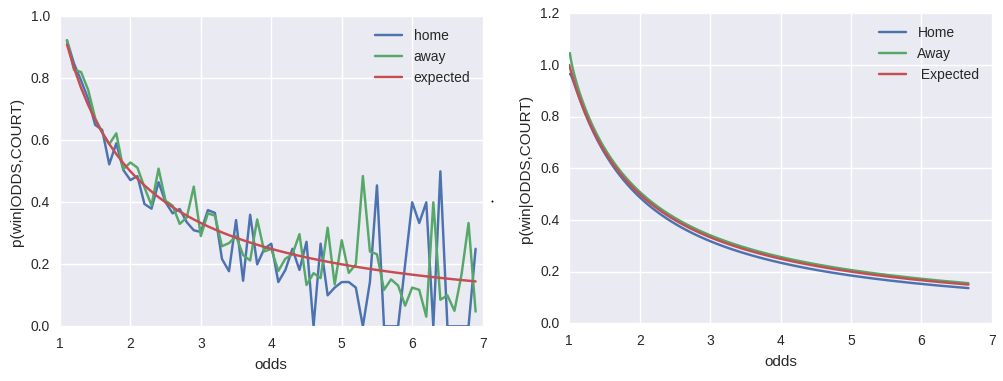}
\caption{Raw projections of the data to the probability space may display artifacts with the increasing odds and the respective decrease in sample size (left). When smoothed out with a Beta distribution prior, the bookmakers' odds distribution follows the true distribution very closely (right).}
\label{fig:odss2prob}
\end{figure}




\subsection{Experimental Setup}
\label{sec:experiments_setup}
Here we describe the particular learning and investment strategy used, and the experimental protocol followed.

\paragraph{Learning Model}

For the predictive modelling, we used a variant of a convolutional network~\cite{fukushima1982neocognitron}, described in~\cite{hubacek2017thesis}. The input to the network are two matrices (for the home and away teams, respectively), with players in rows and all the player features in the columns. The rows are sorted by the time-in-play of the corresponding players, and only the top 10 players w.r.t. this criterion are included. The convolutional layer is then defined by a vector of 10 tunable real weights. The output of the layer is a vector where each component is the dot product of the former vector with one of the matrix columns. The vector may alternatively be viewed as a filter sliding horizontally on the the first input matrix, and then on the second.
Intuitively, the convolutional network may be perceived as a learnable aggregation function extracting all the player-related variables into some latent ``team-level'' feature space, which the subsequent fully-connected layers transform into the final home-away binary prediction.

\paragraph{Error Function}
In sports betting, one cannot observe the true probability of an opportunity, even retrospectively. Instead, the label $\realVal\in\{\alpha,\beta\}$ provided with each learning sample $\omega_i$ reflects merely the binary outcome realization $\zeta$ endowed by the underlying probability. This does not pose a problem to standard classification error measures (even $MSE$), as a properly regularized model will naturally regress the underlying probability while trying to predict the discrete outcome labels (encoded e.g. as $\{0,1\}$) over a large enough data sample. However, it does pose a problem to the extra decorrelation term (Equation~\ref{eq:loss_fcn}), which would degenerate into a mere multiplication of the estimates (and their complements). To accommodate the proposed $MSE^*$ loss into this particular setting, lacking measurement of the true probability, we thus alter it slightly to
\begin{equation*}
\label{eq:new_loss}
    {MSE}^*_{{\Omega}}(\realRV,\makerRV,\takerRV) = \EX [(\takerRV - \realRV)^2 - \gamma \cdot (\takerRV - \makerRV)^2] = \frac{1}{|{\Omega}|} \sum_{\omega_i \in {\Omega}} (\takerVal - \realVal)^2 - \gamma \cdot (\takerVal - \makerVal)^2
\end{equation*}
where the $\makerVal$ are the bookmaker's probability estimates induced from the published odds $o_i$. We thus simply try to minimize the partial covariance of the model by penalizing estimates that are too close to the market price. We note that this is different from the direct penalization of the covariance of the residuals in Equation~\ref{eq:loss_fcn}. Nevertheless, the motivation is that a similar effect should be achieved through the integration with the remaining $MSE$, pushing the estimates $\takerVal$ to be unbiased w.r.t. $\realVal$.

\paragraph{Betting Strategy}

Firstly, we evaluated the models by betting with the basic uniform investment strategy (Section~\ref{sec:strat:unit}), for which we derived the theoretical reasoning in Section~\ref{sec:proof}. We further refer to this simple strategy as~\un.
Additionally, we used the classic modern portfolio theory approach to investigate compatibility of the decorrelation concept with standard investment practices. Here, we needed co calculate the expected returns and their covariances (Section~\ref{sec:MPT}). Following the probabilistic setting detailed in~Equation~\ref{eq:bet_exp_profit}, the expected profit can be defined as
\begin{equation}
    \EX_{R}[w_i] = \realVal \cdot (\frac{f_i}{\makerVal} - 1) + (1-\realVal) \cdot (-f_i) = (\frac{\realVal}{\makerVal}-1)\cdot f_i
\end{equation}
Assuming further that there is no correlation between the assets, corresponding to the underlying sport matches, it is sufficient to consider only the profit variances of the individual independent opportunities\footnote{Recall we can only bet on one of the exclusive game outcomes.}, instead of the whole covariance matrix $\Sigma$. Following from the underlying Bernoulli distribution, the profit variance of each opportunity can then be defined as
\begin{equation}\label{eq:varPi}
	\Var_{R}[w_i] = \EX[w_i^2]-\EX[w_i]^2 = (1-\realVal) \realVal f_i^2 \frac{1}{\makerVal^2}
\end{equation} 


Naturally, we used the model estimates $\takerVal$ instead of the true (unknown) values $\realVal$ in the actual calculations (Section~\ref{sec:estimators}). We then chose the particular portfolio $\bm{f}$ following the Sharpe ratio criterion (Section~\ref{sec:MPT}), and used the algorithm of sequential quadratic programming \cite{nocedal2006sequential} to identify its unique maximizer. We further refer to this strategy as \opt.

\paragraph{Evaluation Protocol}

Training and evaluation of the models and betting strategies followed the natural chronological order of the data w.r.t individual seasons, i.e. only past seasons were ever used for training a model evaluated on the upcoming season. To ensure sufficient training data, the first season to be evaluated was $2006$, with a training window made up of seasons $2000$--$2005$, progressively expanding all the way to evaluation on $2014$, trained on the whole preceding range of $2000$--$2013$.
There are in total 9093 games in the evaluated seasons $2006$–-$2014$ which count towards the reported statistics.

\subsection{Results}
\label{sec:results}

The main purpose of the experiments was to test the proposed decorrelation concept (Section~\ref{sec:decorrelating}), as materialized through the augmented loss function (Section~\ref{sec:machine_learning}) with machine learning from real data (Section~\ref{sec:data}). This is to analyse the relationships between the model accuracy, correlation and the resulting profits. Particularly, we are interested in the effects of the hyperparameter $\gamma$ in the loss function (Equation~\ref{eq:loss_fcn}), controlling the trade-off between the accuracy and decorrelation. A similar trade-off effect is also to be evaluated w.r.t. inclusion of the market odds as an input feature to the predictive model (Section~\ref{sec:data}). Lastly, we are interested in the interconnection of the decorrelated models with the investment strategies, and their relative performances. The results are displayed in Table~\ref{tab:results}.


Firstly, as expected, we can observe from the results that models utilizing the highly informative odds feature achieve consistently higher accuracies. The accuracy (\% of correctly predicted outcomes) of the bookmakers, calculated by predicting the team with the smaller odds to win, leveled over the seasons at $69\pm2.5$ and was generally superior to the neural models.
Similarly, the models that included the bookmakers' odds were anticipated to be more correlated with the bookmaker. This is convincingly confirmed by measurements of Pearson coefficients which stay at $0.87$ for the models trained without odds as a feature, and $0.95$ for the models which included them.

Finally, the results provide important insights into the relationship between decorrelation and profitability. In the analytical part of this paper, we have provided a theoretical support for lowering the correlation in order to improve model profitability (Section~\ref{sec:proof}) and proposed a simple idea how to train such models from data (Section~\ref{sec:machine_learning}). However, it was not yet clear that such model properties can be actually achieved in practice against real market makers. From the results in Table~\ref{tab:results}, it is now clear that the proposed decorrelation loss (Equation~\ref{eq:new_loss}), providing a flexible trade-off between accuracy and decorrelation, can be used as an effective proxy to achieve the desired decorrelation effect against a (highly-efficient) sports betting market. We can clearly see that augmenting the standard $MSE$ with the decorrelation term generally helps both the basic \un~as well as the standard \opt~strategies to generate more profits, while there is a sweet spot in tuning the trade-off parameter $\gamma$, generating the maximal returns w.r.t. each of the strategies.

Moreover, the benefits of the standard investment strategies, such as the MPT (Section~\ref{sec:strategies}), remain effective with the technique, as argued in Section~\ref{sec:kelly_problem_frac}, while the \opt~strategy consistently dominates performance of the simpler~\un. Last but not least, with the \opt~strategy, we were consistently able to generate statistically significant positive returns, despite the clearly inferior accuracy of the predictive model.

\begin{table}
\caption{Averages and standard errors of profits (from 10 runs over seasons $2006$--$2014$) for the two strategies (\opt, \un) with accuracies of the prediction models (Section \ref{sec:experiments_setup}) across different levels of decorrelation~(Section \ref{sec:machine_learning}).}
\centering
\begin{tabular}{@{}r|cccccc|@{}}
\cmidrule(l){2-7}
\multicolumn{1}{l|}{}     & \multicolumn{3}{c|}{without odds}                      & \multicolumn{3}{c|}{with odds}                         \\ \midrule
\multicolumn{1}{|r|}{$\gamma$} &
  \multicolumn{1}{c|}{$  \sf{W}_{\opt}$} &
  \multicolumn{1}{c|}{$ \sf{W}_{\un}$} &
  \multicolumn{1}{c|}{Accuracy} &
  \multicolumn{1}{c|}{$  \sf{W}_{\opt}$} &
  \multicolumn{1}{c|}{$ \sf{W}_{\un}$} &
  Accuracy \\ \midrule
\multicolumn{1}{|r|}{0.0} & 0.38 $\pm$ 0.10  & -5.12 $\pm$ 0.11 & 67.62 $\pm$ 0.03 & -0.12 $\pm$ 0.24 & -3.83 $\pm$ 0.22 & 68.80 $\pm$ 0.06 \\ \cmidrule(r){1-1}
\multicolumn{1}{|r|}{0.2} & 1.05 $\pm$ 0.12  & -3.31 $\pm$ 0.13 & 67.47 $\pm$ 0.03 & 0.72 $\pm$ 0.13  & -2.50 $\pm$ 0.14 & 68.37 $\pm$ 0.04 \\ \cmidrule(r){1-1}
\multicolumn{1}{|r|}{0.4} & 1.74 $\pm$ 0.14  & -1.73 $\pm$ 0.18 & 67.15 $\pm$ 0.10 & 1.49 $\pm$ 0.10  & -1.30 $\pm$ 0.12 & 67.48 $\pm$ 0.10 \\ \cmidrule(r){1-1}
\multicolumn{1}{|r|}{0.6} & 1.32 $\pm$ 0.14  & -0.61 $\pm$ 0.28 & 66.19 $\pm$ 0.09 & 1.02 $\pm$ 0.20  & -1.15 $\pm$ 0.22 & 66.55 $\pm$ 0.10 \\ \cmidrule(r){1-1}
\multicolumn{1}{|r|}{0.8} & 1.10 $\pm$ 0.29  & -0.39 $\pm$ 0.22 & 64.93 $\pm$ 0.35 & 1.00 $\pm$ 0.35  & -0.45 $\pm$ 0.28 & 65.19 $\pm$ 0.27 \\ \cmidrule(r){1-1}
\multicolumn{1}{|r|}{1}   & -1.92 $\pm$ 0.81 & -2.59 $\pm$ 0.57 & 61.30 $\pm$ 0.48 & -1.22 $\pm$ 0.51 & -2.25 $\pm$ 0.30 & 61.77 $\pm$ 0.44 \\ \bottomrule
\end{tabular}

\label{tab:results}
\end{table}










\section{Related Work}
\label{sec:related}
\paragraph{Portfolio optimization}
The approach of splitting the trader's workflow into the two steps of predictive modeling and investment optimization has a long tradition, and has been exploited in absolute majority of works~\cite{michaud2008efficient,jorion1992portfolio,perold1984large,noon2013kelly,thorp2008kelly,fitt2009markowitz}, with some notable exceptions~\cite{granger2000economic,leitch1991economic}.
Extracting the parameter estimation out of the portfolio optimization problem then enabled the respective economic research to thrive in an isolated mathematical environment, giving rise to the frameworks of Markowitz~\cite{markowitz1952portfolio} and Kelly~\cite{kelly1956new}, and their many successors~\cite{breiman1961optimal,whitrow2007algorithms,thorp2008kelly,latane2011criteria,noon2013kelly}. While widely adopted, the optimality of the resulting portfolios is based on rather unrealistic assumptions, which has been progressively criticized by many~\cite{holton2009markowitz,samuelson1971fallacy,michaud2008efficient,peters2016evaluating,samuelson2011we,maclean2010good}. From the perspective studied in this paper, the main underlying issue is the separation from the problem of estimation of the return (price) parameters, which are simply assumed at input. The resulting issues with uncertainty in the portfolios are then typically mitigated with additional practical methods~\cite{maclean2011medium,noon2013kelly,kan2007optimal}. There are also some principled approaches to tackle the input parameter uncertainty, such as considering the portfolio optimization problem within the framework of Bayesian decision making~\cite{browne1996portfolio,balka2017kelly,chu2018modified} or distributionally robust setting~\cite{sun2018distributional,blanchet2018distributionally}. However, to our best knowledge, all of these methods are aimed at mitigating the additional (structural) risk, stemming from the uncertainty in the input parameters, rather than increasing the profits.

\paragraph{From accuracy to profit}
Since the return parameters are assumed at input, the resulting portfolio optimality is inherently based on their quality. This has been previously nicely demonstrated within the mathematically elegant framework of Kelly~\cite{kelly1956new}, resulting into the idea of information advantage~\cite{cover2012elements}. Given the mathematical superiority of Kelly investments to any other strategy, this has further incited practitioners to focus merely on striving for a superior predictive model, while considering the investment optimization a solved problem.
This approach has been long established in the literature, where the Kelly strategy is often used to benchmark a model against a particular market maker~\cite{thorp2008kelly}. Subsequently, the problem of beating the market has been commonly perceived as even more difficult than it actually is~\cite{Song2007}.
Nevertheless, without the information advantage required to profit with Kelly, practitioners often fell for the fractional modification~\cite{boshnakov2017bivariate,baker2013forecasting}, yet without any further analysis of the implications on the relative performance of their models w.r.t. the market maker.
Given the rising interest in the Kelly strategy among traders~\cite{kim2017comparison}, the idea seems to gain further momentum, and it is not uncommon to see even researchers confuse profitability with accuracy~\cite{koopman2015dynamic,lessmann2010alternative,mchale2011bradley}.

\paragraph{Beating the betting market}
The problem of profiting on betting markets using statistical modelling has continuously attracted many researchers~\cite{trippi1992trading, Levitt2004a}. As the trading was historically dominated by humans relying on their expert knowledge, this domain became a ground for another human versus machine battle. The idea that experts bring some additional value was questioned early by comparing magazine experts against a statistical model~\cite{Simmons2000}. However, later it has been pointed out that these experts were often pressured to diversify their predictions, and showed that bookmakers under financial pressure can compete with the models~\cite{Forrest2005}.
As the betting industry grew, the focus shifted to comparing the models against the market makers. The market makers demonstrated the capability to outperform both experts and models~\cite{Song2007}. In a study across multiple sports, betting lines proved to consistently outperform the statistical models~\cite{Stekler2010}. Later, prediction markets, where multiple market makers compete for their positions, proved to be even more accurate than the bookmakers~\cite{Spann2009, franck2010prediction}. One of the common caveats when evaluating the statistical models is the usage of datasets not large enough to draw statistically significant conclusions~\cite{Haghighat2013}. While there appear to be inefficiencies in the markets~\cite{Paul2010, angelini2018efficiency}, works showing positive returns on larger datasets remain scarce~\cite{hubavcek2019exploiting}. Selecting a particular bookmaker can then have a large effect on the profitability, as their forecasting capabilities vary considerably~\cite{Strumbelj2014, angelini2018efficiency}. 

\paragraph{Profiting with inferior models}
There have been a number of works experimentally demonstrating positive profits with a, presumably or verifiably, inferior predictive model~\cite{peeters2018testing,baker2013forecasting,constantinou2012pi,koopman2015dynamic}. Nevertheless, as the main focus was again on the predictive modelling part, the phenomenon itself was not further analyzed in either of these works. Consequently, it was not clarified whether the idea is sound and the profits are not simply stemming from inappropriate measurements, such as the common problem of a small sample size~\cite{forrest2008sentiment,mchale2011bradley}.
To our best knowledge, the explicit idea of exploiting a betting market with a less accurate model was first openly tackled in~\cite{hubacek2017thesis}, the content of which was later presented in~\cite{hubavcek2019exploiting}. This is also the work upon which we directly follow in this paper with a more in-depth analysis of the decorrelation concept.
A very closely related is a recent work of~\cite{wunderlich2020betting}, where the authors also conclude that the relationship between model accuracy and betting returns has not yet been sufficiently studied, and clearly demonstrate that it is possible to generate positive returns in absence of superior predictive accuracy. They further analyse the underlying reasons through illustrative examples and a simulation, similarly to~\cite{hubavcek2019exploiting}\footnote{As the authors themselves note, the simulation in~\cite{wunderlich2020betting} differs from~\cite{hubavcek2019exploiting} in that they directly assume correlation between prediction \textit{errors}, as opposed to the predictions themselves. While this reflects the underlying principle more directly, it is also considerably unrealistic, since the error is not observable in practice, as well as the true probability required for the reported MSE calculation. Consequently, ad-hoc assumptions about the market distribution had to be made in~\cite{wunderlich2020betting}, as opposed to the beta distributions (i.e. conjugate priors to the observed Bernoulli trials) fitted directly to the market data in~\cite{hubavcek2019exploiting}, and the associated practical machine learning objective. We also note that, despite decorrelating the errors and decorrelating the estimates are different in principle, both approaches lead to the same statistical effect on profitability given the equal means ($0.5$) of all the estimators~\cite{hubavcek2019exploiting}.}. While their approach is not directly transferable into a practical setting, they do provide a more detailed study of the bookmaker's margin, n-way markets, and the problem of a small sample size~\cite{Haghighat2013}.
W.r.t.~\cite{wunderlich2020betting} and~\cite{hubavcek2019exploiting}, in this paper we aimed to provide a more thorough analysis of the relationships between the underlying model properties, market distributions, and investment strategies\footnote{including the Kelly strategy the study of which in this context was proposed to be done in~\cite{wunderlich2020betting}.}.








\section{Conclusion}
\label{sec:conclusion}

We have introduced a statistical setting in which an inferior predictive model can be used to generate profits in a (partially) inefficient market. Importantly, we did not restrict the properties of the market (maker) any further, and demonstrated the concepts across two different domains of stock and betting markets. After detailing an extensive background on the topic, we provided some important insights into the problem of profitability w.r.t. the properties of the predictive models to show that the relationship goes well beyond the basic measures of predictive accuracy. We thoroughly analysed the relationship from the perspective of the joint distribution of the price estimators, and proved that decreasing the partial correlation of the market taker with the market maker directly benefits the trader's profitability for most common market distributions and investment strategies.

Additionally, we translated the analysis into a practical machine learning setting, and demonstrated viability of the ``decorrelation'' concept in practice. Particularly, we have shown on real world data from a highly efficient betting market that incorporating the decorrelation objective consistently improves the profits across diverse settings, and were even able to achieve positive total returns in combination with modern portfolio theory.

\subsection{Future Work}

There are lots of avenues to improve the presented work. For instance, the chosen approach of statistical explanation of the relationships between the estimators through the standard measure of correlation (covariance) follows mostly from convenience of description rather than appropriateness for the particular problem in scope. While the decorrelation concept coincides nicely with profitability in some particular settings assumed (Section~\ref{sec:decorrelating}), it would be generally more appropriate to base the analysis directly on an integral over (the discrete regions of) the vector field of profitability (Figure~\ref{fig:projection}). While correlation is a convenient intermediate link for explanation of the phenomenon, it is certainly not suited to convey the problem insights in full.

Similarly, it is clear that the proposed $MSE$-based loss function to be optimized is rooted merely in an intuition rather than formal mathematical derivation. Consequently, it is very likely that a more appropriate loss function can be derived. Particularly, the loss should take the context of the subsequent investment strategy and the corresponding utility function into account (such as the $XENT$ for Kelly). As a general direction, instead of regression of the true asset value, it would seem more appropriate to consider the proposed scenario as a robust classification problem w.r.t. $\makerRV$ and $\realRV$ from the start. This would render measures such as Hinge loss as a more viable alternative for adaptation (modification) into the proposed setting. Ultimately, it would be naturally desirable to integrate the whole workflow into a single end-to-end profit (utility) optimization problem, removing the need for the intermediate analysis and the surrogate loss function completely.





\subsection*{Acknowledgments}

We are very grateful to Matej Uhrin for the many interesting discussions on the Kelly criterion, sharpening our understanding in defence of the decorrelation concept.

\subsection*{Notice}
If you find content of this paper useful and/or would like to collaborate, please let us know!

\newpage

\appendix

\section{Appendix}

\subsection{Simulated Data}
To further demonstrate the proposed decorrelation concept (Definition~\ref{def:decorrelation}), we present a simulation adopted from~\cite{hubavcek2019exploiting}. Here we simulated the ground truth $\realRV$ as well as both of the estimates $\takerRV,\makerRV$ with various levels of their correlation, and measured the profits made by the \opt\ and \un\ strategies for these different levels.

More precisely, we sampled triples $\realRV,\takerRV,\makerRV$ from a multi-variate Beta distribution. The distribution is parameterized with the marginal means and variances of the three variables and their pair-wise correlations. The mean of each of the three variables was set to $1/2$, reflecting the mean probability of the complementary binary outcomes. The variance of $\makerRV$ was determined as 0.054 from real bookmaker's data (Section \ref{sec:data}), and $\takerRV$'s variance copies this value. The variance of $\realRV$ was set to $0.08$ reflecting the \emph{glass ceiling} thesis\footnote{Articulated by \cite{Zimmermann2013} and expressing that sports games are predictable with a maximum of 75\% accuracy at best. When $\realRV$ is sampled with mean $1/2$ and variance $0.08$, then with 0.75 probability the event $(\realRV>1/2)$ predicts correctly the outcome of a Bernoulli trial parameterized with $\realRV$.}.

We let the correlations $Corr(\takerRV,\realRV)$, $Corr(\takerRV,\makerRV)$ and $Corr(\realRV,\makerRV)$ range over the values $\{0.85,0.90,0.95\}$. The former two represent the independent variables of the analysis, acting as factors in Table \ref{tab:correlations}, while the presented numbers average over the 3 values of $Corr(\realRV,\makerRV)$. 
For each setting of $Corr(\takerRV,\realRV)$ and $Corr(\takerRV,\makerRV)$, we drew $\realVal,\takerVal,\makerVal$ ($i=1,2,\ldots, n = 30$) samples, to simulate one round of betting. Then we set the odds $o_i = 1/\makerVal$ (the bookmaker's margin being immaterial here) for $1 \leq i \leq n$, and determined the bets $f_1,f_2,\ldots,f_{n}$ from $o_1,o_2,\ldots,o_{n}$ and $\taker_1,\taker_2,\ldots,\taker_{n}$ using the \opt\ and \un\ strategies. Finally, the match outcomes were established by a Bernoulli trial for each of the $\real_1, \real_2, \ldots, \real_n$. With these inputs, we calculated the cumulative profit $W=w_1+w_2+\ldots+w_{n}$ of one round. This procedure was repeated 10 000 times (rounds), averaging the $\mathrm{W}_{\opt}$ and $\mathrm{W}_{\un}$\footnote{Note that this is not the same (for $\mathrm{W}_{\opt}$) as setting $n=30\cdot 10 000$ without repeating the procedure, as the full unit budget is supposed to be spent in each round.}.

\begin{table}[t]
	\centering
	\begin{tabular}{llrrrrrrr}
		\toprule
		$\corr{\takerRV,\realRV}$ & $\corr{\takerRV,\makerRV}$ & $ \varnothing \mathrm{W}_{\opt}$ & $ \varnothing \mathrm{W}_{\un}$ & Accuracy  & Consensus  & Upset  & Missed  & Spotted  \\ 
		\midrule
        0.85 & 0.85 &      11.15 &    3.14 &     70.11 &      61.99 &   20.37 &    9.53 &     8.12 \\
             & 0.90 &       6.14 &    0.52 &     70.05 &      63.60 &   22.04 &    7.91 &     6.45 \\
             & 0.95 &      -1.73 &   -5.46 &     70.08 &      65.74 &   24.12 &    5.80 &     4.34 \\
        0.90 & 0.85 &      18.14 &    8.56 &     71.48 &      62.66 &   19.70 &    8.82 &     8.83 \\
             & 0.90 &      14.05 &    5.74 &     71.48 &      64.36 &   21.35 &    7.17 &     7.12 \\
             & 0.95 &       9.60 &    3.38 &     71.45 &      66.39 &   23.50 &    5.05 &     5.06 \\
        0.95 & 0.85 &      25.30 &   13.42 &     72.91 &      63.34 &   18.98 &    8.11 &     9.57 \\
             & 0.90 &      22.95 &   12.69 &     72.93 &      65.02 &   20.62 &    6.46 &     7.91 \\
             & 0.95 &      20.79 &   11.87 &     72.92 &      67.21 &   22.74 &    4.33 &     5.71 \\       
        \bottomrule
	\end{tabular}
	\caption{
	Average profits (in \%) $W_\opt$ of the	\opt\ and 	$W_\un$ of the 	\un\ strategies w.r.t. to the correlations of the	(estimated)	probabilities. Accuracy denotes the \% of correct outcome predictions by the bettor (predict win if $\takerVal>1/2$). The four last columns break down the proportions (in \%) of different combinations of predictions by $\taker$ (bettor) and $\maker$ (bookmaker):	\emph{Consensus} (both right), \emph{Upset} (both wrong), \emph{Missed} (bettor wrong, bookmaker right), \emph{Spotted} (bettor right,	bookmaker wrong).
	}
	\label{tab:correlations}
\end{table}

Table \ref{tab:correlations} then shows the average profits as well as the accuracy of the bettor's outcome predictions (call win if $\takerRV>1/2$), and the percentual breakdown of 4 possible combinations of bettor's and bookmaker's predictions. The accuracies, as well as the four latter proportions, are also averaged over all bets in all simulated rounds.

Besides the unsurprising observation that bettor's prediction accuracy grows with $Corr(\takerRV,\realRV)$, the results show that profits indeed decay systematically as the bettor's and bookmaker's predictions become more correlated (increasing $Corr(\takerRV,\makerRV)$ decreases profit). An instructive observation is that the proportion of \emph{spotted} opportunities is in all cases higher when the bookmaker's and bettor's predictions are less correlated. Moreover, we can see that the betting strategy is another independent factor strongly influencing the profit, with the \opt\ strategy being superior to the \un\ strategy.

Following the projection of the space of the $(\realRV,\makerRV,\takerRV)$ estimates onto profitability displayed Figure~\ref{fig:projection}, and analysed in Section~\ref{sec:taker_adv}, we also provide a direct 2D projection of a similar simulated betting scenario in Figure~\ref{fig:betting}. We can again observe that decreasing the correlation is directly reflected in increase of profits.

\begin{figure}

\centering
\resizebox{0.99\textwidth}{0.9\textheight}{
\includegraphics{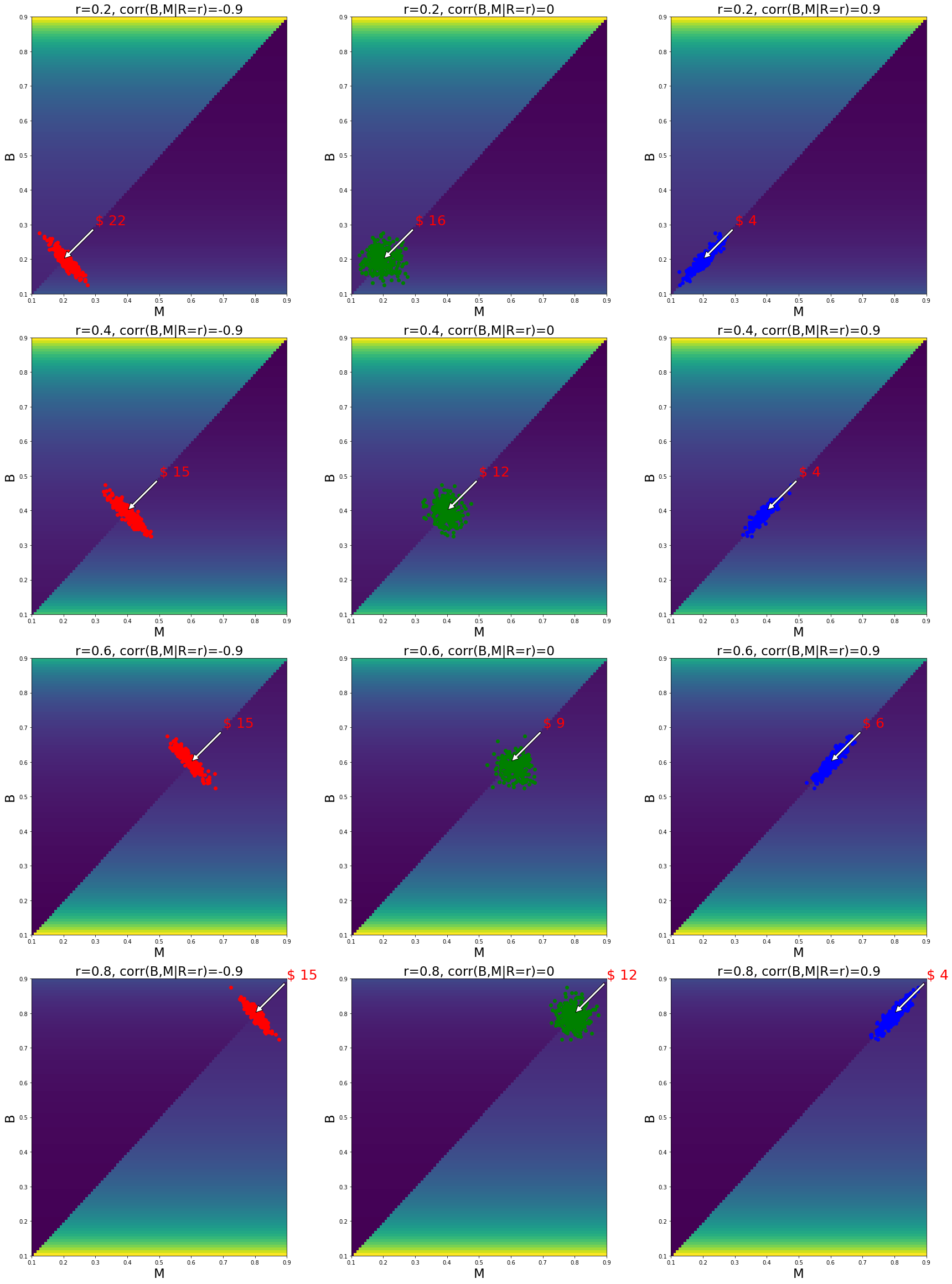}
}
\caption{Particular examples of partially correlated, independent and negatively correlated betting w.r.t. varying levels of the true probability $\real$. Measured on a random sample from a multivariate distribution with means in the respective $(\real,\real,\real)$ and equal variances of both the estimators $\var{\takerRV|\realRV}=\var{\makerRV|\realRV}$.}
\label{fig:betting}
\end{figure}

\bibliography{references}
\bibliographystyle{plain}

\end{document}